\documentclass{scrartcl}

\usepackage[greekit,amsthm]{./pamas} 

\usepackage{xparse}
\usepackage{xifthen}%

\usepackage{tabularx,booktabs,amssymb}
\usepackage{amsmath}

\newlength{\wordlength}
\newcommand{\wordbox}[3][c]{\settowidth{\wordlength}{#3}\makebox[\wordlength][#1]{#2}}
\newcommand{\mathwordbox}[3][c]{\settowidth{\wordlength}{$#3$}\makebox[\wordlength][#1]{$#2$}}
\newcommand{\padsuperscript}[0]{\mathwordbox{}{^0}}
\newcommand{\pss}{\padsuperscript}

\usetikzlibrary{calc,shapes,positioning}

\usepackage[ruled]{algorithm2e}
\usepackage{nicefrac}

\usepackage{mathtools}

\usepackage[]{subfig}
\newcommand{\subfigureCMD}[3]{\ifthenelse{\equal{#2}{}}{{\subfloat[#1]{#3}}}{{\subfloat[#1\label{#2}]{#3}}}}

\newcommand{\bothwayedges}[1]{\overleftrightarrow{#1}}
\newcommand{\set}[1]{\{#1\}} 
\newcommand{\weight}{{\mathrm w}}
\newcommand{\bigo}{O}
\newcommand{\none}{\text{---}}

\newcommand{\ba}[1][]{\ifthenelse{\equal{#1}{}}{{\mathit{BA}}}{{\mathit{BA}(#1)}}}
\newcommand{\sla}[1][]{\ifthenelse{\equal{#1}{}}{{\mathit{SL}}}{{\mathit{SL}(#1)}}}
\newcommand{\teq}[1][]{\ifthenelse{\equal{#1}{}}{\mathit{TEQ}}{\mathit{TEQ}(#1)}}
\newcommand{\rp}[1][]{\ifthenelse{\equal{#1}{}}{\mathit{RP}}{\mathit{RP}(#1)}}

\makeatletter
\let\@@pmod\pmod
\DeclareRobustCommand{\pmod}{\@ifstar\@pmods\@@pmod}
\def\@pmods#1{\mkern4mu({\operator@font mod}\mkern 6mu#1)}
\makeatother

\newcommand{\defineDecisionProblem}[4][]{%
	\medskip
	\noindent\textsc{#2}~#1\\
	\noindent%
	\begin{tabularx}{\textwidth}{@{}lX@{}}
		\noindent\emph{Input:} & #3 \\
		\noindent\emph{Question:} & #4
	\end{tabularx}
	}

\newcommand{\incompart}[1]{\tilde#1} 
\newcommand{\props}{P} 
\newcommand{\naivecheckthree}{{\textsc{$2$-Partition-Check-$3$-Majority}}\xspace}
\newcommand{\satcheckk}[1][k]{\textsc{\sat-Check-$#1$-Majority}\xspace}
\newcommand{\checkkmajorityproblem}{\textsc{Check-$k$-Majority}\xspace}
\newcommand{\checktwomajorityproblem}{\textsc{Check-$2$-Majority}\xspace}
\newcommand{\checkthreemajorityproblem}{\textsc{Check-$3$-Majority}\xspace}
\newcommand{\majdimname}{majority dimension\xspace}
\newcommand{\majdim}{\ensuremath{\dim}}
\newcommand{\kmaj}{\ensuremath{k_\text{maj}}}
\newcommand{\majexprname}{majoritarian expressiveness\xspace}
\newcommand{\majexpr}{\ensuremath{n^\mathcal{T}}}

\newcommand{\preflib}{\textsc{PrefLib}\xspace}

\newcommand{\cnf}{CNF\xspace}
\newcommand{\threecnf}{3-\cnf}
\newcommand{\fewcnf}{\textsc{Few}-\cnf}
\newcommand{\redfewcnf}{\textsc{ReducedFew}-\cnf}
\newcommand{\orderedthreecnf}{\textsc{Ordered3}-\cnf}

\newcommand{\sat}{\textsc{Sat}\xspace}
\newcommand{\threesat}{3\sat}
\newcommand{\fewsat}{\textsc{Few}\sat} 
\newcommand{\redfewsat}{\textsc{ReducedFew}\sat}
\newcommand{\orderedthreesat}{\textsc{Ordered3}\sat}
\newcommand{\maxsat}{\textsc{Max}\sat}

\newcommand{\redform}{\ensuremath{red}} 

\newcommand{\NP}{NP} 

\newcommand{\midd}{\colon}

\newcommand{\T}{\mathcal{T}}
\newcommand{\G}{\mathcal{G}}
\newcommand{\Tn}{\mathcal{T}_n}
\newcommand{\Gn}{\mathcal{G}_n}

\def\nd{1.5em} 	

\tikzset{%
node distance = \nd,
}

\tikzset{ 
verysmallvertex/.style={
  draw,
  circle,
  inner sep=0,
  minimum size=.2cm}
}

\tikzset{ 
smallvertex/.style={
  draw,
  circle,
  inner sep=0,
  minimum size=.4cm}
}

\tikzset{
vertex/.style={
	draw,
	circle,
	inner sep=0pt,
	minimum size=1.8em}
}

\tikzset{
edge/.style={
	-latex
	}	
}

\tikzset{
invedge/.style={
	latex-
}
}

\begin{document}

\markboth{Georg Bachmeier et al.}{$k$-Majority Digraphs and the Hardness of Voting with a Constant Number of Voters}

\title{$k$-Majority Digraphs and the Hardness of Voting with a Constant Number of Voters}

\author{Georg Bachmeier\textsuperscript{1}, 
Felix Brandt\textsuperscript{2}, 
Christian Geist\textsuperscript{2}, \\
Paul Harrenstein\textsuperscript{3}, 
Keyvan Kardel\textsuperscript{2}, \\
Dominik Peters\textsuperscript{3}, and
Hans Georg Seedig\textsuperscript{2}\\
\\
\textsuperscript{1}ETH Z\"urich\quad
\textsuperscript{2}TU M\"unchen\quad
\textsuperscript{3}University of Oxford}

\date{}

\maketitle

\begin{abstract}
Many hardness results in computational social choice make use of the fact that every directed graph may be induced as the pairwise majority relation of some preference profile. However, this fact requires a number of voters that is almost linear in the number of alternatives. It is therefore unclear whether these results remain intact when the number of voters is bounded, as is, for example, typically the case in search engine aggregation settings.
In this paper, we provide a systematic study of majority digraphs for a constant number of voters resulting in analytical, experimental, and complexity-theoretic insights.
First, we characterize the set of digraphs that can be induced by two and three voters, respectively, and give sufficient conditions for larger numbers of voters.
Second, we present a surprisingly efficient implementation via SAT solving for computing the minimal number of voters that is required to induce a given digraph and experimentally evaluate how many voters are required to induce the majority digraphs of real-world and generated preference profiles.
Finally, we leverage our sufficient conditions to show that the winner determination problem of various well-known voting rules 
remains hard even when there is only a small constant number of voters. In particular, we show that Kemeny's rule is hard to evaluate for 7 voters, while previous methods could only establish such a result for constant even numbers of voters.
\end{abstract}

\newpage

\section{Introduction} 
\label{sec:introduction}

	A significant part of \emph{computational social choice}, a relatively new interdisciplinary area at the intersection of social choice theory and computer science, is concerned with the computational complexity of voting problems. For most of the voting rules proposed in the social choice literature, it has  been studied how hard it is to determine winners, to identify beneficial strategic manipulations, or to influence the outcome by bribing, partitioning, adding, or deleting voters \citep[see, \eg][]{BCE+14a,Roth15a,FHHR09a}.
	In many cases, the corresponding problems turned out to be NP-hard. Depending on the nature of the problem, this can be interpreted as bad news---as in the case of winner determination---or good news---as in the case of manipulation, bribery, and control. 

	Many voting rules are based on the pairwise majority relation, which establishes a fruitful connection between voting theory and graph theory. Perhaps the most fundamental result in this context is \emph{McGarvey's theorem}, which states that \emph{every} asymmetric directed graph may be induced as the pairwise majority relation of some preference profile \citep{McGa53a}. Unsurprisingly, McGarvey's theorem is the basis of most hardness results concerning majoritarian voting rules.
	McGarvey's original construction requires $n(n-1)$ voters, where $n$ is the number of alternatives. This number has subsequently been improved by \citet{Stea59a} and \citet{ErMo64a}, who have eventually shown that the number of required voters is of order $\Theta(n / \log n)$. 
	Since the result by \citet{ErMo64a} gives a lower bound as well, it implies that, for any constant number of voters, there are majority digraphs that cannot be induced by any preference profile.\footnote{The same holds for tournaments \citep[][Ex.1 (d)]{Moon68a}.} As a consequence, the mentioned hardness results only hold if the number of voters is roughly of the same order as the number of alternatives. 
	
	In some applications, however, the number of voters is much smaller than the number of alternatives and it is unclear whether hardness still holds. A typical example is search engine aggregation, where the voters correspond to Internet search engines and the alternatives correspond to the webpages ranked by the search engines \citep[see, \eg][]{DKNS01a}.
	\citet{Hudr08a} writes that ``to my knowledge, when not trivial, the complexity for lower values of $m$ [the number of voters] remains unknown. In particular, it would be interesting to know whether some of the problems [\ldots] remain NP-hard if $m$ is a given constant''.
	For the particularly interesting case of Kemeny's \citeyear{Keme59a} rank aggregation rule, there is a notable exception: a proof by \citet{DKNS01a} shows NP-hardness of this rule without needing to appeal to McGarvey's result, and shows that this hardness holds even if there are only $4$ voters. The result also holds for every larger \emph{even} constant number of voters. In the \emph{Handbook of Computational Social Choice}, \citet{FHN15a} note that ``quite intriguingly, the case for any odd $n \ge 3$ remains open''. In a similar vein, \citet{Hudr08a} writes that
``it would be interesting to decide whether it is still the case for fixed values of $m$ with $m$ odd'', and \citet{BBD09a} find that \citeauthor{DKNS01a}'s method ``does not work for odd numbers of [voters]'' and write that particularly the case for $n = 3$ remains ``wide open''.

In this paper, we provide a systematic study of majority digraphs for a constant number of voters resulting in analytical, experimental, and complexity-theoretic insights.
Starting from a discussion of bounds on the size of the smallest tournaments that require a certain number of voters (\secref{sec:bounds}), we analyze the structure of majority digraphs inducible by a constant number of voters. Obviously, the less voters there are, the more restricted is the corresponding class of inducible majority digraphs. 
	For instance, digraphs induced by two voters have to be acyclic (and are subject to some additional restrictions). 
	
	\emph{Analytically}, we
	characterize digraphs inducible by two and three voters, respectively, and provide sufficient conditions for digraphs to be induced by $k$ voters (\secref{sec:majority_relations_of_few_voters}). 
	We propose a surprisingly efficient implementation via SAT solving for computing the minimal number of voters that is required to induce a given digraph (\secref{sec:determining_the_majority_dimension_of_a_digraph}) and \emph{experimentally} evaluate how many voters are required to induce the majority digraphs of real-world and generated preference profiles (\secref{sec:analyzing_majority_dimensions}). 
	
	In \secref{sec:hardness_of_voting_with_few_voters}, we then finally leverage the conditions from \secref{sec:majority_relations_of_few_voters} to investigate whether common, computationally intractable voting rules (the Banks set, the tournament equilibrium set, Kemeny's rule, Slater's rule, and ranked pairs) remain intractable when there is only a small constant number of voters. 
	This is achieved by analyzing existing hardness proofs and checking whether the class of majority digraphs used in these constructions can be induced by small constant numbers of voters. 
	Perhaps surprisingly, it turns out that all hardness proofs we studied can be constructed using at most 11 voters, and for many proofs, including one for Kemeny's rule, 7 voters suffice.

	The paper concludes with an overview of the achieved hardness results, summarized in \tabref{tbl:results}, and a brief outlook on future research in \secref{sec:conclusions_and_discussion}. 

	Our work can be viewed as complementary to the work by \citet{CSL07a} who considered manipulation problems with a constant number of \emph{candidates}. \citeauthor{CSL07a} determined how many candidates are required such that the problem of coalitional manipulation with weighted voters becomes NP-hard for a number of tractable voting rules including Borda's rule, Copeland's rule, and maximin.

\section{Preliminaries} \label{sec:preliminaries} 

	This section introduces the notation and terminology required to state our results.

	A \emph{directed graph} or \emph{digraph} is a pair $(V,E)$, where $V$ is finite a set of vertices and $E\subseteq V\times V$ is a set of arcs (directed edges). The \emph{size} of a digraph is its number of vertices $|V|$.
	By $\G$ we denote the class of all digraphs and by $\Gn$ the class all digraphs of size $n$.
	The \emph{converse of $E$} is $\overline E=\set{(w,v)\midd (v,w)\in E}$, where the direction of all arcs is reversed. Often it will be useful to effectively disregard orientations by considering $\bothwayedges E=E\cup \overline E$.
	We say that $E_1$ and $E_2$ are \emph{orientation compatible} if $E_1\cap(\bothwayedges{E_1}\cap\bothwayedges{E_2})=E_2\cap(\bothwayedges{E_1}\cap\bothwayedges{E_2})$, \ie if
	for all $e\in\bothwayedges{E_1}\cap\bothwayedges{E_2}$, $e\in E_1$ if and only if $e\in E_2$.

	The \emph{incomparability graph} $\incompart G=(V,\incompart E)$ associated with a digraph $(V,E)$ is defined such that for all $v,w\in V$, 
	\[
		(v,w)\in \incompart E \text{ if and only if neither } (v,w)\in E \text{ nor } (w,v)\in E.
	\]
	Obviously, $\bothwayedges{\incompart E}=\incompart E$.

	A digraph $G=(V,E)$ is said to be \textit{transitive} if for all $x,y,z \in V$, $(x,y) \in E$ and $(y,z) \in E$ imply $(x,z) \in E$. Moreover,~$G$ is \textit{acyclic} if for all $x_1,\dots,x_k \in V$, $(x_1,x_2),(x_2,x_3),...,(x_{k-1},x_k) \in E$ implies  $(x_k,x_1) \not\in E$. Also $G$ is \emph{asymmetric} if $(v,w)\in E$ implies $(w,v)\notin E$. 
	A \emph{tournament} is an asymmetric digraph $(V,E)$ where 
	$E$ is complete, \ie if for all distinct $v,w\in V$, either $(v,w)\in E$ or $(w,v)\in E$. We denote the sets of all tournaments by $\T$ and the set of those with $n$~vertices by~$\Tn$.
	Moreover, a digraph $(V,E)$ is \textit{transitively (re)orientable} if there exists
	a transitive and asymmetric digraph $(V,E')$ with $\bothwayedges{E'}=\bothwayedges{E}$. $E'$ is also referred to as a \emph{reorientation of~$E$}.

	The digraphs in this paper are assumed to be induced by the preferences of a set of voters. Let
	$N=\set{1,\dots, k}$ be a set of~$k$ voters (or an \emph{electorate} of size~$k$) and $V$ a set of alternatives.
	The preferences of each voter~$i$ are given as \emph{linear orders}, \ie transitive, complete, and antisymmetric relations $R_i$ over a set of alternatives~$V$.
	A \emph{preference profile} $R=(R_1,\dots,R_k)$ associates a preference relation with each voter. Each preference profile gives rise to a \emph{majority relation}, which holds between two alternatives~$v$ and~$w$ if the number of voters preferring~$v$ to~$w$ exceeds the number of voters preferring~$w$ to~$v$. We say that $(V,E)$ is \emph{the majority digraph} of preference profile~$R$ if
	\[
			\text{$(v,w)\in E$ if and only if $|\set{i\in N\colon v\mathrel{R_i}w}|>|\set{i\in N\colon w\mathrel{R_i}v}|$.}
	\]
	Majority digraphs are asymmetric. Moreover, if the number of voters is odd, the majority digraph is complete and thus a tournament. 

	We say that $(V,E)$ is \emph{$k$-inducible} if $(V,E)$ is the majority digraph for some preference profile involving~$k$ voters. Equivalently, we say that $(V,E)$ is a \emph{$k$-majority digraph}.\footnote{\citet{ABKK+06a} used the term \emph{$k$-majority tournament} for tournaments that are induced by a $(2k-1)$-voter profile because every majority consists of at least $k$~voters. We chose to follow the terminology of \citet{KTC+09a} instead.} 
As an example, \figref{fig:digraph} shows a tournament which is induced by a $3$-voter profile, and thus this tournament is a $3$-inducible majority digraph.

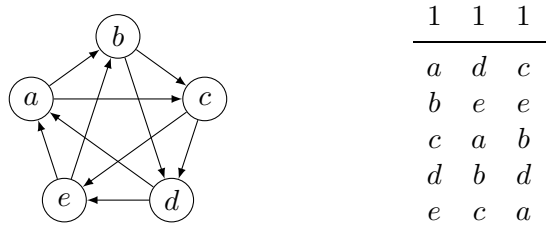
\begin{figure}[htb]
  \centering

\begin{minipage}{10em}
\begin{tikzpicture}[auto]
	\draw[use as bounding box,opacity=0] (-1.7,-1.8) rectangle (1.7,1.5);
  \node[fill=white,circle,draw,minimum size=1.5em,inner sep=0pt] (1) at (162:1.2) {$a$};
  \node[fill=white,circle,draw,minimum size=1.5em,inner sep=0pt] (2) at (90: 1.2) {$b$};
  \node[fill=white,circle,draw,minimum size=1.5em,inner sep=0pt] (3) at (18: 1.2) {$c$};
  \node[fill=white,circle,draw,minimum size=1.5em,inner sep=0pt] (4) at (306:1.2) {$d$};
  \node[fill=white,circle,draw,minimum size=1.5em,inner sep=0pt] (5) at (234:1.2) {$e$};
  \draw[-latex] (4) to (1);
  \draw[latex-] (4) to (2);
  \draw[latex-] (4) to (3);
  \draw[latex-] (1) to (5);
  \draw[latex-] (2) to (5);
  \draw[-latex] (3) to (5);
  \draw[-latex] (4) to (5);
  \draw[-latex] (1) to (2);
  \draw[-latex] (2) to (3);
  \draw[latex-] (3) to (1);
\end{tikzpicture}
		\end{minipage}
\quad\quad
\begin{minipage}{12em}
\begin{tikzpicture}[auto]
		\draw[use as bounding box,opacity=0] (-1.7,-1.8) rectangle (1.7,1.5);
	\tikzstyle{every node}=[fill=white,circle,draw,minimum size=1.5em,inner sep=0pt]
	\node[draw=white,rectangle] (1) at (0,.15) {
$
\begin{array}{c@{\quad}c@{\quad}c}
	1	&	1	&	1	\\\midrule
	a	&	d	&	c	\\
	b	&	e	&	e	\\
	c	&	a	&	b	\\
	d	&	b	&	d	\\
	e	&	c	&	a	
\end{array}			
$};
\end{tikzpicture}
\end{minipage}


  \caption{A majority digraph and a $3$-voter preference profile that induces the digraph.}
  \label{fig:digraph}
\end{figure}

	We also consider \emph{weighted digraphs~$(V,\weight)$}, where~$V$ is a set of vertices and~$\weight\colon V\times V\to \mathbb Z$ a weight function assigning weight $\weight(v,w)$ to the arc $(v,w)$.
	With a slight abuse of notation we also refer to weighted digraphs as a pair $(V,E)$, where the weight function is subsumed and it is understood that $E=\set{(v,w)\midd \weight(v,w)>0}$. We say that a weighted digraph $(V,\weight)$ is induced by~$R$ if  for all $v,w\in V$, $\weight(v,w)=|\set{i\in N\colon v\mathrel{R_i}w}|-|\set{i\in N\colon w\mathrel{R_i}v}|$. In this case, $(V,\weight)$ is a \emph{weighted $k$-majority digraph}. 

	Given a (weighted) digraph we are interested in the minimal number of voters needed such that the digraph represents the (weighted) majority relation of the voters' preferences. This is captured in the majority dimension of the digraph.\footnote{This complexity measure of digraphs can also be interpreted as a complexity measure of preference profiles. The majority dimension of a given preference profile is then simply defined as the majority dimension of the induced majority digraph.	} 
	Formally, the \emph{majority dimension} of a digraph $G=(V,E)$ or a weighted digraph $G=(V,\weight)$ is the smallest number of voters in a profile that induces~$G$, \ie
	\[
	\majdim(G) = \min \{k \midd G \text{ is a (weighted) $k$-majority digraph}\}.
	\]
	Also, let $\kmaj(n)$ denote the minimum electorate size required to induce all digraphs of size~$n$, \ie
	\[
		\kmaj(n) = \min \{k \midd \majdim(G) \leq k \text{ for all $G \in \Gn$}\}.
	\]
	If we restrict our attention to tournaments, we will write $\kmaj^\T(n)$ instead. Note that $\kmaj^\T(n)\leq\kmaj(n)$ since $\T\subset \G$.

	Conversely, define the \emph{\majexprname} of (electorates of size)~$k$ to be the maximum integer $\majexpr(k)$ such that every complete majority relation on up to~$\majexpr(k)$ alternatives is $k$-inducible. 
	Since the work by \citet{ErMo64a}, which we discuss in more detail in \secref{sec:bounds}, it is known that $\majexpr(k)$ is finite for every~$k$.
	Note that this implies, that the smallest tournament that cannot be induced by~$k$ voters is of size~$\majexpr(k)+1$.

	By a \emph{voting rule} we understand a function that maps each preference profile to a non-empty subset of the alternatives. Over the years, a large number of voting rules have been proposed. The ones we will be concerned with in this paper---because they are based on majority digraphs and computationally intractable---are the \emph{Banks set ($\ba$)}, the \emph{tournament equilibrium set ($\teq$)}, \emph{Slater's rule ($\sla$)}, and \emph{ranked pairs ($\rp$)}. The definitions of these rules are given in \secref{sec:hardness_of_voting_with_few_voters}.

\section{Bounds on the Majority Dimension and Majoritarian Expressiveness} 
\label{sec:bounds}

	This section is concerned with bounds on the \majdimname of digraphs and on the \majexprname of electorates of a fixed size.

	The first such result is due to \citet{McGa53a}, who showed that every digraph can be induced by some (finite) preference profile. He gave a construction that requires exactly two voters per arc in the digraph. In our notation, this implies that $\kmaj(n)\leq n(n-1) < \infty$ for all~$n$.
	
	The work by \citeauthor{McGa53a} has been followed up by \citet{Stea59a} who showed that $\kmaj(n)\leq n+2$ which was later improved by \citet{Fiol92a} to $\kmaj(n)\leq n-\lfloor \log n \rfloor+1$. For larger $n$, \citet{ErMo64a} gave the asymptotically better bound $\kmaj(n)\leq c\cdot \frac{n}{\log n}$ for some constant $c$. Their work nicely complemented a second result by \citet{Stea59a} who proved that $\kmaj(n)>0.55\cdot\frac{n}{\log n}$ for large~$n$. These results together asymptotically capture the growth of $\kmaj(n)$.
	
	\begin{theorem}
		[\citealp{Stea59a}, \citealp{ErMo64a}] 
		$\kmaj(n) \in \Theta(\frac{n}{\log n})$.
	\end{theorem}

	In the following, we are particularly interested in the majority dimension of tournaments. The following simple observation about the parity of $\majdim(G)$ will be useful.

	\begin{lemma}
	\label{lem:parity-of-majdim}
	The \majdimname $\majdim(G)$ is odd if $G$ is a tournament and even otherwise.
	\end{lemma}
	\begin{proof}
		Let~$G$ be a tournament and assume that $\majdim(G)=k$ was even. 
		Then there exists a preference profile~$R$ with $k$ voters that induces~$T$. 
		Since $k$ is even, the majority margin must be even for every pair of alternatives and can furthermore never be zero as $T$ is a tournament. 
		Therefore, removing any single voter from~$R$ gives a profile $R'$ with just $k-1$ voters that still induces~$T$, a contradiction.

		For incomplete digraphs, the statement follows directly from the fact that for all preference profiles $R$ with an odd number of voters $k$, the majority relation is complete and anti-symmetric (as no majority ties can occur).
	\end{proof}

	Note that the lower bound on $\kmaj(n)$ due to \citet{Stea59a} 
	shows that for every electorate of size~$k$, there exist \emph{digraphs} that are not $k$-inducible. 
	Still, the \majdimname of \emph{tournaments}, $\kmaj^\T(n)$, could be bounded by a constant. In that case, all tournaments could be inducible by some constant electorate size. The following lemma shows that this is not the case. The argument is similar to the one by \citeauthor{Stea59a}.

	\begin{lemma}
	\label{lem:majdim_lower_bound_necessary_T}
	If $\kmaj^\T(n) = k \geq 3$, then
	\begin{equation}\label{eq:majdim-necessary}
	\binom{n}{2} \cdot \ln(2) \leq 
	k\cdot \left(\ln(2) + \sum\limits_{i=2}^{n}\ln(i)\right) - \ln(k!)\text{.}
	\end{equation}
	\end{lemma}
	\begin{proof}
		If every tournament on $n$ vertices can be induced by $k$ voters, then for every $T \in \Tn$, there needs to be at least one  $k$-voter profile that induces $T$. 
		There are $n!$ possible preference orders over $n$ alternatives, and---when ignoring the identities of voters---the number of $k$-voter profiles is $\binom{n! +k -1}{k}$. Also, the number of labeled tournaments on~$n$ vertices is $2^{\binom{n}{2}}$ implying that
		\[
			2^{\binom{n}{2}} \leq \binom{n! +k -1}{k} \leq \frac{(2(n!))^k}{k!}
		\]
		where the last inequality follows from \citeauthor{Fiol92a}'s bound stated above. The result follows immediately.
	\end{proof}

	Using the lemma, we can search for an upper bound on the majoritarian expressiveness $\majexpr(k)$ for a given~$k$ by finding the minimal~$m$ such that \eqref{eq:majdim-necessary} is violated. \tabref{tab:bounds} shows some upper bounds for small $k$. For example, there exists a tournament of size~$42$ that is not $5$-inducible.\footnote{A slightly tighter analysis even gives the existence of such a tournament of size~$41$.}

	\begin{table}[th]
	\centering
	\begin{tabular}{ ccccccccccc }
	\toprule
		$k$ & $3$ & $5$ & $7$ & $9$ & $11$ & $13$ & $15$ & $17$ & $19$ & $21$\\
	\midrule 
	$\majexpr(k)$ & $18$ & $41$ & $66$ & $93$ & $122$ & $152$ & $183$ & $216$ & $249$ & $282$\\
	\bottomrule
	\end{tabular}
	\caption{Upper bounds on the size $\majexpr(k)$ of the smallest tournament that is not $k$-inducible for small odd $k$.}
	\label{tab:bounds}
	\end{table}

	It is clear, however, that these bounds are not tight. For example, the results in the table imply there has to exist a tournament of size~$19$ that is not $3$-inducible. 
	In fact, \citet{ShTo09a} proved that every tournament that contains a certain digraph of size $8$ as a subgraph is not $3$-inducible. In \secref{sub:exhaustive_analysis}, we will argue that there are no smaller tournaments with this property. An example of such a tournament is shown in \figref{fig:not-3inducible}.

	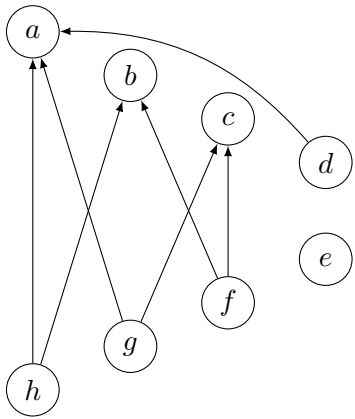
\begin{figure}
	\centering
		\begin{tikzpicture}[]
			\node[vertex] (a) {$a$};
			\node[vertex] (b) [right=of a,yshift=-\nd] {$b$};
			\node[vertex] (c) [right=of b,yshift=-\nd] {$c$};
			\node[vertex] (d) [right=of c,yshift=-\nd] {$d$};
			\node[vertex] (e) [below=of d            ] {$e$};
			\node[vertex] (f) [left=of  e,yshift=-\nd] {$f$};
			\node[vertex] (g) [left=of  f,yshift=-\nd] {$g$};
			\node[vertex] (h) [left=of  g,yshift=-\nd] {$h$};
			\draw[edge] (d) to[bend right=25] (a);
			\foreach \x / \y in {h/a, g/a, h/b, f/b, g/c, f/c}
			{
				\draw[edge] (\x) to (\y);
			}
		\end{tikzpicture}
		\caption{A tournament on~$8$ vertices with majority dimension~$5$. This is a smallest tournament that cannot be induced by three voters. Omitted arcs point downwards.}
		\label{fig:not-3inducible}
	\end{figure}
 
	Using an argument by \citet{Alon06a}, one can construct, for any odd $k$, concrete tournaments that are not $k$-inducible. 
	Unfortunately, these tournaments turn out to be very large. 
	
	\begin{example}
		Consider the \emph{quadratic residue tournaments} $Q_p$ of size $p$.\footnote{$Q_p=(V,E)$ with $V=(v_1,\ldots,v_p)$ and $(v_i, v_j)\in E$ if and only if $(i-j)^{\frac{p-1}{2}} \equiv 1 \mod p$.}
		If $p$ is a prime congruent to $3 \pmod 4$ and $p > 603,979,776$ then $Q_p$ is not $5$-inducible. 
	\end{example}
	
	The argument by \citet{Alon06a} runs along the following lines. 
	A \emph{dominating set} of a digraph $G = (V,E)$ is a set $U \subseteq V$ such
	that for all $v \in V \setminus U$, there exists a $u \in U$ with $(u,v) \in E$. 
	\citet{Alon06a} showed that the size of the smallest dominating set of any $k$-majority
	digraph for odd $k$ is bounded from above by a function $\mathcal{F}(k)$ with
	$\mathcal{F}(k) \in \bigo(k \log k)$ and $\mathcal{F}(k) \in \Omega
	(\frac{k}{\log k})$ with rather large constants hidden in the Landau
	notation ($80$ for the upper bound). This means that if a given tournament $T$
	does not have a dominating set of size $\mathcal{F}(k)$, then $T$ is not
	$k$-inducible.

	This can be leveraged to construct a tournament that is not $k$-inducible due
	to the following constructive result by \citet{GrSp71a}. Let $f(x) = p
	> x^2 2^{2x-2}$ where $x$ is a positive integer and $p$ is the smallest prime
	congruent to $3 \pmod 4$ satisfying the inequality (the construction works for
	any such $p$). Then, the quadratic residue tournament $Q_p$ of size $p$
	does not exhibit a dominating set of size~$x$.
	
	Together, this yields, for any odd $k$, a construction for a tournament on $(f \circ
	\mathcal{F})(k)$ vertices that is not $k$-inducible. 
	Unfortunately, $f(x)$ is exponential in $x$, and the value of $\mathcal{F}(k)$
	is known precisely only for $k = 3$ where we have $\mathcal{F}(3) = 3$. To the best of our knowledge, the best currently available bound for $k=5$ is $\mathcal{F}(5) \leq 12$ \citep{Fidl11a}. Together, we get that for the smallest (or any
	other) prime $p$ congruent to $3 \pmod 4$ satisfying the inequality
	\[
	p > 12^2 \cdot 2^{2\cdot 12 -2} = 603,979,776
	\]
	$Q_p$ is not $5$-inducible. 
		
	Bounds on $\mathcal{F}(k)$ for larger odd $k$ give significantly worse values: for $7$ voters, we only know that $\mathcal{F}(7)
	\leq 44$ \citep{Fidl11a}.


\section{Majority Digraphs of Few Voters} 
\label{sec:majority_relations_of_few_voters}

	In this section, we analyze the structure of $k$-inducible digraphs for constant~$k$. Building on earlier work by \citet{DuMi41a} which entails a characterization of $2$-majority digraphs, we give a characterization for the case of three voters. In addition, we present sufficient conditions for larger majority dimensions that will be leveraged in \secref{sec:hardness_of_voting_with_few_voters}.

	\subsection{Two and Three Voters}

	Given a preference profile~$R$, the \emph{Pareto relation} holds between two alternatives~$v$ and~$w$ if all voters prefer~$v$ over~$w$. \citet{DuMi41a} specified sufficient and necessary conditions for relations to be induced as the Pareto relation of a $2$-voter profile.  As for two voters the majority relation and the Pareto relation coincide, we can rephrase their result for majority digraphs as follows.

	\begin{lemma}[\citealp{DuMi41a}] \label{lem:2-voter-char}
		A majority digraph $(V,E)$ is $2$-inducible if and only if it is transitive and its incomparability graph $(V,\incompart E)$ is transitively orientable. Moreover, the weight of every arc is~$2$.
	\end{lemma}

	See \figref{fig:non2voters} for an example of a digraph that is not $2$-inducible even though it is transitive. If it was $2$-inducible, there would have to exist a transitive reorientation $E'$ of $\incompart{E}$. We can assume without loss of generality that $(b,d)\in E'$. But then $(a,d)$ and $(e,b)$ also have to be in $E'$ leaving no possibility to orient $\{a,e\}$ without obtaining a contradiction to the assumed transitivity of~$E'$.

	\begin{figure}[htb]
	\centering
	\subfigureCMD{This digraph is not $2$-inducible. Dotted arcs denote the incomparability graph.}{fig:non2voters}{
		\begin{tikzpicture}[auto,circle]
			\draw[use as bounding box,opacity=0] (-1,-2.5) rectangle (5,1.5);
			\node[draw,vertex] (1) at (0,0) {a};
			\node[draw,vertex] (2) at (1,-1) {b};
			\node[draw,vertex] (3) at (2,-2) {c};
			\node[draw,vertex] (4) at (4,0) {e};
			\node[draw,vertex] (5) at (3,-1) {d};
			\node[draw,vertex] (6) at (2,1) {f};
			\foreach \x / \y in {1/2, 2/3, 4/5, 5/3}
				\draw[edge] (\x) -- (\y);
			\draw[edge] (1) to [bend left=10] (6);
			\draw[edge] (4) to [bend right=10] (6);
			\draw[edge] (1) to [bend right=25] (3);
			\draw[edge] (4) to [bend left=25] (3);
			\foreach \x / \y in {1/4,1/5,2/4,2/5,3/6,2/6,5/6}
				\draw[dotted] (\x) -- (\y);
		\end{tikzpicture}
	}
	\hspace{2em}
	\subfigureCMD{Every forest of unidirected stars is $2$-inducible.}{fig:forest-of-stars}{
		\begin{tikzpicture}[auto,circle,draw]			]
			\node[draw,smallvertex] (1) at (1.2,2.0) {};
			\node[draw,smallvertex] (2) at (0,0) {};
			\node[draw,smallvertex] (3) at (2.3,0.5) {};
			\foreach \centerNode/\startAngle/\stepAngle/\stopAngle/\arrowDirection in {1/0/40/320/-latex, 2/10/60/310/invedge, 3/20/72/308/edge}
			{
				\pgfmathparse{int(\startAngle+\stepAngle)} 
				\foreach \angle in {\startAngle, \pgfmathresult, ..., \stopAngle}
				{
					\node[verysmallvertex] (\centerNode\angle) at ($(\centerNode) + (\angle:.7cm)$) {};
					\draw[\arrowDirection] (\centerNode) to (\centerNode\angle);
				}
			}
		\end{tikzpicture}
	}
	\caption{Examples of transitive digraphs.}
	\label{fig:2voter-examples}
	\end{figure}

	If, on the other hand, a digraph~$(V,E)$ is in fact induced by a $2$-voter profile $(R_1,R_2)$, then~$R_1$ and~$R_2$ coincide on~$E$ and are opposed on $\incompart{E}$, \ie $R_1\cap R_2=E$. As $R_1$ and $R_2$ are both transitive, so is $E$. If $E'$ is the respective reorientation of~$\incompart{E}$, then $R_1=E\cup E'$ and $R_2=E\cup\overline {E'}$, or \emph{vice versa}. 

	A digraph $(V,E)$ is an \emph{unidirected star} if there is some $v^*\in V$ such that either $E$ or $\overline E$ equals $\set{v^*}\times \left(V\setminus\set{v^*}\right)$.
	Clearly, $(V,E)$ is transitive as there are no $v,w,u\in V$ such that both $(v,w),(w,u)\in E$. Moreover, every transitive relation over the leaves of $(V,E)$ serves as a transitive orientation of~$\incompart E$. With \lemref{lem:2-voter-char} this gives us the following result, which is a special case of Lemma 1 by \citet{ErMo64a}.

	\begin{lemma}
	\label{lem:stars-2-inducible}
	Every unidirected star is $2$-inducible.
	\end{lemma}

	Another insight that follows from \lemref{lem:2-voter-char}, is that the union of pairwise disjoint digraphs that are induced by $2$-voter profiles is also induced by a $2$-voter profile.

	\begin{lemma}\label{lem:trans-disjoint-rels}
		Let $V_1,\dots,V_k$ be pairwise disjoint and $(V_1,E_1),\dots,(V_k,E_k)$ $2$-inducible majority digraphs. Then, $(V_1\cup\dots\cup V_k,E_1\cup\dots\cup E_k)$ is also $2$-inducible.
	\end{lemma}

\begin{proof}
	Let $V=V_1\cup\dots\cup V_k$ and $E=E_1\cup\dots\cup E_k$ and consider the digraph $(V,E)$.
	Since each of $(V_1,E_1),\dots,(V_k,E_k)$ is $2$-inducible, by \lemref{lem:2-voter-char}, each of $E_1,\dots,E_k$ is transitive and each of $\incompart{E}_1,\dots,\incompart E_k$ is transitively orientable. Let $E'_1,\dots,E'_k$ be the respective transitive reorientations of $\incompart{E}_1,\dots,\incompart E_k$. Since $V_1,\dots,V_k$ are pairwise disjoint, $E_1\cup\dots\cup E_2$ can readily be seen to be transitive as well. Let furthermore $E^*=\bigcup_{1\le i<j\le k}(V_i\times V_j)$. 
	Observe that  $\incompart E=\incompart E_1\cup\dots\cup\incompart{E}_k\cup \incompart{E^*}$ and that $E'_1\cup\dots\cup E'_k\cup E^*$ is a transitive reorientation of $\incompart E$. 
	The claim then follows by another application of \lemref{lem:2-voter-char}. 
	\end{proof}

	Consequently, every forest of (unidirected) stars such as the one shown in \figref{fig:forest-of-stars} is $2$-inducible.\footnote{\label{foot:bilevel}\citet{ErMo64a} gave a class of digraphs that are $2$-inducible which they call \emph{bilevel graphs}. A bilevel graph is the union of a finite number of vertex-disjoint digraphs $(V_1,E_1),(V_2,E_2), \cdots$ such that each $(V_i,E_i)$ is complete bipartite and unidirected, \ie there is a partition into vertex sets $V_{i,1}, V_{i,2}$ such that $E_i=V_{i,1}\times V_{i,2}$.}

	Apart from a family of tournaments of order eight that are \emph{not} $k$-inducible \citep{ShTo09a}, little was known about $3$-majority digraphs. In a much similar vein as \lemref{lem:2-voter-char}, we now provide a characterization of these digraphs.

	\begin{lemma} \label{3-voter-char}
		A tournament $(V,E)$ is $3$-inducible if and only if 
		there are {disjoint} sets $E_1,E_2$ with $E=E_1\cup E_2$ such that $E_1$ is transitive and $E_2$ is both acyclic and transitively reorientable.
		Then, the weight of every arc in~$E_1$ is either~$1$ or~$3$ and that of each arc in~$E_2$ is~$1$.
	\end{lemma}
	\begin{proof}
		For the if-direction, assume that there are {disjoint} sets $E_1,E_2$ with $E=E_1\cup E_2$ such that $E_1$ is transitive and $E_2$ is both acyclic and transitively reorientable. Consider the digraph $(V,E_1)$ and observe that for the corresponding incomparability graph $(V,\incompart E_1)$, $\incompart E_1=\bothwayedges{E_2}$. It follows that $\incompart E_1$ is transitively orientable and, by \lemref{lem:2-voter-char}, that $(V,E_1)$ is induced by a $2$-voter profile~$(R_1,R_2)$ and that all arcs in $E_1$ have weight~$2$. As~$E_2$ is acyclic, there is a (strict) preference relation~$R_3$ with $E_2\subseteq R_3$. Now consider the majority digraph induced by the preference profile $(R_1,R_2,R_3)$, which apparently coincides with $(V,E)$. $E_1$ is determined by~$R_1$ and~$R_2$ and each of its arcs  obtains weight~$1$ or~$3$ depending on whether~$R_3$ agrees with both~$R_1$ and~$R_2$ or not. Moreover, $E_2$ is determined by~$R_3$, as~$R_1$ and~$R_2$ can be assumed to specify contrary preferences on this part. 
		
		For the only-if-direction, assume that $(V,E)$ is the majority digraph induced by the $3$-voter profile  $(R_1,R_2,R_3)$. Let furthermore~$(V,E_1)$ be the majority digraph induced  by $(R_1,R_2)$ and $E_2=R_3\cap((V\times V)\setminus\bothwayedges{E_1})$. By \lemref{lem:2-voter-char}, $(V,E_1)$ is transitive and $\incompart{E_1}$ is transitively (re)orientable, where $(V,\incompart{E_1})$ is the incomparability graph of $(V,E_1)$.
		Since~$R_3$ is transitive (and strict) $E_2$ is obviously acyclic. Observe furthermore that
		$\bothwayedges{R_3\cap((V\times V)\setminus\bothwayedges{E_1})}=\bothwayedges{\incompart{E}_1}$. It follows that~$E_2$ is transitively reorientable. 
	\end{proof}

In order to illustrate \thmref{3-voter-char}, again consider the introductory example given in \figref{fig:digraph}. This digraph is $3$-inducible because its arc set can be partitioned into a transitive part and an acyclic and transitively reorientable part (see \figref{fig:3voter-example})

\begin{figure}[htb]
  \centering
\begin{minipage}{9.5em}
\begin{tikzpicture}[auto,scale=1]
	\draw[use as bounding box,opacity=0] (-1.3,-2) rectangle (1.3,1.5);
  \node[fill=white,circle,draw,minimum size=1.5em,inner sep=0pt] (1) at (162:1.2) {$a$};
  \node[fill=white,circle,draw,minimum size=1.5em,inner sep=0pt] (2) at (90: 1.2) {$b$};
  \node[fill=white,circle,draw,minimum size=1.5em,inner sep=0pt] (3) at (18: 1.2) {$c$};
  \node[fill=white,circle,draw,minimum size=1.5em,inner sep=0pt] (4) at (306:1.2) {$d$};
  \node[fill=white,circle,draw,minimum size=1.5em,inner sep=0pt] (5) at (234:1.2) {$e$};
  \draw[-latex] (4) to (1);
  \draw[latex-] (4) to (2);
  \draw[latex-] (4) to (3);
  \draw[latex-] (1) to (5);
  \draw[latex-] (2) to (5);
  \draw[-latex] (3) to (5);
  \draw[-latex] (4) to (5);
  \draw[-latex] (1) to (2);
  \draw[-latex] (2) to (3);
  \draw[latex-] (3) to (1);
\end{tikzpicture}
\end{minipage}
\begin{minipage}{9.5em}
\begin{tikzpicture}[auto,scale=1]
	\draw[use as bounding box,opacity=0] (-1.3,-2) rectangle (1.3,1.5);
	\draw(0,-1.77) node[scale=.85](){$E_1$ is transitive.};
  \node[fill=white,circle,draw,minimum size=1.5em,inner sep=0pt] (1) at (162:1.2) {$a$};
  \node[fill=white,circle,draw,minimum size=1.5em,inner sep=0pt] (2) at (90: 1.2) {$b$};
  \node[fill=white,circle,draw,minimum size=1.5em,inner sep=0pt] (3) at (18: 1.2) {$c$};
  \node[fill=white,circle,draw,minimum size=1.5em,inner sep=0pt] (4) at (306:1.2) {$d$};
  \node[fill=white,circle,draw,minimum size=1.5em,inner sep=0pt] (5) at (234:1.2) {$e$};
  \draw[-latex] (4) to (5);
  \draw[-latex] (1) to (2);
  \draw[-latex] (2) to (3);
  \draw[latex-] (3) to (1);
\end{tikzpicture}
\end{minipage}
\begin{minipage}{9.5em}
\begin{tikzpicture}[auto,scale=1]
	\draw[use as bounding box,opacity=0] (-1.3,-2) rectangle (1.3,1.5);
	\draw(0,-1.77) node[scale=.85](){
	\parbox{12em}{\centering $E_2$ is acyclic and}};
  \node[fill=white,circle,draw,minimum size=1.5em,inner sep=0pt] (1) at (162:1.2) {$a$};
  \node[fill=white,circle,draw,minimum size=1.5em,inner sep=0pt] (2) at (90: 1.2) {$b$};
  \node[fill=white,circle,draw,minimum size=1.5em,inner sep=0pt] (3) at (18: 1.2) {$c$};
  \node[fill=white,circle,draw,minimum size=1.5em,inner sep=0pt] (4) at (306:1.2) {$d$};
  \node[fill=white,circle,draw,minimum size=1.5em,inner sep=0pt] (5) at (234:1.2) {$e$};
  \draw[-latex] (4) to (1);
  \draw[latex-] (4) to (2);
  \draw[latex-] (4) to (3);
  \draw[latex-] (1) to (5);
  \draw[latex-] (2) to (5);
  \draw[-latex] (3) to (5);
\end{tikzpicture}
\end{minipage}
\begin{minipage}{9.5em}
\begin{tikzpicture}[auto,scale=1]
	\draw[use as bounding box,opacity=0] (-1.3,-2) rectangle (1.3,1.5);
	\draw(0,-1.77) node[scale=.85](){
	\parbox{12em}{\centering transitively reorientable.}};
  \node[fill=white,circle,draw,minimum size=1.5em,inner sep=0pt] (1) at (162:1.2) {$a$};
  \node[fill=white,circle,draw,minimum size=1.5em,inner sep=0pt] (2) at (90: 1.2) {$b$};
  \node[fill=white,circle,draw,minimum size=1.5em,inner sep=0pt] (3) at (18: 1.2) {$c$};
  \node[fill=white,circle,draw,minimum size=1.5em,inner sep=0pt] (4) at (306:1.2) {$d$};
  \node[fill=white,circle,draw,minimum size=1.5em,inner sep=0pt] (5) at (234:1.2) {$e$};
  \draw[latex-] (4) to (1);
  \draw[latex-] (4) to (2);
  \draw[latex-] (4) to (3);
  \draw[-latex] (1) to (5);
  \draw[-latex] (2) to (5);
  \draw[-latex] (3) to (5);
\end{tikzpicture}
\end{minipage}
  \caption{A $3$-inducible majority digraph and its arc set partitioning into $E_1$ and $E_2$.}
  \label{fig:3voter-example}
\end{figure}
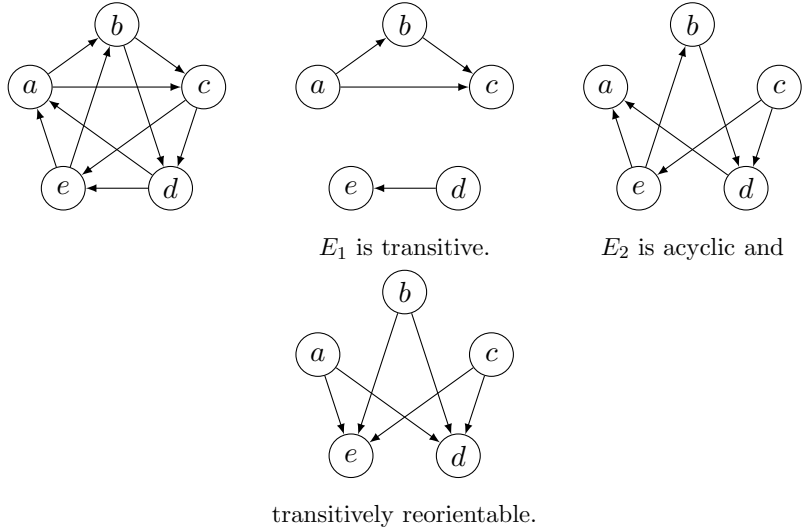

\subsection{More than Three Voters} 
\label{sub:more_than_three_voters}

	Extensions of these results provide useful sufficient conditions for a digraph to be induced by a constant larger number of voters. If the arc set of a digraph can be decomposed into pairwise orientation compatible sets that satisfy the conditions of \lemref{lem:2-voter-char}, the digraph is induced by a profile with two voters per set. 

	\begin{lemma} \label{lem:2n-voter-suff}
		Let  $(V,E_1),\dots,(V,E_k)$ be majority digraphs induced by $2$-voter profiles such that $E_1,\dots,E_k$ are pairwise orientation compatible. Then, $(V,E_1\cup\dots\cup E_k)$ is induced by a $2k$-voter profile.
	\end{lemma}

	\begin{proof}
		Let for each~$i$ with $1\le i\le k$, $(R^i_1,R^i_2)$ be a $2$-voter profile that induces $(V,E_i)$. By \lemref{lem:2-voter-char}, for every $(v,w)\in E_i$ we know that both $v\mathrel{R_1^i}w$ and $v\mathrel{R_2^i}w$ and for every $(v,w)\notin E_i$, $v\mathrel{R_1^i}w$ if and only if $w\mathrel{R_2^i}v$.  Now consider the preference profile
		$
			(R^1_1,R^1_2,\dots,R^k_1,R^k_2)
		$ and the majority digraph $(V,E)$ it induces. We argue that $E=E_1\cup\dots\cup E_k$. First assume that $(v,w)\in E_i$ for some~$i$ with $1\le i\le k$. Then, both $v\mathrel{R_1^i}w$ and $v\mathrel{R_2^i}w$. Since $E_1,\dots,E_k$ are pairwise orientation compatible, $(w,v)\in E_j$ for no~$j$ with $1\le j\le k$, \ie for all ~$j$ with $1\le j\le k$  either $v\mathrel{R_1^j}w$ and $v\mathrel{R_2^j}w$, or  $v\mathrel{R_1^j}w$ if and only if $w\mathrel{R_2^j}v$. It follows that a majority prefers $v$ over $w$ and thus $(v,w)\in E$. 
		Now assume that $(v,w)\in E_i$ for no~$i$ with $1\le i\le k$. Then for all~$i$ with $1\le i\le k$ either both $w\mathrel{R_1^i}v$ and $w\mathrel{R_2^i}v$ or $w\mathrel{R_1^j}v$ if and only if $v\mathrel{R_2^j}w$. It is easy to see that $v$ is not majority preferred to $w$, \ie $(v,w)\notin E$.
	\end{proof}

	Next, we show that a similar condition suffices for a digraph to be inducible by a given odd number of voters.%
\footnote{The if-direction of \lemref{3-voter-char} can also be obtained as a special case of this lemma.}

	\begin{lemma} \label{lem:2n+1-voter-suff}
		Let $(V,E)$ be a tournament and $(V,E_1),\dots,(V,E_k)$ be majority digraphs induced by $2$-voter profiles such that $E,E_1,\dots,E_k$ are orientation compatible. Let, moreover, $E_{k+1}\supseteq E\setminus(E_1\cup\dots\cup E_k)$ be acyclic. Then, $(V,E)$ is induced by a $2k+1$-voter profile.
	\end{lemma}

	\begin{proof}
		In virtue of \lemref{lem:2n-voter-suff} we know that $(V,E_1\cup\dots\cup E_k)$ is induced by a $2k$-voter profile $(R^1_1,R^1_2,\dots,R^k_1,R^k_2)$. Inspection of the proof also reveals that every arc $(v,w)\in E_1\cup\dots\cup E_k$ has a positive even weight of at least two. As $E_{k+1}$ is acyclic and asymmetric, there is some (strict) preference relation $R^{k+1}$ with $E_{k+1}\subseteq R^{k+1}$. Moreover, since $E_{k+1}$ corresponds to only one voter and every arc in $E_1\cup\dots\cup E_k$ has a majority of at least two, $E_{k+1}$ does not have to be orientation compatible with any of $E_1,\ldots,E_k$. It can then easily be seen that the majority digraph induced by $(R^1_1,R^1_2,\dots,R^k_1,R^k_2,R^{k+1})$ equals $(V,E)$, $E_1\cup\dots\cup E_k$ being determined by majorities of at least one in $(R^1_1,R^1_2,\dots,R^k_1,R^k_2,R^{k+1})$ and $E\setminus(E_1\cup\dots\cup E_k)$ by~$R^{k+1}$, each arc in which has then weight one.
	\end{proof}

\section{Determining the Majority Dimension of a Digraph} 
\label{sec:determining_the_majority_dimension_of_a_digraph}
	This section addresses the computational problem of computing the \majdimname. To this end, we define the problem of checking whether a given digraph $G$ is $k$-inducible, \ie whether $G$ is a $k$-majority digraph.

	\defineDecisionProblem{\checkkmajorityproblem}{A digraph $G$ and a positive integer $k$.}{Is $G$ a $k$-majority digraph?}

	Recall that for a digraph~$G$, whether $\majdim(G)$ is odd or even depends on whether~$G$ is complete (\ie a tournament) or not, according to \lemref{lem:parity-of-majdim}. While \checktwomajorityproblem can be solved in polynomial time \citep{Yann82a}, the complexity of \checkkmajorityproblem remains open for every fixed $k\ge 3$.

	In the following, we provide an implementation for solving \checkkmajorityproblem. This implementation relies on an encoding of the problem as a Boolean satisfiability (SAT) problem which is then solved by a SAT solver. This technique turns out to be surprisingly efficient and easily outperforms an implementation for \checkthreemajorityproblem based on the graph-theoretic characterization in \secref{sec:majority_relations_of_few_voters}. 

	\subsection{Computing the Majority Dimension via SAT} 
	\label{sub:computing_the_majority_dimension_via_sat}

	The number of objects potentially involved in the \textsc{Check-$k$-Majority} problem are given in \tabref{tab:magnitudes}. It is immediately clear that a naive algorithm will not solve the problem in a satisfactory manner.

	\begin{table*}[tb]
	\centering
	\begin{tabular}{>{$}l<{$}*{4}{>{$}c<{$}}}
	\toprule
				& \multicolumn{3}{c}{Preference profiles} 										& \text{Tournaments}	\\
				& k=1 							& k=3 						& k=5 								& \text{(unlabeled)}	\\
	\midrule
	n=5 		& 120 							& \sim 1.7 \cdot 10^6\pss\pss 	& \sim 2.5 \cdot 10^{10}\pss 	& 12 					\\
	n=10 		& \sim 3.6 \cdot 10^{6\pss\pss} & \sim 4.8 \cdot 10^{19}\pss 	& \sim 6.3 \cdot 10^{32}\pss 	& \sim 9.7 \cdot 10^{6}\pss\pss \\
	n=25 		& \sim 1.6\cdot 10^{25}\pss 	& \sim 3.7 \cdot 10^{75}\pss 		& \sim 9.0\cdot 10^{125} 	& \sim 1.3 \cdot 10^{65}\pss\\
	n=50 		& \sim 3.0\cdot 10^{64}\pss 	& \sim 2.8 \cdot 10^{193} 			& \sim 2.6 \cdot 10^{322} 	& \sim 1.9\cdot 10^{305}\\
	n=100 		& \sim 9.3\cdot 10^{157}		& \sim 8.1 \cdot 10^{473} 			& \sim 7.1\cdot 10^{789}	& >10^{1332}	\\
	\bottomrule
	\end{tabular}
	\caption{Number of objects involved in the \textsc{Check-$k$-Majority} problem for one, three, and five voters.}
	\label{tab:magnitudes}
	\end{table*}

	Thus, in order to answer 
	\checkkmajorityproblem,
	we translate the problem to propositional logic (on a computer) and use state-of-the-art SAT solvers to find a solution. 
	At a glance, the overall solving steps are shown in~\algref{alg:SATcheck}.

	\begin{algorithm}[t]
		\textbf{Input:} digraph $(V,E)$, positive integer $k$\\
		\textbf{Output:} whether $(V,E)$ is a $k$-majority digraph\\
			\tcc{Encoding of problem in CNF}
			File cnfFile\;
			\ForEach{voter $i$}{
				cnfFile += Encoder.reflexivePreferences($i$)\;
				cnfFile += Encoder.completePreferences($i$)\;
				cnfFile += Encoder.transitivePreferences($i$)\;
				cnfFile += Encoder.antisymmetricPreferences($i$)\;
			}
			cnfFile += Encoder.majorityImplications($(V,E)$)\;
			\If{$E$ is not complete}{
				cnfFile += Encoder.indifferenceImplications($(V,E)$)\;
			}
			\tcc{SAT solving}
			Boolean satOutcome = SATsolver.solve(cnfFile)\;
			\Return satOutcome;
			\caption{\sat-\textsc{Check-$k$-Majority}}
			\label{alg:SATcheck}
	\end{algorithm}

	A satisfying instance of the propositional formula to be designed should represent a preference profile that induces the given digraph. We encode the preference profile in question using Boolean variables $r_{i,a,b}$, which encode whether $a\mathrel{R_i} b$, \ie whether voter $i$ ranks alternative $a$ at least as high as alternative $b$. We then need to impose the following constraints.
	\begin{enumerate}
		\item All $k$ voters have linear orders over the $n$ alternatives as their preferences.
		\item For each majority arc $(x,y)\in E$ in the digraph, a majority of voters needs to prefer $x$ over $y$.
		\item For each missing arc ($x\nsucc y$ and $y\nsucc x$) in the digraph, \emph{exactly} half the voters need to prefer $x$ over $y$.\footnote{This axiom is only required for incomplete digraphs.}
	\end{enumerate}

	For the first constraint, we encode reflexivity, completeness, transitivity, and anti-symmetry of the relation $R_i$ for all voters $i$. 
	The complete translation to CNF (conjunctive normal form, the established standard input format for SAT solvers) is given exemplarily for the case of transitivity; reflexivity, completeness, and anti-symmetry are converted analogously: 

	\begin{eqnarray*}
		&& (\forall i)(\forall x,y,z)\left(x\mathrel{R_i}y \wedge y\mathrel{R_i}z \rightarrow x\mathrel{R_i}z \right) \\
		&\equiv & (\forall i)(\forall x,y,z)\left(r_{i,x,y} \wedge r_{i,y,z} \rightarrow r_{i,x,z} \right) \\
		&\equiv & \bigwedge_i\bigwedge_{x,y,z}\left(\neg \left(r_{i,x,y} \wedge r_{i,y,z}\right) \vee r_{i,x,z} \right) \\
		&\equiv & \bigwedge_i\bigwedge_{x,y,z}\left(\neg r_{i,x,y} \vee \neg r_{i,y,z} \vee r_{i,x,z} \right)\text{.}
	\end{eqnarray*}
	The key in the translation of the inherently higher order axioms to propositional logic is (as pointed out by \citet{GeEn11a} already) that due to finite domains, all quantifiers can be replaced by finite conjunctions or disjunctions, respectively.

	Majority and indifference implications can be formalized in a similar fashion. 
	We describe the translation for the majority implications here; the procedure for the indifference implications (needed for incomplete digraphs) is analogous.  
	In the following, we denote the smallest number of voters required for a positive majority margin by $m(k):=\lfloor k\cdot \frac{1}{2} \rfloor + 1$. 
	Note that, due to the anti-symmetry of individual preferences, for $(x,y)\in E$ it suffices that there exist $m(k)$ voters who prefer $x$ to $y$. 
	In formal terms:
	\begin{eqnarray*}
		&& (\forall x,y)\left((x,y)\in E \rightarrow |\{i\midd x\mathrel{R_i} y\}| > |\{i\midd y\mathrel{R_i} x\}| \right) \\
		&\equiv & (\forall x,y)\left((x,y)\in E \rightarrow |\{i\midd x\mathrel{R_i} y\}| \geq m(k) \right) \\
		&\equiv & (\forall x,y)\left((x,y)\in E \rightarrow  
		 (\exists M\subseteq K) |M|=m(k) \wedge (\forall i\in M) x\mathrel{R_i} y \right) \\
		&\equiv & \bigwedge_{(x,y)\in E}\bigvee_{|M|=m(k)} \bigwedge_{i\in M} r_{i,x,y} \text{.}
	\end{eqnarray*}

	In order to avoid an exponential blow-up when converting this formula to CNF, the standard technique of variable replacement (also known as Tseitin transformation \citep{Tsei83a}) is applied.
	Note that the conditions like $|M|=m(k)$ can easily be checked during generation of the corresponding CNF formula on a computer. 

	Overall, this encoding leads to a total of $k\cdot n^{2} + \binom{k}{m(k)} \cdot n^{2} = n^{2} \cdot \left(k + \binom{k}{m(k)}\right)$ variables for the case of tournaments and $n^{2} \cdot \left(k + \binom{k}{m(k)} + \binom{k}{\nicefrac{k}{2}} \right)$ variables for incomplete digraphs.
	The number of clauses is equal to $k\cdot (n^3+n^2) + \frac{n^2-n}{2} \cdot \left(1+\binom{k}{m(k)}\cdot m(k)\right)$ for tournaments, and at most $k\cdot (n^3+n^2) + (n^2-n) \cdot \left(1+\binom{k}{k/2}\cdot \frac{k}{2}\right)$ for incomplete digraphs, respectively.

	With all axioms formalized in propositional logic, we are now ready to analyze arbitrary digraphs $G$ for their \majdimname $\majdim(G)$.
	Before we do so, however, we describe an optimization technique for tournament graphs, which, for certain instances, significantly speeds up the computation. 

	\subsubsection{Optimization through Decomposition of Tournaments}

	An important structural property in the context of tournaments is whether a tournament admits a non-trivial decomposition \citep[see, \eg][]{Lasl97a,BBH15a}. 
	\citet{BBS11a} show how these decompositions can be exploited to recursively determine the winners of certain voting rules.
	In this section, we prove that a similar optimization can be carried out for the computation of the \majdimname $\majdim(G)$ of a given tournament $G$. 
	In particular, we show that the \majdimname of a tournament is equal to the maximum of the \majdimname of its components and the corresponding summary.

	In formal terms, a non-empty subset $C$ of $V$ is a \emph{component} of a tournament $G=(V,E)$ if, for all $v\in V\setminus C$, either  $(v,w)\in E$ for all $w\in C$ or $(w,v)\in E$ for all $w\in C$.
	A \emph{decomposition} of $G$ is a set of pairwise disjoint components $\{C_1,\dots,C_p\}$ of $T$ such that $V=\bigcup_{j=1}^p C_j$. 
	Every tournament admits a decomposition that is minimal in a well-defined sense \citep{Lasl97a} and that can be computed in linear time \citep{McMo05a,BBS11a}.
	Given a particular decomposition $\tilde{C}=\{C_1,\dots,C_p\}$, the \emph{summary} of $(V,E)$ with respect to $\tilde{C}$ is defined as the tournament $(\tilde{C},\tilde{E})$ on the individual components rather than the alternatives, \ie
	\[
		(C_q,C_r)\in \tilde{E} \quad\text{ if and only if }\quad (v,w)\in E \text{ for all $v\in C_q$, $w\in C_r$.}
	\]

	The following lemma enables the recursive computation of $\majdim(G)$ along the component structure of $G$.
	\begin{lemma}
		\label{lem:decomposition}
		Let $G=(V,E)$ be a tournament, $\tilde{C}=\{C_1,\dots,C_p\}$ a decomposition of $G$, $G_j=(C_j,E|_{C_j})$ for all $j\in\{1,\dots,p\}$, and $\tilde{G}=(\tilde{C},\tilde{E})$. Then
		\[
		\majdim(G) = \max_j \{\majdim(G_j,\majdim(\tilde{G}))\}\text{.}
		\]
	\end{lemma}
	\begin{proof}
		Let $R$ be a minimal profile inducing $G$. Then, $R|_{C_j}$ induces $G_j$ for every $C_j$ establishing $\majdim(G)\geq \majdim(G_j)$. 
		That $\majdim(G)\geq \majdim(\tilde{G})$ holds is also easy to see by considering a variant of $R$ in which from each component all but one vertex are chosen arbitrarily and removed. 
		The remaining profile then induces $\tilde{G}$.
		For the other direction, let $z=\max_j \{\majdim(G_j),\majdim(\tilde{G})\}$. 
		We know by \lemref{lem:parity-of-majdim} that $\majdim(\tilde{G})$ and every $\majdim(G_j)$ is odd as these are all tournaments. 
		Each $G_j$ (and $\tilde{G}$) has a minimal profile $R^j$  (and $\tilde{R}$, respectively). 
		We can add pairs of voters with opposing preferences to each profile without changing its majority relation. This way, we get profiles $R'^j$ (and $\tilde{R}'$) that still induce $G_j$ (and $\tilde{G}$, respectively) but now all have the same number of voters $z$.
		Now, create a new profile $R$ from $R'$ in which $R_i^j$ replaces alternative $j$ as a segment in $R'_i$ for each voter $i$ and every alternative $j$. 
		It is easy to check that $R$ has $z$ voters and still induces $G$, \ie $\majdim(G)\geq z$.
	\end{proof}

	We have implemented this optimization and found that many real-world majority digraphs exhibit non-trivial decompositions, speeding up the computation of \satcheckk.

\subsection{Computational Efficiency} 
\label{sub:computational_efficiency}

	The characterization of $3$-majority digraphs in \secref{sec:majority_relations_of_few_voters} allows for a straightforward algorithm, which is expected to have a much better running time than any naive implementation enumerating all preference profiles (also compare \tabref{tab:magnitudes}). 
	The corresponding algorithm \naivecheckthree  is given in \algref{alg:naive-check3}. 
	Besides enumerating all $2$-partitions of the majority arcs, the only non-trivial part is to check whether a relation has a transitive reorientation. 
	This can be done efficiently using an algorithm by \citet{PLE71a}.

	We compared the running times of \naivecheckthree with the ones of our implementation via \sat as described in \secref{sub:computing_the_majority_dimension_via_sat} (see also \algref{alg:SATcheck}).\footnote{As a programming language, Java was used in both cases.}

	It turns out that---even though it is much more universal---\satcheckk[3] offers significantly better running times (see \tabref{tab:runtime-comparison}).
	Moreover, \satcheckk[k] directly returns a resulting preference profile with $k$ voters.

	\begin{table*}[t]
	\centering
	\begin{tabular}{lcrrrccccc}
		\toprule
		Algorithm	&	$5$	&	\wordbox[c]{$6$}{$<0.1$s}	&	\wordbox[c]{$7$}{$<0.1$s}	&	\wordbox[c]{$8$}{$<0.1$s}	&	$9$ & $10$	& $20$	& $50$ & $100$ \\
		\midrule		
		\sat				& $<0.1$s	& $<0.1$s	& $<0.1$s	& $<0.1$s	& $<0.1$s	& $<0.1$s & $0.1$s	& $1.5$s & $12.5$s		 \\ 
		\textsc{$2$-Partition}	& $<0.1$s	& $<0.1$s	& $0.1$s		& $\mathwordbox[r]{530}{<0.1}$s 	& --- 		& ---	& ---	& --- & --- \\
		\bottomrule
	\end{tabular}	
	\caption{Comparison of average runtimes of \satcheckk[3] and \naivecheckthree for randomly sampled tournaments of size $n$ with a cutoff time of one hour.}
	\label{tab:runtime-comparison}
	\end{table*}
		
	Further runtimes, which exhibit the practical power of our SAT approach (and its limits), can be obtained from \tabref{tab:sat-runtime}. All experiments were run on an Intel Core i5, 2.66GHz (quad-core) machine with 12 GB RAM using the SAT solver \textsc{plingeling} \citep{Bier13a}. 
	Interestingly, an integer programming (IP) formulation of the problem by \citet{EHW13a} appears to perform worse than our SAT-based formalization: \citeauthor{EHW13a} report that for $n>20$ runtimes are prohibitively large. 

	\begin{algorithm}[htb]
		\caption{\naivecheckthree}
		\label{alg:naive-check3}
		\textbf{Input:} tournament $(V,E)$\\
		\textbf{Output:} whether $(V,E)$ is a $3$-majority digraph\\
				\ForEach{$2$-partition $\{E_1,E_2\}$ of $E$}{
					\If{$E_1$ is transitive and $E_2$ is acyclic and $E_2$ has a transitive reorientation}{
						\Return true\;
					} 
				}
				\Return false\;
	\end{algorithm}

	\begin{table*}[tb]
		\centering
		\scalebox{0.85}{
		$
		\begin{array}{cr*{9}{@{\quad}r}}
		\toprule
		n \backslash k &	\mathwordbox[c]{3}{00.00} &	\mathwordbox[c]{4}{00.00} &	\mathwordbox[c]{5}{00.00} &	\mathwordbox[c]{6}{00.00} &	\mathwordbox[c]{7}{00.00} &	\mathwordbox[c]{8}{00.00} &	\mathwordbox[c]{9}{00.00} &	\mathwordbox[c]{10}{00.00} &	\mathwordbox[c]{11}{00.00} &	\mathwordbox[c]{12}{00.00}                    \\
		\midrule
		3  &	.04 &	.04 &	.03 &	.04 &	.04 &	.04 &	.04 &	.05   &	.08 &	.10                   \\
		4  &	.03 &	.04 &	.03 &	.04 &	.04 &	.04 &	.05 &	.07   &	.10 &	.18                   \\
		5  &	.03 &	.04 &	.03 &	.04 &	.06 &	.05 &	.06 &	.09   &	.16 &	.35                   \\
		6  &	.03 &	.04 &	.04 &	.04 &	.05 &	.06 &	.08 &	.12   &	.27 &	.63                   \\
		7  &	.04 &	.04 &	.04 &	.05 &	.05 &	.07 &	.10 &	.17   &	.45 &	1.10                  \\
		8  &	.04 &	.05 &	.05 &	.05 &	.07 &	.08 &	.13 &	.23   &	.69 &	1.80                  \\
		9  &	.04 &	.05 &	.05 &	.64 &	.07 &	.10 &	.17 &	.33   &	1.06 &	2.83                  \\
		10 &	.05 &	.05 &	.06 &	.67 &	.09 &	.12 &	.23 &	.46   &	1.56 &	4.25                  \\
		11 &	.06 &	.06 &	.06 &	1.92 &	.10 &	.14 &	.30 &	.63   &	2.22 &	6.37                  \\
		12 &	.06 &	.07 &	.07 &	3.35 &	.12 &	.19 &	.40 &	.85   &	3.18 &	8.48                  \\
		13 &	.07 &	.07 &	.09 &	3.93 &	.15 &	.27 &	.52 &	1.16  &	4.44 &	12.30                 \\
		14 &	.07 &	.09 &	.10 &	4.15 &	.18 &	.36 &	.64 &	1.51  &	5.99 &	16.84                 \\
		15 &	.08 &	.10 &	.13 &	3.89 &	.21 &	.88 &	.79 &	2.22  &	7.67 &	\none                 \\
		16 &	.09 &	.11 &	.14 &	4.12 &	.25 &	4.55 &	.99 &	2.90  &	9.80 &	\none                 \\
		17 &	.10 &	.12 &	.19 &	4.41 &	.29 &	7.15 &	1.23 &	4.69  &	12.48 &	\none                 \\
		18 &	.11 &	.14 &	.23 &	4.76 &	.35 &	17.51 &	1.53 &	8.25  &	15.97 &	\none                 \\
		19 &	.12 &	.15 &	.35 &	4.97 &	.43 &	\none &	1.80 &	\none   &	19.99 &	\none                 \\
		20 &	.13 &	.17 &	.54 &	5.04 &	.47 &	\none &	2.21 &	\none   &	\none &	\none                      \\
		21 &	.14 &	.18 &	5.87 &	6.15 &	.63 &	\none &	2.71 &	\none   &	\none &	\none                      \\
		22 &	.16 &	.20 &	11.07 &	5.43 &	.96 &	\none &	3.24 &	\none   &	\none &	\none                      \\
		23 &	.17 &	.23 &	18.95 &	5.76 &	1.57 &	\none &	4.12 &	\none   &	\none &	\none                      \\
		24 &	.20 &	.26 &	\none &	5.87 &	2.56 &	\none &	4.60 &	\none   &	\none &	\none              \\
		25 &	.22 &	.29 &	\none &	6.12 &	4.21 &	\none &	5.85 &	\none   & 	\none &	\none		      \\
		\bottomrule
	\end{array}
	$
	} 
	\caption{Runtime in seconds of \satcheckk[k] for different number of alternatives and different number of voters $k$ when average runtimes did not exceed $20$ seconds. For this table, averages were taken over $5$ samples from the uniform random tournament model.}
	\label{tab:sat-runtime}
	\end{table*}

\section{Analyzing Majority Dimensions} 
\label{sec:analyzing_majority_dimensions}	
	
	Using the algorithm described in the previous section, we are now in a position to analyze the majority dimension of digraphs. In this section, we report on our findings for different sources of digraphs.

\subsection{Exhaustive Analysis} 
\label{sub:exhaustive_analysis}

	Using the tournament generator from the \textsc{nauty} toolkit \citep{McPi13a}, we generated all tournaments with up to $10$ alternatives and found that all of them are $5$-inducible.
	In fact, all tournaments of size up to seven are even $3$-inducible, confirming a conjecture by \citet{ShTo09a}. \citeauthor{ShTo09a} also showed that there exist tournaments of size $8$ that are not $3$-inducible. We also confirmed that the exact number of such tournaments is $96$ (out of $6880$) as found by \citet{EHW13a}. One of these tournaments is depicted in \figref{fig:not-3inducible}.

	Like \citet{EHW13a}, we have not encountered a single tournament for which we could show that it is not $5$-inducible.
	Since quadratic residue tournaments of enormous size are the only concrete tournament of which we know that they have higher \majdimname (see \secref{sec:bounds}), we examined small tournaments of this kind as well and found that
	\[
		\majdim(Q_{11})=3 \quad \text{and} \quad \majdim(Q_{19})=5\text{.}
	\]
	Unfortunately, we were not able to check whether the \majdimname of $Q_{23}$ is equal to $5$ or larger as the SAT solver did not terminate within a total of six weeks.\footnote{Another specific tournament that we considered is a tournament on $24$ alternatives used by \citet{BrSe13a} to disprove a conjecture in social choice theory \citep{BCK+11a}. We found that this tournament is $5$-inducible, which implies that the negative consequences of the counterexample already hold for settings with only $5$ voters (and at least $24$ alternatives).}

\subsection{Empirical Analysis} 
\label{sub:empirical_analysis}

	In the preference library \preflib \citep{MaWa13a}, scholars have contributed data sets from real world scenarios ranging from preferences over movies or sushi via Formula~1 championship results to real election data. Accordingly, the number of voters whose preferences originally induced these data sets vary heavily between $4$ and $44,000$.
	At the time of writing, \preflib contained $354$ tournaments induced from pairwise majority comparisons as well as $185$ incomplete majority digraphs.

	Among the tournaments in \preflib, $58$ are $3$-inducible. The two largest tournaments in the data set have $240$ and $242$ vertices, respectively. The first one is a $5$-majority tournament and the \sat solver did not terminate on the second one within one day. The remaining tournaments are transitive and thus $1$-inducible. Therefore, all checkable tournaments in \preflib are $5$-inducible.
 
	For the non-complete majority digraphs in \preflib, we found that the indifference constraints which are imposed on all missing arcs change the picture. Not only does it negatively affect the running time of \satcheckk in comparison to tournaments which made us restrict our attention to instances with at most $40$ alternatives, but it also results in higher majority dimensions of up to $8$ among the $85$ feasible instances. 

\subsection{Stochastic Analysis} 
\label{sub:stochastic_analysis}

	Additionally, we consider \emph{stochastic} models to generate tournaments of a given size~$n$. Many different models for linear preferences (or orderings) have been considered in the literature. We refer the interested reader to \citet{CFV91a}, \citet{Mard95a}, \citet{MFG12a}, and \citet{BrSe14a}. In this work, we decided to examine tournaments generated with five different stochastic models.

	In the \emph{uniform random tournament model}, the same probability is assigned to each \emph{labeled} tournament of size $n$, \ie
	\[
		\Pr(T)=2^{-\binom{n}{2}} \text{ for each } T \text{ with } |T|=n\text{.}
	\]
	In all of the remaining models, we sample preference profiles and work with the tournament induced by the majority relation. In accordance with \citet{McSl06a,BrSe14a}, we generated profiles with $51$ voters.

	The \emph{impartial culture model} (IC) is the most widely-studied model for individual preferences in social choice. It assumes that every possible preference ordering has the same probability of $\frac{1}{n!}$. 
	If we add anonymity by having indistinguishable voters, the set of profiles is partitioned into equivalence classes. In the \emph{impartial anonymous culture} (IAC), each of these equivalence classes is chosen with equal probability.

	In \emph{Mallows-$\phi$ model} \citep{Mall57a}, the distance to a reference ranking is measured by means of the Kendall-tau distance which counts the number of pairwise disagreements. Let $R_0$ be the reference ranking. Then, the Kendall-tau distance of a preference ranking $R$ to $R_0$ is
	\[
		\tau(R,R_0) = \binom{n}{2} - \left( \left| R \cap R_0 \right| -n\right) \text{.}
	\]
	According to the model, this induces the probability of a voter having $R$ as his preferences to be
	\[
		\Pr(R) = \frac{\phi^{\tau(R,R_0)}}{C} 
	\]
	where $C$ is a normalization constant and $\phi\in(0,1]$ is a dispersion parameter. Small values for $\phi$ put most of the probability on rankings very close to $R_0$ whereas for $\phi=1$ the model coincides with IC.%
	\footnote{An interpretation of distance-based models such as Mallows-$\phi$ model is that there exists a pre-existing truth in the form of a reference ordering and the agents report noisy estimates of said truth as their preferences. For these models, \citeauthor{Lasl10a} has introduced the term \emph{Rousseauist cultures} \citep{Lasl10a}.}

	A very different kind of model is the \emph{spatial model}. Here, alternatives and voters are placed uniformly at random in a multi-dimensional space and the voters' preferences are determined by the (Euclidian) distanced to the alternatives. The spatial model plays an important role in political and social choice theory where the dimensions are interpreted as different aspects or properties of the alternatives \citep[see, \eg][]{Orde93a,AuBa00a}.

	For up to $21$ alternatives, we sampled preference profiles (each consisting of $51$ voters\footnote{In another study \citep{BrSe14a}, this size turned out to be sufficiently large to discriminate the different underlying stochastic models.}) from the aforementioned stochastic models and examined the corresponding majority digraphs for their \majdimname using \satcheckk. The average complexities over $30$ instances of each size are shown in Table~\ref{tab:stochastic}. We see that the unbiased models (IC, IAC, uniform) tend to induce digraphs with higher \majdimname. 

	Again, we encountered no tournament that was not a $5$-majority tournament. We also checked more than $8$ million uniform random tournaments with $12$ alternatives.
	These results could be used to argue that the majority dimension itself may be employed as a parameter to govern the generation of realistic preference profiles.

	\begin{table}[htb]
		\centering
		$
	\begin{array}{ccc@{\qquad}c@{\qquad}cc}
		\toprule
		n	& \text{uniform} & \text{IC}	& \text{IAC}  &  \begin{tabular}{@{}c@{}}Mallows-$\phi$ \\ $(\phi=0.95)$\end{tabular}  & \begin{tabular}{@{}c@{}}spatial \\ ($\text{dim}=2$)\end{tabular} %
		\\
		\midrule    
		3	& 1.40 & 1.13  & 1.13 &	1.13   & 1.00               \\
	   	5	& 3.00 & 1.67  & 2.13 &	1.33   & 1.13               \\
	   	7	& 3.00 & 2.67  & 2.67 &	2.47   & 1.33               \\
	   	9	& 3.13 & 3.00  & 3.00 &	2.67   & 1.60               \\
	   	11	& 3.93 & 3.07  & 3.00 &	2.87   & 2.33               \\
	   	13	& 4.80 & 3.07  & 3.20 &	2.93   & 2.53               \\
	   	15	& 5.00 & 3.27  & 3.40 &	3.00   & 2.67               \\
	   	17	& 5.00 & 3.40  & 3.80 &	2.93   & 2.80               \\
	   	19	& 5.00 & 4.27  & 4.20 &	3.00   & 2.80               \\
	   	21	& 5.00 & 4.47  & 4.33 &	3.00   & 2.87               \\
		\bottomrule	
	\end{array}
	$
	\caption{Average \majdimname in tournaments generated by stochastic (preference) models. The given values are averaged over $30$ samples each.}
	\label{tab:stochastic}
	\end{table}

\section{Hardness of Voting with Few Voters} 
\label{sec:hardness_of_voting_with_few_voters}

	In this section, we show that the winner determination problem of four well-studied voting rules remains \NP-hard even if the number of voters is a small constant. Our general method is to analyze existing hardness constructions for these rules with respect to their susceptibility to the sufficient conditions in \lemref{lem:2n-voter-suff} or \lemref{lem:2n+1-voter-suff}. In all cases, we slightly modify the hardness constructions to get better bounds on the number of voters. Before we proceed, we introduce two new constrained classes of propositional formula (\orderedthreecnf, to be used for the results in Sections \ref{sub:the_banks_set} and \ref{sub:the_tournament_equilibrium_set}, and \redfewcnf, to be used for the result in Section \ref{sub:the_slater_set}) and show for both that the problem of deciding whether a given formula is satisfiable is \NP-complete.

	A formula of propositional logic in conjunctive normal form (CNF) is in \threecnf if each clause has at most three literals. We say that a formula~$\varphi$ from \threecnf is in \orderedthreecnf if its clauses all contain exactly three distinct literals and are ordered within~$\varphi$ in such a way that for each propositional variable $p$, all clauses containing the literal~$p$ precede all clauses containing~$\neg p$.
	It is known that \threesat, the problem of deciding whether a given formula in \threecnf is satisfiable, is \NP-complete \citep{Karp72a}. For formulae in \orderedthreecnf, we call the corresponding decision problem \orderedthreesat.

	\begin{lemma}
		\label{lem:ordered-threesat-hardness}
		\orderedthreesat is \NP-complete.
	\end{lemma}
		\begin{proof}
			Membership in {\NP} is obvious. 
			For hardness, we reduce from \threesat. Let $\varphi$ be some formula in \threecnf.
			Let $P$ denote the set of variables of the propositional language in which~$\varphi$ is formulated and let $C=(c_1,\ldots,c_{|C|})$ denote the clause set of~$\varphi$.
			We may assume without loss of generality that no clause contains the same variable twice, that all literals in a clause are ordered according to a fixed  ordering $(p_1,p_2,\ldots)$, and that every clause is of size three. The latter is due to the fact that clauses of size one can be easily used to simplify~$\varphi$ and the remaining clauses $(p\vee q)$ of size two can be padded with a new variable~$x$ to $(p \vee q \vee x) \wedge (p \vee q \vee \neg x)$. 
			We call all variables that occur at least once in $\varphi$ \emph{original} variables.
			
			For the reduction, we construct an ordered formula~$\varphi'$ in \threecnf with $6\cdot|C|$ clauses and $4\cdot|C|$ additional variables that is satisfiable if and only if~$\varphi$ is.
			For every clause $c_i=(\ell_1 \vee \ell_2 \vee \ell_3)$, define a set of new clauses $\varphi_i=\bigwedge_{j=1}^6 c_i^j$ with
			\begin{align*}
				c_i^1 &= (\ell_1 \vee x_i \vee x'_i), 	 & c_i^2 &= (\ell_2 \vee \neg x_i \vee y_i),\\
 				c_i^3 &= (\ell_2 \vee \neg x'_i \vee y_i), &	c_i^4 &= (\ell_3 \vee \neg y_i \vee z_i),\\
				c_i^5 &= (\neg x_i \vee \neg y_i \vee \neg z_i), \text{ and } & c_i^6 &= (\neg x'_i \vee \neg y_i \vee \neg z_i)
			\end{align*}
			where $x_i, x'_i, y_i$, and $z_i$ are new propositional variables. 
			It is easy to check that $c_i$ is satisfiable if and only if $\varphi_i$ is. Since the literals associated with original variables are spread over different $\varphi_i$ just as they were over the different clauses $c_i$ in $\varphi$, this implies that $\bigwedge_i\varphi_i$ is satisfiable if and only if $\varphi$ is.

			What remains to be shown is that all the clauses~$c_i^j$ can be arranged in such a way that the resulting formula is ordered.
			To this end, we define for each original variable~$p$ and $j\in\{1,\ldots,4\}$ the clause sets
			\[
			C^{p,j} = \bigcup_i \{c_i^j \midd p\in c_i^j\} \qquad \text{ and } \qquad
			C^{\neg p,j} = \bigcup_i \{c_i^j \midd p\in c_i^j\}
			\]
			as well as
			\[
				C^5 = \bigcup_i c_i^5 \qquad\text{ and } \qquad
				C^6 = \bigcup_i c_i^6\text{.}
			\]
			We are now in a position to define $\varphi'$ to be
			\[
				\varphi' ={}
				\bigwedge_{i=1}^{|P|} 
					\Bigg(
					\bigg( \bigwedge_{j=1}^4 
						\bigwedge_{c\in C^{p_i,j}} c
					\bigg) \wedge 
						\bigg( \bigwedge_{j=1}^4 
						\bigwedge_{c\in C^{\neg p_i,j}} c
					\bigg)
					\Bigg) \wedge
				\bigwedge_{c \in C^5\cup C^6} c.
			\]
			
			We claim that $\varphi'$ is ordered. We show this for original and new variables separately. 
			For each original variable~$p$, all positive occurrences are in the $C^{p,j}$, preceding the negative occurrences in the~$C^{\neg p, j}$.

			For all new variables, the clauses in $C^5\cup C^6$ only contain negative occurrences and are at the back of~$\varphi$. Therefore, we only have to check that orderedness holds in the first part of $\varphi'$. For each~$z_i$, this is trivially the case as it only occurs once (as a positive literal) outside of $C^5\cup C^6$.
			For the others that we denoted by $x_i,x'_i,$ and $y_i$, the positive occurrences in $C^{p_\ell,j}\cup C^{\neg p_\ell,j}$ for some~$\ell$ and $j\in\{1,2,3\}$ always precede the single negative occurrence in $C^{p_L,J}\cup C^{\neg p_L,J}$ for some $J\in\{2,3,4\}$ and $L\neq \ell$: due to the fixed ordering of the literals within a clause we have that $L>\ell$. 
	\end{proof}

	We say that a formula from \threecnf is in \fewcnf if each literal appears at most twice, and each variable appears at most thrice. We call the problem of checking whether a formula given in \fewcnf is satisfiable \fewsat. 
	\citeauthor{Tove84a} has shown that \fewsat is \NP-complete \citep[][Thm. 2.1]{Tove84a}. We follow his proof to show that this still holds for formulae in \redfewcnf where we additionally require that every variable occurs in at most one three-literal clause and every literal in at most one two-literal clause. Denote the corresponding decision problem by \redfewsat.

	\begin{lemma}
		\label{lem:redfewsat-hardness}
		\redfewsat is \NP-complete.
	\end{lemma}
	\begin{proof}
		Membership in {\NP} is obvious.
		For hardness, we reduce from \threesat. 
		Let $\varphi := \bigwedge_{i=1}^n (x_i \vee y_i \vee z_i)$ be some formula in \threecnf where no clause contains the same variable twice. 
		For every variable $v$ occurring in $\varphi$, replace each of its $L$ occurrences with a new variable $v_j$ where $1 \leq j \leq L$. 
		Now for every $v$ occurring in $\varphi$, add the clauses 
		\[
			\varphi_v = (\neg v_L \vee v_1) \wedge \bigwedge_{j=1}^{L-1} (\neg v_j \vee v_{j+1})
		\]
		which are equivalent to $(v_L \Rightarrow v_1) \wedge \bigwedge_{j=1}^{L-1} (v_j \Rightarrow v_{j+1})$. Call the formula resulting from these transformations $\redform(\varphi)$. Note that $\redform(\varphi)$
		only contains clauses with three literals (original clauses with replaced variables) or two literals (the new clauses); denote these clause sets by $C^3$ and $C^2$, respectively. Also observe that every variable occurs exactly once in~$C_3$ and every literal exactly once in~$C^2$, \ie $\redform(\varphi)$ is in \redfewcnf.

		For every old variable $v$, we can only satisfy $\varphi_v$ by setting all $v_j$ to the same value. Since setting all $v_j$ to the same value $t$ satisfies $\varphi_v$ and has the same effect on the original part of $\redform(\varphi)$ that setting $v$ to $t$ has on $\varphi$, it follows that $\varphi$ is satisfiable if and only if $\redform(\varphi)$ is satisfiable. 
	\end{proof}

	\subsection{The Banks Set} 
	\label{sub:the_banks_set}
	The \emph{Banks set} associates with each majority tournament the maximal elements of its maximal (with respect to set-inclusion) transitive subtournaments \citep[see, \eg][]{Lasl97a,BBH15a}.

	Although \emph{some} alternative in the Banks set can be found in polynomial time using a greedy algorithm \citep{Hudr04a}, deciding whether a \emph{specific} alternative belongs to the Banks set is \NP-complete as shown by \citet{Woeg03a} by a reduction from $3$-colorability. \citet{BFHM09a} gave an arguably simpler proof of this result by a reduction from \threesat: every formula~$\varphi$ in \threecnf can be transformed in polynomial time into a tournament~$T^{\ba}_\varphi$ with a decision vertex~$c_0$ such that $c_0$ is in the Banks set of~$T^{\ba}_\varphi$ if and only if~$\varphi$ is satisfiable. 
	Due to \lemref{lem:ordered-threesat-hardness}, this reduction works just as well if~$\varphi$ is assumed to be ordered.
	Again, we have~$\props$ denote the set of variables of the propositional language in which~$\varphi$ is formulated.

	A tournament $(V,E)$ is in \emph{the class $\mathcal G^{\ba}$} if it satisfies the following properties. There is an odd integer~$m$ such that,
	\[V=C\cup U_1\cup\cdots\cup U_m \text,\]
	where $C,U_1,\ldots, U_{m}$ are pairwise disjoint and $C=\set{c_0,\ldots,c_{m}}$. We have $C_i$ denote the singleton~$\set{c_i}$ and $U=\bigcup_{i=1}^m U_i$.
	If~$i$ is odd, $U_i=\set{u_i^1,u_i^2,u_i^3}$ whereas if~$i$ is even $U_i$ is a singleton~$\set{u_i}$. Let $X=\bigcup\set{U_i:\text{$i$ is odd}}$ and $Y=\bigcup\set{U_i:\text{$i$ is even}}$. 
	Intuitively, $(V,E)$ is $T^{\ba}_\varphi$ for some~$\varphi$ in ordered \threecnf with $\frac{1}{2}(m+1)$ clauses. If~$i$ is odd,~$U_i$ corresponds to a clause of~$\varphi$ and the vertices it contains represent (tokens of) literals. We assume each of these vertices~$u_i^j$ to be labeled by the literal~$\lambda(u_i^j)$ it represents.
	 For \emph{odd} $i \in \set{1,\ldots,m}$ and $ j \in \set{1,2,3}$ we define,
	\begin{align*}
		U_i^j        &=  	\set{u_i^j}\text,\\
		U_i^{p}    &=\set{u\in U_i: \lambda(u)=p}\text{, and} \\
		U_i^{\neg p}    &=\set{u\in U_i: \lambda(u)=\neg p}\text.
	\end{align*}
	Moreover, for \emph{even} $i \in \set{1,\ldots,m}$ and $ j \in \set{1,2,3}$, we let
	\[
		U_i^j=U^{p}_i=U^{\neg p}_i=\emptyset.
	\]
	Observe that $\bigcup_{\substack{p\in\props\\1\le i\le m}}(U_i^p\cup U_i^{\neg p})=X$.

	We are now in a position to define the arc set~$E$, almost as in \citet{BFHM09a}.\footnote{There is only a slight change compared to the original construction by \citet{BFHM09a}. Specifically, we now have arcs $U^1_i\times U^3_i$ instead of the other way around. It is not difficult to check that the argument of the reduction is not affected---it is irrelevant whether the crucial transitive subtournament with~$c_0$ as its maximal element may contain one, two, or three vertices from each $U_i$.} Let
		\begin{align*}
			E	=	
			&	\mathwordbox{\displaystyle\bigcup_{i<j}}{\displaystyle\bigcup_{1\le i\le m}}(C_j\times C_i) \cup  \bigcup_{i<j}\big((U_i\times U_j)\setminus \overline {E^\varphi}\big)\cup\text{}\\
			&	\mathwordbox{\displaystyle\bigcup_{1\le i\le m}}{\displaystyle\bigcup_{1\le i\le m}}\big((U^1_i\times U^2_i)\cup(U^2_i\times U^3_i)\cup(U^1_i\times U^3_i)\big)\cup\text{}\\
			& \mathwordbox{\displaystyle\bigcup_{i\neq j}}{\displaystyle\bigcup_{1\le i\le m}}(C_i\times U_j)\cup \mathwordbox{\displaystyle\bigcup_{i}}{\displaystyle\bigcup_{i<j}}(U_i\times C_i)	\cup E^\varphi
							\text,
		\end{align*}
	where
	\[
		E^\varphi=\bigcup_{\substack{p\in P\\i < j}} (U^p_j\times U^{\neg p}_i)\cup \bigcup_{\substack{p\in P\\i < j}} (U^{\neg p}_j\times U^{p}_i)\text.
	\]
	\figref{fig:banks} illustrates this type of tournament.
	The set~$E^\varphi$ is the part of the tournament $T^{\ba}_\varphi$ that depends on the input formula. The arc set 
	\[
	(E\setminus E^\varphi)\cup\overline {E^\varphi}
	\] 
	we refer to as its \emph{skeleton}.
	
	We will show that the skeleton of each tournament $T^{\ba}_\varphi$ is induced by a $3$-voter profile such that the arcs in 
	$\overline {E^\varphi}$
	all get a weight of one. At the same time, $E^\varphi$ is 
	$2$-inducible such that the weight on all arcs is two. A little reasoning and an application of \lemref{lem:2n+1-voter-suff} then gives us the desired result.

	\begin{figure}[tb]
		\centering
		\scalebox{1}{
		  \begin{tikzpicture}[scale=.7]

		  \tikzstyle{every ellipse node}=[draw,inner xsep=3.5em,inner ysep=1.2em,fill=black!15!white,draw=black!15!white]
		  \tikzstyle{every circle node}=[fill=white,draw,minimum size=1.6em,inner sep=0pt]

		  \draw[use as bounding box,draw=white] (-6,1.2) rectangle (6,14.5);

		  \draw (-1,9.75)	node[circle](c0){$c_0$}; 
		  \draw (0,11)		node[circle](c1){$c_1$} ++(0,1.5) node[circle](c3){$c_3$} ++(0,1.5) node[circle](c5){$c_5$};
		  \draw (-2,11.75)	node[circle](c2){$c_2$} ++(0,1.5) node[circle](c4){$c_4$};

		  \draw (-1,2) node[ellipse] (ellipse3){}		++(-1.5,0)	node[circle](x31){$u_5^1$} ++(1.5,0) 	node[circle](x32){$u_5^2$} ++(1.5,0) node[circle](x33){$u_5^3$} ;
		  \draw (-1,5) node[ellipse] (ellipse2){}		++(-1.5,0)	node[circle](x21){$u_3^1$} ++(1.5,0) 	node[circle](x22){$u_3^2$} ++(1.5,0) node[circle](x23){$u_3^3$} ;
		  \draw (-1,8) node[ellipse] (ellipse1){}		++(-1.5,0)	node[circle](x11){$u_1^1$} ++(1.5,0) 	node[circle](x12){$u_1^2$} ++(1.5,0) node[circle](x13){$u_1^3$} ;

		  \draw (-1,3.5)		node[circle](y2){$u_4$};	
		  \draw (-1,6.5)		node[circle](y1){$u_2$};

			\draw (7,9.75)  node{$\resizebox{1.25ex}{2.6ex}{$\}$}\;C_0$} ;
			\draw (7,11)    node{$\resizebox{1.25ex}{2.6ex}{$\}$}\;C_1$}  
			++(0,0.75) node     {$\resizebox{1.25ex}{2.6ex}{$\}$}\;C_2$} 
			++(0,0.75) node     {$\resizebox{1.25ex}{2.6ex}{$\}$}\;C_3$} 
			++(0,0.75) node     {$\resizebox{1.25ex}{2.6ex}{$\}$}\;C_4$} 
			++(0,0.75) node     {$\resizebox{1.25ex}{2.6ex}{$\}$}\;C_5$} ;

			\draw (7,8) node{\wordbox[l]{$\resizebox{1.5ex}{4ex}{$\}$}\;U_1$}{$\big\}\;U_4$}}
			++(0,-1.5) node {\wordbox[l]{$\resizebox{1.5ex}{4ex}{$\}$}\;U_2$}{$\big\}\;U_2$}} 
			++(0,-1.5) node {\wordbox[l]{$\resizebox{1.5ex}{4ex}{$\}$}\;U_3$}{$\big\}\;U_3$}} 
			++(0,-1.5) node {\wordbox[l]{$\resizebox{1.5ex}{4ex}{$\}$}\;U_4$}{$\big\}\;U_4$}} 
			++(0,-1.5) node {\wordbox[l]{$\resizebox{1.5ex}{4ex}{$\}$}\;U_5$}{$\big\}\;U_5$}} 
			;

		\draw[-latex] (x11) -- (x12);
		\draw[-latex] (x12) -- (x13);
		\draw[-latex] (x11.60) .. controls ++(40:.8) and ++(140:.8) .. (x13.120);
		\draw[-latex] (x21) -- (x22);
		\draw[-latex] (x22) -- (x23);
		\draw[-latex] (x21.60) .. controls ++(40:.8) and ++(140:.8) .. (x23.120);
		\draw[-latex] (x31) -- (x32);
		\draw[-latex] (x32) -- (x33);
		\draw[-latex] (x31.60) .. controls ++(40:.8) and ++(140:.8) .. (x33.120);

		\draw[-latex,dotted] (ellipse1.5) to [bend right=70] (c1);
		\draw[-latex,dotted] (ellipse2.5) to [bend right=70] (c3);
		\draw[-latex,dotted] (ellipse3.5) to [bend right=70] (c5);
		\draw[-latex, dotted] (y1.180) .. controls  +(-4.5,.3) and ++(-2.5,-.3) .. (c2.180);
		\draw[-latex, dotted] (y2.180) .. controls  +(-6,.3)   and ++(-6,-.3)   .. (c4.180);
		\draw[-latex,dashed] (x33) to [bend right=20] (x13);
		\draw[-latex, dashed,] (x23) -- (x13);

	\end{tikzpicture}
	}
	  \caption{A tournament $T^{\ba}_\varphi=(V,E)$ in the class $\mathcal G^{\ba}$. Omitted arcs point downwards. Moreover, $\lambda(u^3_5)=\lambda(u^3_3)=\overline\lambda(u^3_1)$. The dotted and dashed upward arcs correspond to the arc sets~$E_1$ and~$E_2$ in \thmref{thm:banks-np}, respectively. The remaining arcs, \ie all downward arcs and the arcs within the $U_i$ form an acyclic arc set and correspond to~$E_3$.}
	  \label{fig:banks}
	\end{figure}
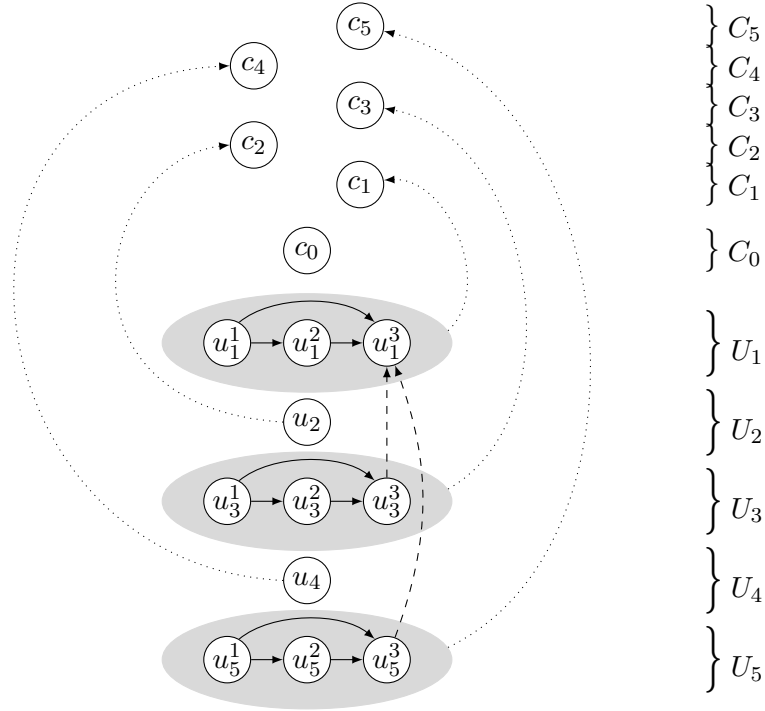

	\begin{theorem} \label{thm:banks-np}
	Computing the Banks set is \NP-hard if the number of voters is at least $5$. 
	\end{theorem}
	\begin{proof}
		Let $(V,E)$ be a tournament in~$\mathcal G^{\ba}$. It suffices to show that $(V,E)$ is induced by a $5$-voter profile. To this end define
	\begin{align*}
		E_1 &= \mathwordbox{\displaystyle\bigcup_{i}}{\displaystyle\bigcup_{\substack{p\in P\\i < j}}} (U_i \times C_i)\text,\\
		E_2 &= E^\varphi\text{, and}\\
		E_3 &= E\setminus(E_1\cup E_2)\text.
	\end{align*}
	Observe that $E=E_1\cup E_2\cup E_3$ and that $E_1$, $E_2$, and~$E_3$ are pairwise disjoint. In virtue of \lemref{lem:2n+1-voter-suff}, it therefore suffices to show that $(V,E_1)$ and $(V,E_2)$ are induced by $2$-voter profiles and that $(V,E_3)$ is acyclic.

	For $(V,E_1)$ it is easy to see that it is a union of unidirected stars and therefore $2$-inducible.
	For $(V,E_2)$, let 
	\[
	E_2^p={\displaystyle\bigcup_{i < j}} \left(U^p_j\times U^{\neg p}_i\right) \cup \left(U^{\neg p}_j \times U^p_i\right)
	\]
	be the arcs in $E_2$ associated with a variable~$p$. Note that $E_2 = \bigcup_{p\in P} E_2^p$ and that all $E_2^p$ are vertex-disjoint from each other.
	Recall that $(V,E)$ was in induced through a construction that was based on an \emph{ordered} formula. This implies that whenever $U^p_j\neq \emptyset \neq U^{\neg P}_i$ we have that $i$ is greater than~$j$. Therefore, $E_2^p$ can also be written as $\bigcup_{i,j}(U^{\neg p}_i\times U^{p}_j)$. In this representation, it is clear that $E_2^p$ is a complete, unidirected bipartite digraph. But then, $E_2$ as a vertex-disjoint union of such digraphs is a \emph{bilevel graph} and $2$-inducible according to Lemma~2 by \citet{ErMo64a}, \cf Footnote~\ref{foot:bilevel} on page \pageref{foot:bilevel}.

	To see that $E_3$ is acyclic, note that it forms a subset of 
	\[
		C \times U \cup \bigcup_{i < j} (U_i\times U_j) \cup \bigcup_{i>j} (C_i \times C_j) \cup \bigcup_{1\le i\le m}\big((U^1_i\times U^2_i)\cup(U^2_i\times U^3_i)\cup(U^1_i\times U^3_i)\big)
	\]
	and corresponds to all (shown) horizontal and (missing) downward arcs in \figref{fig:banks}.
	\end{proof}


	\subsection{The Tournament Equilibrium Set} 
	\label{sub:the_tournament_equilibrium_set}

	The \emph{tournament equilibrium set} ($\teq$) is another voting rule that, like the Banks set, selects a subset of alternatives from each tournament \citep[see, \eg][]{Lasl97a,BBH15a}. Its recursive definition is based on the notion of retentiveness. Given a tournament $(V,E)$, a subset $X\subseteq V$ is said to be \emph{$\teq$-retentive} if for all $v\in X$ all alternatives chosen by~$\teq$ from the subtournament of $(V,E)$ induced by $\set{w\in V\colon (w,v)\in E}$ are contained in~$X$. $\teq$ is then defined so as to select the union of the inclusion-minimal $\teq$-retentive sets from each tournament.

	\citet{BFHM09a} have shown that computing~$\teq$ is \NP-hard by a reduction from \threesat. By \lemref{lem:ordered-threesat-hardness}, the very same construction is also a valid reduction from \orderedthreesat.
	For every formula~$\varphi$ in ordered \threecnf a tournament~$T^{\teq}_\varphi$ 
	can be constructed such that $\teq$ selects a decision vertex~$c_0$ from $T^{\teq}_\varphi$ if and only if $\varphi$ is satisfiable. 
	The class of these tournaments~$T^{\teq}_\varphi$ is denoted by $\mathcal G^{\teq}$ and the tournaments in this class bear a strong structural similarity to those in $\mathcal G^{\ba}$, which can be exploited to show that every tournament in~$\mathcal G^{\teq}$ is induced by a $7$-voter profile.

	A tournament $(V,E)$ is in \emph{the class $\mathcal G^{\teq}$} if it satisfies the following properties. There is an odd integer~$m$ with $m\equiv 1\pmod{4}$ such that,
	\[V=C\cup U_1\cup\cdots\cup U_m \text,\]
	where $C,U_1,\ldots, U_{m}$ are defined the same as in $\mathcal G^{\ba}$. We have $C_i$ denote the singleton~$\set{c_i}$.
	Moreover, let $X=\bigcup\set{U_i:i\equiv 1\pmod 4}$, $Y=\bigcup\set{U_i:\text{$i$ is even}}$, and $Z=\bigcup\set{U_i:i\equiv 3\pmod 4}$. 

	Intuitively, $(V,E)$ is $T^{\teq}_\varphi$ for some~$\varphi$ in ordered \threecnf with $\frac{1}{4}(m+3)$ clauses. Every~$U_i\in X$ corresponds to a clause of~$\varphi$ and the vertices it contains represent (tokens of) literals. Again, we assume each of these vertices~$u_i^j$ to be labeled by the literal~$\lambda(u_i^j)$ it represents. 
	 For $i \in \set{1,5,\ldots,m}$ and $ j \in \set{1,2,3}$ we define,
	\begin{align*}
		U_i^j        &=  	\set{u_i^j}\text{,}\\
		U_i^{p}    &=\set{u\in U_i: \lambda(u)=\mathwordbox[r]{p}{\neg p}}\text{, and} \\
		U_i^{\neg p}    &=\set{u\in U_i: \lambda(u)=\neg p}\text{.}
	\end{align*}
	Moreover, for the other values of $i$, and $ j \in \set{1,2,3}$, we stipulate,
	\[
		U_i^j=U^{p}_i=U^{\neg p}_i=\emptyset\text.
	\]
	Observe that $\bigcup_{\substack{p\in\props\\1\le i\le m}}(U_i^p\cup U_i^{\neg p})=X$.

	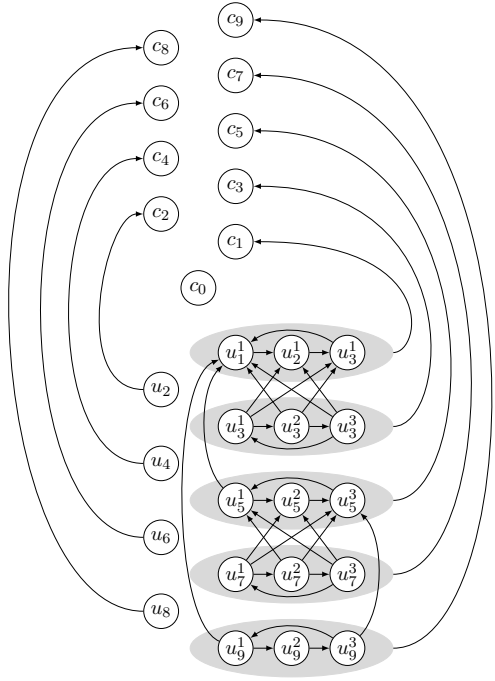
\begin{figure}[htb]
	  \centering
	\scalebox{.7}{
	  \begin{tikzpicture}[auto,scale=.7]
	  \draw[use as bounding box,draw=white] (-6.5,-1) rectangle (7.5,17.5);
	    \tikzstyle{every ellipse node}=[draw,inner xsep=3.5em,inner ysep=1em,fill=black!15!white,draw=black!15!white]
	 \tikzstyle{every circle node}=[fill=white,draw,minimum size=1.7em,inner sep=0pt]

	  \foreach \x / \y in {0/3,4/2,8/1,2/2a,6/1a}
	    \draw  (1.5,\x) node[ellipse] (ellipse\y){};
	  \draw (0,0)  		node[circle](x31){$u_9^1$} ++(1.5,0) node[circle](x32){$u_9^2$} ++(1.5,0) node[circle](x33){$u_9^3$};
	    \draw (0,2)		node[circle](z21){$u_7^1$} ++(1.5,0) node[circle](z22){$u_7^2$} ++(1.5,0) node[circle](z23){$u_7^3$};
	  \draw (0,4)		node[circle](x21){$u_5^1$} ++(1.5,0) node[circle](x22){$u_5^2$} ++(1.5,0) node[circle](x23){$u_5^3$};
	    \draw (0,6)		node[circle](z11){$u_3^1$} ++(1.5,0) node[circle](z12){$u_3^2$} ++(1.5,0) node[circle](z13){$u_3^3$};
	  \draw (0,8)		node[circle](x11){$u_1^1$} ++(1.5,0) node[circle](x12){$u_2^1$} ++(1.5,0) node[circle](x13){$u_3^1$};
	    \draw (-2,1)  	node[circle](y4){$u_8$};
	    \draw (-2,3)	node[circle](y3){$u_6$};
	  \draw (-2,5)	node[circle](y2){$u_4$};	
	    \draw (-2,7)	node[circle](y1){$u_2$};
	    \draw (-1,9.75)		node[circle](d){$c_0$};
	    \draw (0,11)		node[circle](c1){$c_1$} ++(0,1.5) node[circle](c1a){$c_3$} ++(0,1.5) node[circle](c2){$c_5$} ++(0,1.5) node[circle](c2a){$c_7$} ++(0,1.5) node[circle](c3){$c_9$};
	    \draw (-2,11.75)		node[circle](c1b){$c_2$} ++(0,1.5) node[circle](c2b){$c_4$} ++(0,1.5) node[circle](c3b){$c_6$} ++(0,1.5) node[circle](c4b){$c_8$};
	    \foreach \x in {1,2,3}
	    \draw[-latex] (x\x3) .. controls ++(-1,.7) and ++(1,.7) .. (x\x1);
	    \foreach \x in {1,2,3}
	    \draw[-latex] (x\x1) -- (x\x2);
	    \foreach \x in {1,2,3}
	    \draw[-latex] (x\x2) -- (x\x3);
	  \foreach \x in {1,2} \foreach \y / \z in {1/2,2/3}
	    \draw[-latex] (z\x\y) -- (z\x\z);
	  \foreach \x in {1,2}
	    \draw[-latex] (z\x3) .. controls ++(-1,-.7) and ++(1,-.7) .. (z\x1);
	    \draw[-latex] (y1) .. controls  +(-2.5,0) and +(-1.5,0) .. (c1b);
	    \draw[-latex] (y2) .. controls +(-3.25,0) and +(-3,0) .. (c2b);
	    \draw[-latex] (y3) .. controls +(-4,0) and +(-4.25,0) .. (c3b);
	    \draw[-latex] (y4) .. controls +(-4.75,0) and +(-5.75,0) .. (c4b);
	    \draw[-latex] (ellipse1) .. controls ++(3.5,0) and +(5.5,0) .. (c1);
	    \draw[-latex] (ellipse2) .. controls ++(5,0) and +(7,0) .. (c2);
	    \draw[-latex] (ellipse3) .. controls ++(6.5,0) and +(8.5,0) .. (c3);
	  \draw[-latex] (ellipse1a) .. controls ++(4.25,0) and +(6.25,0) .. (c1a);
	  \draw[-latex] (ellipse2a) .. controls ++(5.75,0) and +(7.75,0) .. (c2a);
	    \foreach \x in {1,2} \foreach \y / \z in {1/2,1/3,2/1,2/3,3/1,3/2}
	  \draw[-latex] (z\x\y) -- (x\x\z);
	    \draw[-latex] (x21) .. controls ++(-1,1) and ++(-1,-1) .. (x11);
	    \draw[-latex] (x33) .. controls ++(1,1) and ++(1,-1) .. (x23);
	    \draw[-latex] (x31) .. controls ++(-1.75,0.75) and ++(-1.75,-0.75) .. (x11);
	\end{tikzpicture}
	}
	  \caption{A tournament $T^{\teq}_\varphi=(V,E)$ in the class $\mathcal G^{\teq}$. Omitted arcs point downwards.}
	  \label{fig:teq}
	\end{figure}

	We are now in a position to define the arc set~$E$.
		\begin{align*}
			E	=	
			&	\mathwordbox{\displaystyle\bigcup_{i<j}		   }{\displaystyle\bigcup_{1\le i\le m}}(C_j\times C_i) \cup \bigcup_{i\neq j}(C_i\times U_j)\cup \bigcup_{i=j}(U_j\times C_i)\cup\text{}\\
			&	\mathwordbox{\displaystyle\bigcup_{1\le i\le m}}{\displaystyle\bigcup_{1\le i\le m}}\big((U^1_i\times U^2_i)\cup(U^2_i\times U^3_i)\cup(U^3_i\times U^1_i)\big)\cup\text{}\\
			&	\mathwordbox{\displaystyle\bigcup_{i<j}		   }{\displaystyle\bigcup_{1\le i\le m}}\big((U_i\times U_j)\setminus (\overline {E^\varphi}\cup \overline{E^z})\big)\cup E^\varphi \cup E^z
							\text,
		\end{align*}
	where
	\[
		E^\varphi =\bigcup_{\substack{p\in P\\i>j}} (U^p_i\times U^{\neg p}_j)\cup \bigcup_{\substack{p\in P\\i>j}} (U^{\neg p}_i\times U^{p}_j)\qquad\text{and}\qquad
		E^z = \bigcup_{\substack{l\neq l'\\i=j+2}} (U^l_i\times U^{l'}_j)\text.
	\]
	An example of such a tournament is depicted in~\figref{fig:teq}. The notable structural differences to $\mathcal{G}^{\ba}$ are the cycles in~$U_i$ for odd~$i$ and the arcs~$E^z$ between $Z$ and~$X$. Next, we show that every tournament~$T^{\teq}_\varphi$ is induced by a $7$-voter profile, using the same approach as in \thmref{thm:banks-np}.

	\begin{theorem} \label{thm:teq-np}
	Computing $\teq$ is \NP-hard if the number of voters is at least $7$.
	\end{theorem}

	\begin{proof}
		Similar to the proof for~\thmref{thm:banks-np}, it suffices to show that every tournament $(V,E)$ in~$\mathcal G^{\teq}$ is induced by a $7$-voter profile. To achieve this, we partition $E$ into four disjoint arc sets $E_1,E_2,E_3,E_4\subseteq E$ and show that the digraphs $(V,E_1)$, $(V,E_2)$, and $(V,E_3)$ are each induced by $2$-voter profiles as well as that  $(V,E_4)$ is acyclic. Then the result follows from \lemref{lem:2n+1-voter-suff}. 

		While the tournaments in~$\mathcal{G}^{\teq}$ are very similar to the ones in~$\mathcal{G}^{\ba}$, the introduction of new vertices and arcs makes finding an appealing partition a bit trickier.
		We define
		\begin{align*} 
			\displaystyle
			E_1 ={}&  \mathwordbox{\displaystyle\bigcup_{i>j}}{\displaystyle\bigcup_{\substack{i\equiv 3\\\text{mod } 4}}}
					\big(C_i \times (C_j\cup U_j)\big) \cup \mathwordbox{\displaystyle\bigcup_{i}}{\displaystyle\bigcup_{i>j}} (U_i^3 \times U_i^1) \cup 
			  {\bigcup_{\substack{i\equiv 3\\\text{mod } 4}}} \big((U_i^1\cup U_i^3)\times U_{i-2}^2\big)\text{,} \\
			E_2 ={}& \mathwordbox[r]{E^\varphi}{\bigcup_{\substack{i\equiv 3\\\text{mod } 4}}}\text{,} \\
			E_3 ={}&  \mathwordbox[r]{E^z}{\bigcup_{\substack{i\equiv 3\\\text{mod } 4}}} \setminus E_1\text{, and} \\
			E_4 ={}& \mathwordbox[r]{\mathwordbox[l]{E\setminus(E_1 \cup E_2 \cup E_3)\text{.}}{E^\varphi}}{\bigcup_{\substack{i\equiv 3\\\text{mod } 4}}}
		\end{align*}
		It can readily be appreciated that $E_1$, $E_2$, and $E_3$ are contained in~$E$ (see~\figref{fig:teq-subgraphs}). Also, they are pairwise disjoint and therefore $\set{E_1,E_2,E_3,E_4}$ is proper partition of~$E$.
	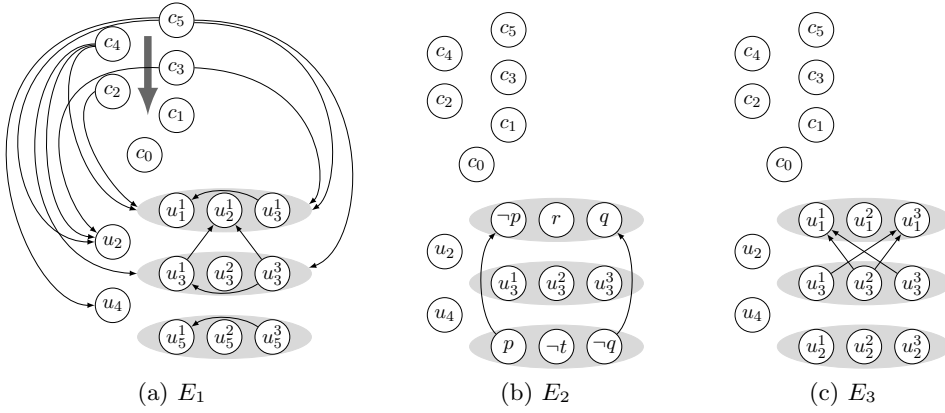
\begin{figure}[htbp]
	\centering
	\subfigureCMD{$E_1$}{}{
	\scalebox{0.7}{
	\begin{tikzpicture}[auto,scale=.6]
		\useasboundingbox (-5.3,14.5) rectangle (5,3);
		\tikzstyle{every ellipse node}=[draw,inner xsep=3em,inner ysep=.75em,fill=black!15!white,draw=black!15!white]
		\tikzstyle{every circle node}=[fill=white,draw,minimum size=1.7em,inner sep=0pt]
		\foreach \x / \y in {4/3,6/2,8/1}
			\draw  (1.5,\x) node[ellipse] (ellipse\y){};
		\draw (0,4)		node[circle](x21){$u_5^1$} ++(1.5,0) node[circle](x22){$u_5^2$} ++(1.5,0) node[circle](x23){$u_5^3$};
		\draw (0,6)		node[circle](z11){$u_3^1$} ++(1.5,0) node[circle](z12){$u_3^2$} ++(1.5,0) node[circle](z13){$u_3^3$};
		\draw (0,8)		node[circle](x11){$u_1^1$} ++(1.5,0) node[circle](x12){$u_2^1$} ++(1.5,0) node[circle](x13){$u_3^1$};
		\draw (-2,5)	node[circle](u4){$u_4$};	
		\draw (-2,7)	node[circle](u2){$u_2$};
		\draw (-1,9.75)		node[circle](d){$c_0$};
		\draw (0,11)		node[circle](c1){$c_1$} ++(0,1.5) node[circle](c3){$c_3$} ++(0,1.5) node[circle](c5){$c_5$} ++(0,1.5);
		\draw (-2,11.75)		node[circle](c2){$c_2$} ++(0,1.5) node[circle](c4){$c_4$} ++(0,1.5);
		\draw[-latex,line width=4pt,gray!80!black] (-.9,13.5) -- (-.9,11);
		\foreach \x in {1,2}
			\draw[-latex] (x\x3) .. controls ++(-1,.7) and ++(1,.7) .. (x\x1);
		\foreach \x in {1}
			\draw[-latex] (z\x3) .. controls ++(-1,-.7) and ++(1,-.7) .. (z\x1);
		\foreach \x in {1} \foreach \y / \z in {1/2,3/2}
			\draw[-latex] (z\x\y) -- (x\x\z);

		\draw[-latex,black,] (c2.180) .. controls ++(-1,-.7) and ++(-1,.7) .. (ellipse1.180-2); 

		\draw[-latex,black,] (c3.0) .. controls ++(4.2,-.2) and ++(0.6,.5) .. (ellipse1.4); 
		\draw[-latex,black,] (c3.180) .. controls ++(-4.5,0.2) and ++(-1.2,1) .. (u2.150); 	
		\draw[-latex,black,] (c4.180+8) .. controls ++(-1.5,0) and ++(-2.2,.7) .. (ellipse1.west);
		\draw[-latex,black,] (c4.180+4) .. controls ++(-2,0) and ++(-2,.2) .. (u2.165);
		\draw[-latex,black,] (c4.180) .. controls ++(-2.5,0) and ++(-4.2,.2) .. (ellipse2.180); 

		\draw[-latex,black,] (c5.0-8) .. controls ++(5,0) and ++(1,.7) .. (ellipse1.east); 
		\draw[-latex,black,] (c5.180) .. controls ++(-6,0) and ++(-3,0) .. (u2.180); 	
		\draw[-latex,black,] (c5.0) .. controls ++(6,0) and ++(2,.7) .. (ellipse2.4); 
		\draw[-latex,black,] (c5.180-8) .. controls ++(-7.5,0) and ++(-2.5,0) .. (u4.west);
	\end{tikzpicture}
	} 
	} 
	\hspace{.5cm}
	\subfigureCMD{$E_2$}{}{
	\scalebox{0.7}{
	\begin{tikzpicture}[auto,scale=.6]
		\tikzstyle{every ellipse node}=[draw,inner xsep=3em,inner ysep=.75em,fill=black!15!white,draw=black!15!white]
		\tikzstyle{every circle node}=[fill=white,draw,minimum size=1.7em,inner sep=0pt]
		\foreach \x / \y in {4/3,6/2,8/1}
			\draw  (1.5,\x) node[ellipse] (ellipse\y){};
		\draw (0,4)		node[circle](x21){$p$} ++(1.5,0) node[circle](x22){$\neg t$} ++(1.5,0) node[circle](x23){$\neg q$};
		\draw (0,6)		node[circle](z11){$u_3^1$} ++(1.5,0) node[circle](z12){$u_3^2$} ++(1.5,0) node[circle](z13){$u_3^3$};
		\draw (0,8)		node[circle](x11){$\neg p$} ++(1.5,0) node[circle](x12){$r$} ++(1.5,0) node[circle](x13){$q$};
		\draw (-2,5)	node[circle](y2){$u_4$};	
		\draw (-2,7)	node[circle](y1){$u_2$};
		\draw (-1,9.75)		node[circle](d){$c_0$};
		\draw (0,11)		node[circle](c1){$c_1$} ++(0,1.5) node[circle](c1a){$c_3$} ++(0,1.5) node[circle](c2){$c_5$} ++(0,1.5);
		\draw (-2,11.75)		node[circle](c1b){$c_2$} ++(0,1.5) node[circle](c2b){$c_4$} ++(0,1.5);
		\draw[-latex] (x23) .. controls ++(1,1) and ++(1,-1) .. (x13);
		\draw[-latex] (x21) .. controls ++(-1,1) and ++(-1,-1) .. (x11);
	\end{tikzpicture}
	} 
	} 
	\hspace{.5cm}
	\subfigureCMD{$E_3$}{}{
	\scalebox{0.7}{
	\begin{tikzpicture}[auto,scale=.6]
		\tikzstyle{every ellipse node}=[draw,inner xsep=3em,inner ysep=.75em,fill=black!15!white,draw=black!15!white]
		\tikzstyle{every circle node}=[fill=white,draw,minimum size=1.7em,inner sep=0pt]
		\foreach \x / \y in {4/3,6/2,8/1}
			\draw  (1.5,\x) node[ellipse] (ellipse\y){};
		\draw (0,4)		node[circle](x21){$u_2^1$} ++(1.5,0) node[circle](x22){$u_2^2$} ++(1.5,0) node[circle](x23){$u_2^3$};
		\draw (0,6)		node[circle](z11){$u_3^1$} ++(1.5,0) node[circle](z12){$u_3^2$} ++(1.5,0) node[circle](z13){$u_3^3$};
		\draw (0,8)		node[circle](x11){$u_1^1$} ++(1.5,0) node[circle](x12){$u_1^2$} ++(1.5,0) node[circle](x13){$u_1^3$};
		\draw (-2,5)	node[circle](y2){$u_4$};	
		\draw (-2,7)	node[circle](y1){$u_2$};
		\draw (-1,9.75)		node[circle](d){$c_0$};
		\draw (0,11)		node[circle](c1){$c_1$} ++(0,1.5) node[circle](c1a){$c_3$} ++(0,1.5) node[circle](c2){$c_5$} ++(0,1.5);
		\draw (-2,11.75)		node[circle](c1b){$c_2$} ++(0,1.5) node[circle](c2b){$c_4$} ++(0,1.5);
		\foreach \x in {1}
			\foreach \y/\z in {1/3,2/1,2/3,3/1}
				\draw[-latex] (z\x\y) to (x\x\z);
	\end{tikzpicture}
	} 
	} 
	\caption{Illustration of the arc sets $E_1,E_2,E_3\subset E$ in $T^{\teq}_\varphi=(V,E)$. The thick arrow on the left represents all arcs $\bigcup_{i<j} \{(c_j,c_i)\}$ being part of~$E_1$.}
	\label{fig:teq-subgraphs}
	\end{figure}

	To show that $(V,E_1)$ is $2$-inducible, we define
	\begin{align*}
		E_1' ={}& \mathwordbox{\displaystyle\bigcup_{\substack{i\leq j\\ U_j\subset X\cup Z}}}{\displaystyle\bigcup_{U_j\subset X\cup Z}} \left(C_i \times U_j\right) \cup \bigcup_{\substack{i\leq j\\ U_j\subset Y}} (U_j \times C_i) \cup \text{} 
		 \mathwordbox{\displaystyle\bigcup_{\substack{i<j\\U_i,U_j \subset X\cup Z}}}{\displaystyle\bigcup_{U_j\subset X\cup Z}} \left(\left(U_i \times U_j\right) \setminus E_1\right) \cup \bigcup_{\substack{i<j\\U_i,U_j\subset Y}} (U_j\times U_i) \cup \text{}\\
		&  \mathwordbox{\displaystyle\bigcup_{i\text{ odd}}}{\displaystyle\bigcup_{U_j\subset X\cup Z}} \big((U_i^1\cup U_i^3)\times U_i^2\big) \cup \left(Y\times (X\cup Z)\right)\text{.}\\
	\end{align*}

	It is straightforward to check that $E'_1$ is a reorientation of $\incompart E_1$. 
	Also, it is easy but tedious, by making the obvious case distinctions, to show for $E_1$ and $E_1'$ that
	the out-neighborhood of each vertex is contained in the out-neighborhood of each of its in-neighbors, implying that~$E_1$ and~$E_1'$ are both transitive.\footnote{In a digraph $(V,E)$, the out-neighbors of a vertex $x$ are given by $D_x=\{y\in V\colon (x,y)\in E\}$. Analogously, the in-neighbors of $x$ are defined as $\overline{D}_x=\{y\in V\colon (y,x)\in E\}$.}
	For example, consider a vertex~$u_i^1 \in X$ in $E_1'$ for which
	\[	
		D=U_i^2 \cup \bigcup_{\substack{j>i\\j\text{ odd}}} U_j \qquad \text{and} \qquad	\overline{D}=Y\cup\bigcup_{\substack{j<i\\j\text{ odd}}} U_j \cup \bigcup_{j\leq i}C_j
	\]
	denote the set of all out-neighbors and all in-neighbors of~$u_i^1$ in $(V,E_1')$, respectively. It is straightforward to check that every vertex in~$\overline{D}$ also has an arc in~$E_1'$ to every vertex in~$D$.
	
	Thus, in virtue of~\lemref{lem:2-voter-char}, $(V,E_1)$ is induced by a $2$-voter profile.

	The proof for $(V,E_2)$ being $2$-inducible is analogous to the proof of the same statement in the Banks construction (see \thmref{thm:banks-np}). This is also where the orderedness of~$\varphi$ is exploited.

	The digraph $(V,E_3)$ is obviously transitive. We also observe that it consists of isomorphic and vertex-disjoint subgraphs $(U_i\cup U_{i-2},E_{3,i})$ for $i \equiv 3 \pmod 4$ with $E_i=(U^l_i\times U^{l'}_{i-2})$ for $l\neq l'$. 
	It is sufficient to find a general transitive reorientation~$E_{3,i}'$ on such a subgraph because then every completion of $\bigcup_{i\equiv 3(\text{mod }4)}E_{3,i}'$ is a transitive reorientation of~$\incompart E_3$.
	We define 
	\begin{align*}
		E_{3,i}' ={}&
		  (U_i\times U_{i-2}^2) \cup \big((U_{i-2}^1\cup U_{i-2}^3)\times U_{i-2}^2\big) \cup \text{} 
		 \big(( U_{i-2}^3 \cup U_{i}^1 ) \times (U_{i-2}^1\cup U_{i}^3) \big) \cup\text{}\\
		& (U_i^1\times U_i^2) \cup (U_i^2 \times U_i^3).
	\end{align*}
	This subgraph set is also shown in \figref{fig:teq-E3-reorientation} and it is easy to verify that it is indeed transitive.

	\begin{figure}[htbp]
	\centering
	\begin{tikzpicture}
		\tikzstyle{every circle node}=[fill=white,draw,minimum size=1.6em,inner sep=0pt]
		\draw (0,6)		node[circle](z11){$u_3^1$} ++(1.5,0) node[circle](z12){$u_3^2$} ++(1.5,0) node[circle](z13){$u_3^3$};
		\draw (0,8)		node[circle](x11){$u_1^1$} ++(1.5,0) node[circle](x12){$u_1^2$} ++(1.5,0) node[circle](x13){$u_1^3$};
		\foreach \x in {1}
			\foreach \y/\z in {1/3,2/1,2/3,3/1}
				\draw[-,dotted] (z\x\y) to (x\x\z);
		\foreach \sourcelevel/\sourceindex/\targetlevel/\targetindex in {x/1/x/2,x/3/x/2,z/1/x/1,z/1/x/2,z/1/z/2,z/2/x/2,z/2/z/3,z/3/x/2,x/3/z/3}
			\draw[-latex] (\sourcelevel1\sourceindex) to (\targetlevel1\targetindex);
		\draw[-latex] (x13) to [bend right=25] (x11);
		\draw[-latex] (z11) to [bend right=25] (z13);
	\end{tikzpicture}
	\caption{The arc set~$E_{3,3}'$ which is part of the reorientation~$E_3'$ of~$\incompart E_3$ in the proof of \thmref{thm:teq-np}. Dotted arcs denote the incomparability subgraph of $E_{3,3}'$.}
	\label{fig:teq-E3-reorientation}
	\end{figure}
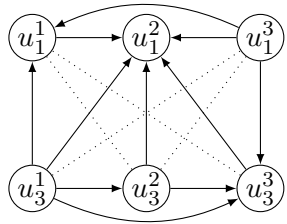

	Finally, to see the acyclicity of $(V,E_4)$, observe that 
	\[
		E_4 = \bigcup_{i<j} (C_i \times U_j) \cup \bigcup_{i<j} \left((U_i \times U_j) \setminus (\overline{E^\varphi} \cup \overline{E^z})\right) \cup
		 		 \bigcup_i \left( (U_i^1 \times U_i^2) \cup (U_i^2 \times U_i^3) \right)
	\]
	and that $E_4$ is thereby contained in the transitive closure of the ordering
	\[
		(c_0,u_1^1,u_1^2,u_1^3,c_1,u_2,c_2,u_3^1,u_3^2,u_3^3,c_3,u_4,c_4,u_5^1,\ldots,c_m)\text{.}
	\]
	This finishes the proof. 
	\end{proof}


\subsection{Slater's Rule and Kemeny's Rule} 
\label{sub:the_slater_set}
Slater's rule seeks linear rankings over alternatives that minimally conflict with the pairwise majority relation and returns the maximal elements of these rankings \citep[see, \eg][]{Lasl97a,BBH15a}. Formally, given a tournament $(V,E)$, the \emph{Slater score} of a linear ranking $\succ$ of $V$ is the number of pairs $(x,y)\in V\times V$ such that both $x\succ y$ and $(x,y)\in E$. A \emph{Slater ranking} is a ranking $\succ$ with maximum Slater score. The \emph{Slater set} consists of all those alternatives $v\in V$ that occur at the top of a Slater ranking. There is a close relationship between Slater rankings and feedback arc sets: maximizing the number of agreed-upon pairwise comparison is the same as minimizing the number of arcs of $(V,E)$ that need to be turned around so as to produce a transitive (and so acyclic) tournament. This connection makes it easy to show that computing Slater rankings is \NP-hard in general digraphs, since it is well-known that the feedback arc set problem is NP-hard \citep{GaJo79a}. To show NP-hardness of the feedback arc set problem restricted to tournaments was a long-standing open problem, which was then solved independently by \citet{Alon06a}, \citet{Coni06a}, and \citet{CTY07a}. As a consequence, we now know that computing Slater rankings and the Slater set is \NP-hard for tournaments \citep[see also][]{Hudr09b}. 

A close relative to Slater's rule is \emph{Kemeny's rule}. While Slater's rule only uses the information contained in the pairwise majority relation, Kemeny's rule also takes into account the magnitude of majority comparison, so that the input to Kemeny's rule is a \emph{weighted} majority tournament. A further difference is that, while Slater's rule is typically used so as to produce a \emph{set of winners}, Kemeny's rule is typically used to find \emph{consensus rankings}. 
Kemeny's rule has very appealing axiomatic properties \citep{YoLe78a}. Let us now formally define Kemeny's rule. Given a weighted digraph $(V,\weight)$, the \emph{Kemeny score} of a linear ranking $\succ$ of $V$ is $\sum_{x\succ y} \weight(x,y)$, and a \emph{Kemeny ranking} is a ranking with maximum Kemeny score. The Kemeny rule just returns all Kemeny rankings. Again, notice the close connection to the (weighted) feedback arc set problem. Further, notice that Kemeny's and Slater's rules coincide on tournaments where every arc has weight 1.

Let us now analyze the complexity of these two rules in a setting where there is a constant number of voters. For Kemeny's rule, \citet{DKNS01a} showed the problem to be hard even for weighted digraphs induced by a profile of 4 voters. Their reduction contained a small error that was fixed by \citet{BBD09a}. With the tools we developed in \secref{sec:bounds}, we can give a short exposition of this reduction.

\begin{theorem}
	Computing Kemeny's rule is NP-hard if the number of voters is even and at least 4.
\end{theorem}
\begin{proof}
	As we noted above, computing Kemeny's rule is equivalent to solving the feedback arc set problem, so we just need to show that this problem remains hard for digraphs inducible by 4-voter profiles. (The case for even $n > 4$ can be seen by just adding two completely reversed orders to this profile.)
	
	We show this by reduction from feedback arc set on general digraphs. Let $(V,E)$ be an instance of this problem. Now produce a new digraph $(V', E')$ from $(V,E)$ by \emph{subdividing} every arc. Thus, for each arc $(a,b)\in E$, introduce a new vertex $e_{ab}\in V'$ and arcs $(a,e_{ab}) \in E'$ and $(e_{ab}, b)\in E'$. Formally, $V' = V \cup S$, where $S = \{ e_{ab} : (a,b)\in E \}$ is the set of subdividers, and $E' = \{ (a,e_{ab}) : (a,b) \in E \} \cup \{ (e_{ab}, b) : (a,b) \in E \}$.
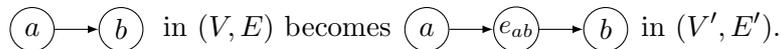
	\[ \tikz[anchor=base, baseline]{
		\node[fill=white,circle,draw,minimum size=1.5em,inner sep=0pt] (a) {$a$}; 
		\node[fill=white,circle,draw,minimum size=1.5em,inner sep=0pt, right=of a] (b) {$b$}; 
		\draw[-latex] (a) edge (b);} 
	\:\:\: \text{in $(V,E)$ becomes} \:\:
	\tikz[anchor=base, baseline]{
		\node[fill=white,circle,draw,minimum size=1.5em,inner sep=0pt] (a) {$a$}; 
		\node[fill=white,circle,draw,minimum size=1.5em,inner sep=1pt, right=of a, font=\small] (ab) {$e_{ab}$}; 
		\node[fill=white,circle,draw,minimum size=1.5em,inner sep=0pt, right=of ab] (b) {$b$}; 
		\draw[-latex] (a) edge (ab) (ab) edge (b);}
	\:\:
	\text{in $(V',E')$.} \]
	This already completes our description of the reduction. We now claim that $(V',E')$ is 4-inducible, and that the size of the minimum feedback arc set of $(V',E')$ is the same as that of the original graph $(V,E)$.
	
	To see that $(V', E')$ is 4-inducible, we partition its arcs into two arc-disjoint forests of unidirected stars. Therefore, by Lemmas~\ref{lem:stars-2-inducible}, \ref{lem:trans-disjoint-rels}, and~\ref{lem:2n-voter-suff}, we deduce that $(V',E')$ is 4-inducible. The promised partition is $E' = E_1 \cup E_2$, where
	\begin{align*}
	  E_1 &= E' \cap (V \times S), \\
	  E_2 &= E' \cap (S \times V).
	\end{align*}
	The set $E_1$ contains arcs from original vertices to subdividers, while the set $E_2$ contains arcs from subdividers to original vertices (see \figref{fig:subdivision-forest}).
		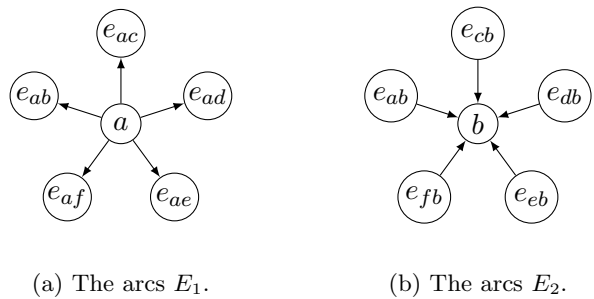
\begin{figure}[htb]
			\centering
			\subfigureCMD{The arcs $E_1$.}{fig:kemeny-4-e1}{
				\begin{tikzpicture}[auto,vertex/.style={circle,draw,inner sep=0pt, minimum size=18pt}]
					\draw[use as bounding box,opacity=0] (-1.7,-1.8) rectangle (1.7,1.5);
					\node[vertex] (1) at (162:1.2) {$e_{ab}$};
					\node[vertex] (2) at (90: 1.2) {$e_{ac}$};
					\node[vertex] (3) at (18: 1.2) {$e_{ad}$};
					\node[vertex] (4) at (306:1.2) {$e_{ae}$};
					\node[vertex] (5) at (234:1.2) {$e_{af}$};
					\node[vertex, minimum size=15pt] (c) at (0,0) {$a$};
					\draw[-latex] (c) to (1);
					\draw[-latex] (c) to (2);
					\draw[-latex] (c) to (3);
					\draw[-latex] (c) to (4);
					\draw[-latex] (c) to (5);
				\end{tikzpicture}
			}
			\hspace{2em}
			\subfigureCMD{The arcs $E_2$.}{fig:kemeny-4-e2}{
				\begin{tikzpicture}[auto,circle]
				\draw[use as bounding box,opacity=0] (-1.7,-1.8) rectangle (1.7,1.5);
				\node[vertex] (1) at (162:1.2) {$e_{ab}$};
				\node[vertex] (2) at (90: 1.2) {$e_{cb}$};
				\node[vertex] (3) at (18: 1.2) {$e_{db}$};
				\node[vertex] (4) at (306:1.2) {$e_{eb}$};
				\node[vertex] (5) at (234:1.2) {$e_{fb}$};
				\node[vertex, minimum size=15pt] (c) at (0,0) {$b$};
				\draw[latex-] (c) to (1);
				\draw[latex-] (c) to (2);
				\draw[latex-] (c) to (3);
				\draw[latex-] (c) to (4);
				\draw[latex-] (c) to (5);
				\end{tikzpicture}
			}
			\caption{Decomposition of a subdivided digraph into two forests of stars.}
			\label{fig:subdivision-forest}
		\end{figure}
		
	It is also easy to see that subdivision preserves the size of the minimal feedback arc set. If $F \subseteq E$ is a feedback arc set of $(V,E)$, then the set $F' = \{ (a, e_{ab}) : (a,b)\in F \}$ is a feedback arc set of $(V', E')$ of the same size (that is, we delete `half' of every arc of $F$). Conversely, a minimal feedback arc set $F'$ of $(V', E')$ will only ever delete one half of an original arc; deleting an arc of $(V,E)$ whenever half of it is deleted in $F'$ gives us a feedback arc set of $(V,E)$ of the same size as $F'$.
\end{proof}

It is easy to see that Kemeny's rule can be computed in polynomial time for only 2 voters (\eg both of the input votes are also optimal Kemeny rankings). Thus, the complexity of Kemeny's rule is settled for every constant \emph{even} number of voters. The complexity for a constant \emph{odd} number of voters was open. We now establish that both Kemeny's and Slater's rules are hard for 7 voters or more.

To do this, we will analyze the reduction by \citet{Coni06a} showing hardness of the feedback arc set problem restricted to tournaments. Conitzer gives a reduction from \maxsat, which asks for an assignment to the propositional variables in a Boolean formula $\varphi$ such that at least a given number $s_1$ of clauses is satisfied. 
Due to \lemref{lem:redfewsat-hardness}, we can constrain $\varphi$ to be in \redfewcnf without affecting the correctness of \citeauthor{Coni06a}'s reduction.
The reduction is based on tournaments $T^{\sla}_{\varphi}$ that admit a Slater ranking with at most $s_2$ inconsistent arcs if and only if an assignment for $\varphi$ with at least $s_1$ satisfied clauses exists, where $s_2$ depends (polynomially) on $\varphi$ and $s_1$.

Let $\mathcal{G}^{\sla}$ denote the class of all tournaments $T^{\sla}_{\varphi}$ obtained from a Boolean formula $\varphi$ in \redfewcnf according to this construction.
A  tournament $(V,E)$ is in the class $\mathcal{G}^{\sla}$ if it satisfies the following properties.
There exist integers $m, l \geq 1$, such that
\[
	V = C \cup \bigcup_{\substack{1\leq i\leq m \\ 1\leq j\leq 6}} T_i^j,
\]
where $C$ and all $T_i^j$ are pairwise disjoint and for $1\leq i \leq m$
\begin{align*}
	C 	 &= \{c_1, \ldots, c_{|C|}\},\\
	T_i^j &= \{t_i^{j,1},\ldots,t_i^{j,l}\}.
\end{align*}
Each subtournament $(T_i^j,E\cap(T_i^j\times T_i^j))$ has to be a transitive component, i.e., it is a linear order and for a vertex $v \in V\setminus T_i^j$ and vertices $v_1, v_2 \in T_i^j$, either $\{(v_1,v),(v_2,v)\}$ or $\{(v,v_1),(v,v_2)\}$ have to be in $E$. 
For our purposes, we can treat $T_i^j$ as a single vertex denoted by $t_i^j$. 
Every~$c_i$ is associated with a clause in~$\varphi$. Abusing notation, we denote this clause with~$c_i$ as well.
Every~$T_i$ corresponds to a variable~$\lambda(T_i)$ in~$\varphi$.
For notational convenience, let
\[
T^j = \bigcup_{1\leq i\leq m} t_i^j \qquad \text{and} \qquad T_i = \bigcup_{1\le j\le 6}t_i^j\text{.}
\]

For $(V,E)$ to be in $\mathcal{G}^{\sla}$, the arc set has to be of the form 
\begin{align*}
	E ={}& E_A \cup \bigcup_i \set{(t_i^1,t_i^2),(t_i^2,t_i^3),(t_i^3,t_i^1)} \cup \text{} \\
		& (T^6 \times C )\cup (C \times T^1) \cup E^\varphi
\end{align*}
where
\begin{align*}
	E_A ={}& \bigcup_{i<j} \set{(c_i,c_j)} \cup \bigcup_{i<j} (T_i \times T_j) \cup 
		\bigcup_i \bigcup_{\substack{J\in\set{4,5,6}\\1\leq j<J}} \set{(t_i^j,t_i^J)}, \text{ and}\\
	E^\varphi ={}& \set{(t_i^2,c_j), (c_j, t_i^3), (c_j, t_i^4), (t_i^5, c_j) \midd \mathwordbox[l]{\lambda(T_i)}{\lambda(T_i),\neg\lambda(T_i)}\in c_j, c_j \in C} \cup\text{}\\
				 & \set{(c_j, t_i^2), (t_i^3,c_j), (t_i^4, c_j), (c_j, t_i^5) \midd \mathwordbox[r]{\neg\lambda(T_i)}{\lambda(T_i),\neg\lambda(T_i)}\in c_j, c_j \in C} \cup\text{}\\
				 & \set{(c_j, t_i^2), (c_j, t_i^3), (t_i^4,c_j), (t_i^5, c_j) \midd \lambda(T_i),\neg\lambda(T_i)\notin c_j, c_j \in C}\text{.}
\end{align*}
So for every clause $c_j$ and every variable $\lambda(T_i)$, the $c_j$ vertex points to exactly three of the six vertices in $T_i$, but which three vertices are pointed to depends on whether and how the variable appears the clause.\footnote{The tournaments $T^{\sla}_{\varphi}$ that we have defined differ very slightly from those used in the original reduction \citep{Coni06a}. In the case that $\neg\lambda(T_i)\in c_j$, the original reduction uses arcs $(c_j, t_i^4), (t_i^5, c_j)$, while we direct them as $(t_i^4, c_j), (c_j, t_i^5)$. Conitzer's correctness proof goes through with minor changes to the sentence beginning with ``Because $d_v$ and $e_v$ always have arcs into $c_k$\dots'' \citep[][Theorem~3]{Coni06a}.}

An illustration of a tournament in $\mathcal{G}^{\sla}$ is depicted in \figref{fig:slater-cases}.

\begin{figure*}
	\centering
	\begin{tikzpicture}[vertex/.style={circle,draw,inner sep=0pt, minimum size=18pt},fixedge/.style={solid},varedge/.style={dashed}]

		\node[vertex] (ap) at (0,4) 	{$t_1^1$};
		\node[vertex] (plusp)  at (0,2.5) {$t_1^2$};
		\node[vertex] (minusp) at (0,1) {$t_1^3$};
		\node[vertex] (bp) at (2,4) 	{$t_1^4$};
		\node[vertex] (dp) at (2,2.5) 	{$t_1^5$};
		\node[vertex] (ep) at (2,1) 	{$t_1^6$}; 
		\node[vertex] (ck) at (1,-1) 	{$c_{1}$};

		\node[draw=none] (T1-label) at (1,2.7) {$T_1$};
		\node[draw=none] (x-label) at (-0.5,-1) {$\lambda(T_1)\in c_1$};
		
		\draw[-latex] (ap) to (plusp);
		\draw[-latex] (plusp) to (minusp);
		\draw[-latex] (minusp) to [bend left=25] (ap);
		\draw[-latex] (bp) to (dp);
		\draw[-latex] (bp) to [bend left=25] (ep);
		\draw[-latex] (dp) to (ep);
		
		\draw[-latex,varedge] (ck) to (minusp);
		\draw[-latex] (ck) to [bend left=40] (ap.225);
		\draw[-latex,varedge] (ck) to [bend right=40] (bp.315);
		\draw[-latex,varedge] (plusp) to (ck);
		\draw[-latex,varedge] (dp) to (ck);
		\draw[-latex] (ep) to (ck);

	\begin{scope}[shift={(4.5,0)}]
		\node[vertex] (ap) at (0,4) 	{$t_2^1$};
		\node[vertex] (plusp)  at (0,2.5) {$t_2^2$};
		\node[vertex] (minusp) at (0,1) {$t_2^3$};
		\node[vertex] (bp) at (2,4) 	{$t_2^4$};
		\node[vertex] (dp) at (2,2.5) 	{$t_2^5$};
		\node[vertex] (ep) at (2,1) 	{$t_2^6$};  
		\node[vertex] (ck) at (1,-1) 	{$c_{2}$};

		\node[draw=none] (T2-label) at (1,2.7) {$T_2$};
		\node[draw=none] (y-label) at (-0.6,-1) {$ \neg \lambda(T_2)\in c_2$};
		
		\draw[-latex] (ap) to (plusp);
		\draw[-latex] (plusp) to (minusp);
		\draw[-latex] (minusp) to [bend left=25] (ap);
		\draw[-latex] (bp) to (dp);
		\draw[-latex] (bp) to [bend left=25] (ep);
		\draw[-latex] (dp) to (ep);
		
		\draw[-latex,varedge] (minusp) to (ck);
		\draw[latex-,varedge] (dp) to (ck);
		\draw[-latex] (ep) to (ck);
		\draw[-latex] (ck) to [bend left=40] (ap.225);
		\draw[-latex,varedge] (ck) to (plusp);
		\draw[latex-,varedge] (ck) to [bend right=40] (bp.315);
	\end{scope}

	\begin{scope}[shift={(9,0)}]
		\node[vertex] (ap) at (0,4) 	{$t_3^1$};
		\node[vertex] (plusp)  at (0,2.5) {$t_3^2$};
		\node[vertex] (minusp) at (0,1) {$t_3^3$};
		\node[vertex] (bp) at (2,4) 	{$t_3^4$};
		\node[vertex] (dp) at (2,2.5) 	{$t_3^5$};
		\node[vertex] (ep) at (2,1) 	{$t_3^6$};
		\node[vertex] (ck) at (1,-1) 	{$c_{3}$};  

		\node[draw=none] (T3-label) at (1,2.7) {$T_3$};
		\node[draw=none] (z-label) at (2.5,-1) {$\begin{array}{r}  \lambda(T_3) \\ \neg \lambda(T_3)\end{array} \notin c_3 $};
		
		\draw[-latex] (ap) to (plusp);
		\draw[-latex] (plusp) to (minusp);
		\draw[-latex] (minusp) to [bend left=25] (ap);
		\draw[-latex] (bp) to (dp);
		\draw[-latex] (bp) to [bend left=25] (ep);
		\draw[-latex] (dp) to (ep);
		
		\draw[-latex] (ck) to [bend left=40] (ap.225);
		\draw[-latex,varedge] (ck) to (plusp);
		\draw[-latex,varedge] (ck) to (minusp);
		\draw[-latex,varedge] (bp.315) to [bend left=40] (ck);
		\draw[-latex,varedge] (dp) to (ck);
		\draw[-latex] (ep) to (ck);
	\end{scope}

	\draw[-latex,line width=4pt,gray] (3,-1.8) -- (9,-1.8);
	\draw[-latex,line width=4pt,gray] (3,4.7) -- (9,4.7);

	\end{tikzpicture}
	\caption{A schematic of a tournament $T_\varphi^{\sla}$ 
	to illustrate the three different cases for the arcs between $T^2 \cup T^3 \cup T^4 \cup T^5$ and $C$.
	These arcs are shown as dashed and are the only ones that depend on~$\varphi$. 
	The thick arrows below and above indicate the fixed order between and within the~$T_i^j$, and in between the~$c_i$---they stand for the following implicit, undepicted arcs: the arcs $(c_i,c_j)$ for $i<j$, the arcs in $T_i \times T_j$ for $i<j$, and the arcs in $\{t_i^1,t_i^2,t_i^3\} \times \{t_i^4,t_i^5, t_i^6\}$.
	}
	\label{fig:slater-cases}
\end{figure*}
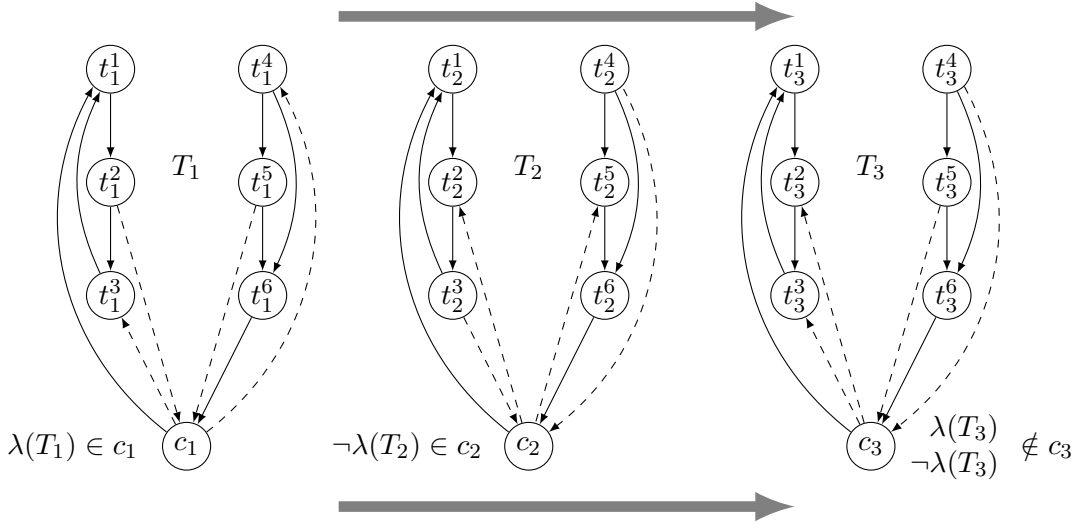

\begin{theorem}
\label{the:slater13}
The problems of computing the Slater set and of computing a Kemeny ranking are \NP-hard if the number of voters is at least $7$.
\end{theorem}
\begin{proof}
	Let $(V,E)$ be a tournament in $\mathcal{G}^{\sla}$ that is constructed from a formula~$\varphi$ in \redfewcnf whose set of clauses $C = C^2 \cup C^3$ is partitioned into clauses in $C^2$ that contain exactly two literals and clauses in $C^3$ that contain exactly three literals. We give a transitive ($1$-inducible) arc set $E_1$ and $2$-inducible arc sets $E_2, E_3, E_4$ such that putting these four profiles together yields a profile with $7$ voters that induces $(V,E)$ as its majority relation, and moreover, so that the majority margins are all equal to $1$. This gives the desired result.
	
	\begin{align*}
	E_1 &= \mathrlap{E_A\cup \bigcup_{i} \:\set{(t_i^1,t_i^2), (t_i^2, t_i^3), (t_i^1, t_i^3)} \cup \left( (T^1 \cup \cdots \cup T^6) \times (C^2 \cup C^3) \right)} \\
	E_2 &= \vphantom{\bigcup_i}\mathrlap{(C^2 \cup C^3) \times (T^1 \cup T^2 \cup T^3)} \\
	E_3 &= \bigcup_{c_j \in C^2} \set{(t_i^3,t_i^1), (t_i^2, c_j) \midd \lambda(T_i) \in c_j} & E_4 &= \bigcup_{c_j \in C^3} \set{(t_i^2, c_j) \midd \lambda(T_i) \in c_j} \\
	&\:\cup \bigcup_{c_j \in C^2} \set{(t_i^3,t_i^1), (t_i^3, c_j) \midd \neg\lambda(T_i) \in c_j} &&\: \cup \bigcup_{c_j \in C^3} \set{(t_i^3, c_j) \midd \neg\lambda(T_i) \in c_j} \\
	&\: \cup \bigcup_{c_j \in C^3} \set{(c_j, t_i^4) \midd \lambda(T_i) \in c_j} &&\: \cup \bigcup_{c_j \in C^2} \set{(c_j, t_i^4) \midd \lambda(T_i) \in c_j} \\
	&\: \cup \bigcup_{c_j \in C^3} \set{(c_j, t_i^5) \midd \neg\lambda(T_i) \in c_j} &&\: \cup \bigcup_{c_j \in C^2} \set{(c_j, t_i^5) \midd \neg\lambda(T_i) \in c_j}
	\end{align*}
	
	\begin{figure}
		\centering
		\subfigureCMD{The arcs in $E_1$ give a transitive tournament, inducible by $1$ voter.}{fig:kemeny-7-e1}{
		\scalebox{0.68}{
			\begin{tikzpicture}[vertex/.style={circle,draw,inner sep=0pt, minimum size=18pt},fixedge/.style={solid},varedge/.style={draw=black, very thick},vedge/.style={draw=black!50},posedge/.style={draw=green!40!black}, negedge/.style={red!50!black}]
			\draw[use as bounding box,opacity=0] (-1,4.5) rectangle (16.5,-1.5);
			
			\begin{scope}[shift={(0,0)}]
			\node[vertex] (ap) at (0,4) 	{$x^1$};
			\node[vertex] (plusp)  at (0,2.5) {$x^2$};
			\node[vertex] (minusp) at (0,1) {$x^3$};
			\node[vertex] (bp) at (2,4) 	{$x^4$};
			\node[vertex] (dp) at (2,2.5) 	{$x^5$};
			\node[vertex] (ep) at (2,1) 	{$x^6$}; 
			\node[vertex] (ck) at (1,-1) 	{$c_i$};

			\draw[-latex,vedge] (ap) to (plusp);
			\draw[-latex,vedge] (plusp) to (minusp);
			\draw[-latex,vedge] (ap) to [bend right=25] (minusp);
			\draw[-latex,vedge] (bp) to (dp);
			\draw[-latex,vedge] (bp) to [bend left=25] (ep);
			\draw[-latex,vedge] (dp) to (ep);
			
			\draw[-latex,vedge] (minusp) to (ck);
			\draw[-latex,vedge] (bp.315) to [bend left=40] (ck);
			\draw[-latex,vedge] (ep) to (ck);
			\draw[-latex,vedge] (ap.225) to [bend right=40] (ck);
			\draw[-latex,vedge] (plusp) to (ck);
			\draw[-latex,vedge] (dp) to [] (ck);
			\end{scope}
			
			\begin{scope}[shift={(4.5,0)}]
			\node[vertex] (ap) at (0,4) 	{$x^1$};
			\node[vertex] (plusp)  at (0,2.5) {$x^2$};
			\node[vertex] (minusp) at (0,1) {$x^3$};
			\node[vertex] (bp) at (2,4) 	{$x^4$};
			\node[vertex] (dp) at (2,2.5) 	{$x^5$};
			\node[vertex] (ep) at (2,1) 	{$x^6$}; 
			\node[vertex] (ck) at (1,-1) 	{$c_j$};

			\draw[-latex,vedge] (ap) to (plusp);
			\draw[-latex,vedge] (plusp) to (minusp);
			\draw[-latex,vedge] (ap) to [bend right=25] (minusp);
			\draw[-latex,vedge] (bp) to (dp);
			\draw[-latex,vedge] (bp) to [bend left=25] (ep);
			\draw[-latex,vedge] (dp) to (ep);
			
			\draw[-latex,vedge] (minusp) to (ck);
			\draw[-latex,vedge] (bp.315) to [bend left=40] (ck);
			\draw[-latex,vedge] (ep) to (ck);
			\draw[-latex,vedge] (ap.225) to [bend right=40] (ck);
			\draw[-latex,vedge] (plusp) to (ck);
			\draw[-latex,vedge] (dp) to [] (ck);
			\end{scope}

			\begin{scope}[shift={(9,0)}]
			\node[vertex] (ap) at (0,4) 	{$x^1$};
			\node[vertex] (plusp)  at (0,2.5) {$x^2$};
			\node[vertex] (minusp) at (0,1) {$x^3$};
			\node[vertex] (bp) at (2,4) 	{$x^4$};
			\node[vertex] (dp) at (2,2.5) 	{$x^5$};
			\node[vertex] (ep) at (2,1) 	{$x^6$}; 
			\node[vertex] (ck) at (1,-1) 	{$c_k$};

			\draw[-latex,vedge] (ap) to (plusp);
			\draw[-latex,vedge] (plusp) to (minusp);
			\draw[-latex,vedge] (ap) to [bend right=25] (minusp);
			\draw[-latex,vedge] (bp) to (dp);
			\draw[-latex,vedge] (bp) to [bend left=25] (ep);
			\draw[-latex,vedge] (dp) to (ep);
			
			\draw[-latex,vedge] (minusp) to (ck);
			\draw[-latex,vedge] (bp.315) to [bend left=40] (ck);
			\draw[-latex,vedge] (ep) to (ck);
			\draw[-latex,vedge] (ap.225) to [bend right=40] (ck);
			\draw[-latex,vedge] (plusp) to (ck);
			\draw[-latex,vedge] (dp) to [] (ck);
			\end{scope}
			
			\begin{scope}[shift={(13.5,0)}]
			\node[vertex] (ap) at (0,4) 	{$x^1$};
			\node[vertex] (plusp)  at (0,2.5) {$x^2$};
			\node[vertex] (minusp) at (0,1) {$x^3$};
			\node[vertex] (bp) at (2,4) 	{$x^4$};
			\node[vertex] (dp) at (2,2.5) 	{$x^5$};
			\node[vertex] (ep) at (2,1) 	{$x^6$}; 
			\node[vertex] (ck) at (1,-1) 	{$c_l$};
			
			\draw[-latex,vedge] (ap) to (plusp);
			\draw[-latex,vedge] (plusp) to (minusp);
			\draw[-latex,vedge] (ap) to [bend right=25] (minusp);
			\draw[-latex,vedge] (bp) to (dp);
			\draw[-latex,vedge] (bp) to [bend left=25] (ep);
			\draw[-latex,vedge] (dp) to (ep);
			
			\draw[-latex,vedge] (minusp) to (ck);
			\draw[-latex,vedge] (bp.315) to [bend left=40] (ck);
			\draw[-latex,vedge] (ep) to (ck);
			\draw[-latex,vedge] (ap.225) to [bend right=40] (ck);
			\draw[-latex,vedge] (plusp) to (ck);
			\draw[-latex,vedge] (dp) to [] (ck);
			
			\end{scope}
			\end{tikzpicture}
		}}
		\\
		\subfigureCMD{The arcs in $E_2$ form a bilevel graph, and are thus $2$-inducible.}{fig:kemeny-7-e2}{
		\scalebox{0.68}{
			\begin{tikzpicture}[vertex/.style={circle,draw,inner sep=0pt, minimum size=18pt},fixedge/.style={solid},varedge/.style={draw=black, very thick},vedge/.style={draw=black!50},posedge/.style={draw=green!40!black}, negedge/.style={red!50!black}]
			\draw[use as bounding box,opacity=0] (-1,4.5) rectangle (16.5,-1.5);
			
			\begin{scope}[shift={(0,0)}]
			\node[vertex] (ap) at (0,4) 	{$x^1$};
			\node[vertex] (plusp)  at (0,2.5) {$x^2$};
			\node[vertex] (minusp) at (0,1) {$x^3$};
			\node[vertex] (bp) at (2,4) 	{$x^4$};
			\node[vertex] (dp) at (2,2.5) 	{$x^5$};
			\node[vertex] (ep) at (2,1) 	{$x^6$}; 
			\node[vertex] (ck) at (1,-1) 	{$c_i$};

			\draw[-latex,varedge] (ck) to (minusp);
			\draw[-latex,varedge] (ck) to [bend left=40] (ap.225);
			\draw[-latex,varedge] (ck) to (plusp);
			\end{scope}
			
			\begin{scope}[shift={(4.5,0)}]
			\node[vertex] (ap) at (0,4) 	{$x^1$};
			\node[vertex] (plusp)  at (0,2.5) {$x^2$};
			\node[vertex] (minusp) at (0,1) {$x^3$};
			\node[vertex] (bp) at (2,4) 	{$x^4$};
			\node[vertex] (dp) at (2,2.5) 	{$x^5$};
			\node[vertex] (ep) at (2,1) 	{$x^6$}; 
			\node[vertex] (ck) at (1,-1) 	{$c_j$};

			\draw[-latex,varedge] (ck) to (minusp);
			\draw[-latex,varedge] (ck) to [bend left=40] (ap.225);
			\draw[-latex,varedge] (ck) to (plusp);
			\end{scope}

			\begin{scope}[shift={(9,0)}]
			\node[vertex] (ap) at (0,4) 	{$x^1$};
			\node[vertex] (plusp)  at (0,2.5) {$x^2$};
			\node[vertex] (minusp) at (0,1) {$x^3$};
			\node[vertex] (bp) at (2,4) 	{$x^4$};
			\node[vertex] (dp) at (2,2.5) 	{$x^5$};
			\node[vertex] (ep) at (2,1) 	{$x^6$}; 
			\node[vertex] (ck) at (1,-1) 	{$c_k$};

			\draw[-latex,varedge] (ck) to (minusp);
			\draw[-latex,varedge] (ck) to [bend left=40] (ap.225);
			\draw[-latex,varedge] (ck) to (plusp);
			\end{scope}
			
			\begin{scope}[shift={(13.5,0)}]
			\node[vertex] (ap) at (0,4) 	{$x^1$};
			\node[vertex] (plusp)  at (0,2.5) {$x^2$};
			\node[vertex] (minusp) at (0,1) {$x^3$};
			\node[vertex] (bp) at (2,4) 	{$x^4$};
			\node[vertex] (dp) at (2,2.5) 	{$x^5$};
			\node[vertex] (ep) at (2,1) 	{$x^6$}; 
			\node[vertex] (ck) at (1,-1) 	{$c_l$};

			\draw[-latex,varedge] (ck) to (minusp);
			\draw[-latex,varedge] (ck) to [bend left=40] (ap.225);
			\draw[-latex,varedge] (ck) to (plusp);
			\end{scope}
			\end{tikzpicture}}}
	
	\subfigureCMD{The arcs in $E_3$ form a ($2$-inducible) forest of unidirected stars, centred at $x^2$, $y^3$, and $c^3$.}{fig:kemeny-7-e3}{
		\scalebox{0.68}{
			\begin{tikzpicture}[vertex/.style={circle,draw,inner sep=0pt, minimum size=18pt},fixedge/.style={solid},varedge/.style={draw=black, very thick},vedge/.style={draw=black!50},posedge/.style={}, negedge/.style={}]
			\draw[use as bounding box,opacity=0] (-1,4.5) rectangle (16.5,-1.5);
			
			\begin{scope}[shift={(0,0)}]
			\node[vertex] (ap) at (0,4) 	{$x^1$};
			\node[vertex] (plusp)  at (0,2.5) {$x^2$};
			\node[vertex] (minusp) at (0,1) {$x^3$};
			\node[vertex] (bp) at (2,4) 	{$x^4$};
			\node[vertex] (dp) at (2,2.5) 	{$x^5$};
			\node[vertex] (ep) at (2,1) 	{$x^6$}; 
			\node[vertex] (ck) at (1,-1) 	{$c^2$};
			
			\node[draw=none, align=center, font=\small] (T1-label) at (1,3.3) {$x$ occurs \\ \emph{positively} \\ in 2-clause};
			
			\draw[-latex,varedge,posedge] (minusp) to [bend left=25] (ap);
			\draw[-latex,varedge,posedge] (plusp) to (ck);
			\end{scope}
			
			\begin{scope}[shift={(4.5,0)}]
			\node[vertex] (ap) at (0,4) 	{$y^1$};
			\node[vertex] (plusp)  at (0,2.5) {$y^2$};
			\node[vertex] (minusp) at (0,1) {$y^3$};
			\node[vertex] (bp) at (2,4) 	{$y^4$};
			\node[vertex] (dp) at (2,2.5) 	{$y^5$};
			\node[vertex] (ep) at (2,1) 	{$y^6$}; 
			\node[vertex] (ck) at (1,-1) 	{$c^2$};
			
			\node[draw=none, align=center, font=\small] (T2-label) at (1,3.3) {$x$ occurs \\ \emph{negatively} \\ in 2-clause};
			
			\draw[-latex,varedge,negedge] (minusp) to [bend left=25] (ap);
			
			\draw[-latex,varedge,negedge] (minusp) to (ck);
			\end{scope}

			\begin{scope}[shift={(9,0)}]
			\node[vertex] (ap) at (0,4) 	{$x^1$};
			\node[vertex] (plusp)  at (0,2.5) {$x^2$};
			\node[vertex] (minusp) at (0,1) {$x^3$};
			\node[vertex] (bp) at (2,4) 	{$x^4$};
			\node[vertex] (dp) at (2,2.5) 	{$x^5$};
			\node[vertex] (ep) at (2,1) 	{$x^6$}; 
			\node[vertex] (ck) at (1,-1) 	{$c^3$};
			
			\node[draw=none, align=center, font=\small] (T1-label) at (1,3.3) {$x$ occurs \\ \emph{positively} \\ in 3-clause};
			
			\draw[-latex,varedge,posedge] (ck) to [bend right=40] (bp.315);
			\end{scope}
			
			\begin{scope}[shift={(13.5,0)}]
			\node[vertex] (ap) at (0,4) 	{$y^1$};
			\node[vertex] (plusp)  at (0,2.5) {$y^2$};
			\node[vertex] (minusp) at (0,1) {$y^3$};
			\node[vertex] (bp) at (2,4) 	{$y^4$};
			\node[vertex] (dp) at (2,2.5) 	{$y^5$};
			\node[vertex] (ep) at (2,1) 	{$y^6$}; 
			\node[vertex] (ck) at (1,-1) 	{$c^3$};
			
			\node[draw=none, align=center, font=\small] (T2-label) at (1,3.3) {$x$ occurs \\ \emph{negatively} \\ in 3-clause};
			
			\draw[-latex,varedge,negedge] (ck) to [] (dp);
			\end{scope}

			\end{tikzpicture}
		}
	}
	\subfigureCMD{The arcs in $E_4$ form a ($2$-inducible) forest of unidirected stars, centred at $x^2$, $y^3$, and $c^2$.}{fig:kemeny-7-e4}{
		\scalebox{0.68}{
			\begin{tikzpicture}[vertex/.style={circle,draw,inner sep=0pt, minimum size=18pt},fixedge/.style={solid},varedge/.style={draw=black, very thick},vedge/.style={draw=black!50},posedge/.style={}, negedge/.style={}]
			\draw[use as bounding box,opacity=0] (-1,4.5) rectangle (16.5,-1.5);
			
			\begin{scope}[shift={(0,0)}]
			\node[vertex] (ap) at (0,4) 	{$x^1$};
			\node[vertex] (plusp)  at (0,2.5) {$x^2$};
			\node[vertex] (minusp) at (0,1) {$x^3$};
			\node[vertex] (bp) at (2,4) 	{$x^4$};
			\node[vertex] (dp) at (2,2.5) 	{$x^5$};
			\node[vertex] (ep) at (2,1) 	{$x^6$}; 
			\node[vertex] (ck) at (1,-1) 	{$c^2$};
			
			\node[draw=none, align=center, font=\small] (T1-label) at (1,3.3) {$x$ occurs \\ \emph{positively} \\ in 2-clause};
			
			\draw[-latex,varedge,posedge] (ck) to [bend right=40] (bp.315);
			\end{scope}
			
			\begin{scope}[shift={(4.5,0)}]
			\node[vertex] (ap) at (0,4) 	{$y^1$};
			\node[vertex] (plusp)  at (0,2.5) {$y^2$};
			\node[vertex] (minusp) at (0,1) {$y^3$};
			\node[vertex] (bp) at (2,4) 	{$y^4$};
			\node[vertex] (dp) at (2,2.5) 	{$y^5$};
			\node[vertex] (ep) at (2,1) 	{$y^6$}; 
			\node[vertex] (ck) at (1,-1) 	{$c^2$};
			
			\node[draw=none, align=center, font=\small] (T2-label) at (1,3.3) {$x$ occurs \\ \emph{negatively} \\ in 2-clause};
			
			\draw[-latex,varedge,negedge] (ck) to [] (dp);
			\end{scope}
			
			\begin{scope}[shift={(9,0)}]
			\node[vertex] (ap) at (0,4) 	{$x^1$};
			\node[vertex] (plusp)  at (0,2.5) {$x^2$};
			\node[vertex] (minusp) at (0,1) {$x^3$};
			\node[vertex] (bp) at (2,4) 	{$x^4$};
			\node[vertex] (dp) at (2,2.5) 	{$x^5$};
			\node[vertex] (ep) at (2,1) 	{$x^6$}; 
			\node[vertex] (ck) at (1,-1) 	{$c^3$};
			
			\node[draw=none, align=center, font=\small] (T1-label) at (1,3.3) {$x$ occurs \\ \emph{positively} \\ in 3-clause};
			
			\draw[-latex,varedge,posedge] (plusp) to (ck);
			\end{scope}
			
			\begin{scope}[shift={(13.5,0)}]
			\node[vertex] (ap) at (0,4) 	{$y^1$};
			\node[vertex] (plusp)  at (0,2.5) {$y^2$};
			\node[vertex] (minusp) at (0,1) {$y^3$};
			\node[vertex] (bp) at (2,4) 	{$y^4$};
			\node[vertex] (dp) at (2,2.5) 	{$y^5$};
			\node[vertex] (ep) at (2,1) 	{$y^6$}; 
			\node[vertex] (ck) at (1,-1) 	{$c^3$};
			
			\node[draw=none, align=center, font=\small] (T2-label) at (1,3.3) {$x$ occurs \\ \emph{negatively} \\ in 3-clause};

			\draw[-latex,varedge,negedge] (minusp) to (ck);
			\end{scope}

			\end{tikzpicture}
		}
	}
		\caption{The decomposition of the arcs of $\mathcal{G}^{\sla}$ used in the proof of \thmref{the:slater13}.}
	\end{figure}
	
	The set $E_1$ is complete and transitive: it is induced by the $1$-voter profile $R_1$ whose voter has order $t_1^1 > t_1^2 > t_1^3 > t_1^4 > t_1^5 > t_1^6 > t_2^1 > \dots > t_m^1 > \dots > t_m^6 > c_1 > \dots > c_{|C|}$. The set $E_2$ is a bilevel graph, and thus $2$-inducible (see Footnote~\ref{foot:bilevel}) by a profile $R_2$ of $2$ voters. The sets $E_3$ and $E_4$ are also $2$-inducible by profiles $R_3$ and $R_4$, since $E_2$ and $E_3$ both consist of vertex-disjoint unidirected stars (to see disjointness one appeals to the definition of \redfewcnf: every variable occurs at most once in $C^3$, and every literal occurs at most once in $C^2$).
	
	Now, notice that the sets $E_2$, $E_3$, $E_4$ are pairwise disjoint, although certain arcs (from $(C^2 \cup C^3) \times (T^2\cup T^3)$) in $E_2$ occur once in the opposite direction in $E_3$ or $E_4$. Thus, we can see that the union of the profiles $R_2$, $R_3$, and $R_4$ induces the following arcs with weight $2$, and all other arcs with weight $0$:
	\begin{align*}
	\smash{\bigcup_i} \{t_i^3, t_i^1\} \:\cup\:
	&\set{(c_j,t_i^1), (c_j, t_i^3), (c_j, t_i^4) \midd \mathwordbox[l]{\lambda(T_i)}{\lambda(T_i),\neg\lambda(T_i)}\in c_j, c_j \in C} \cup\text{}\\
	& \set{(c_j,t_i^1), (c_j, t_i^2), (c_j, t_i^5) \midd \mathwordbox[r]{\neg\lambda(T_i)}{\lambda(T_i),\neg\lambda(T_i)}\in c_j, c_j \in C} \cup\text{}\\
	& \set{(c_j,t_i^1), (c_j, t_i^2), (c_j, t_i^3) \midd \lambda(T_i),\neg\lambda(T_i)\notin c_j, c_j \in C}\text{.}
	\end{align*}	
	
	But this is precisely the set of arcs on which the intended arc set $E$ and the transitive set $E_1$ disagree! Hence adding $R_2$, $R_3$, and $R_4$ to the $1$-voter profile $R_1$ precisely fixes these disagreements, and induces our target tournament $(V,E)$. Also notice that all majority margins have weight $1$. The profile constructed has $7$ voters, as required.
\end{proof}
	

\subsection{Ranked Pairs} 
\label{sub:ranked_pairs}
	The fourth and final voting rule we investigate is called \emph{ranked pairs ($\rp$)} \citep[see, \eg][]{FHN15a}. In contrast to the other rules we discussed in the previous sections, it operates on \emph{weighted} digraphs and, just like Slater's rule, does not require the digraph to be complete. Hence, we have separate results for an odd and for an even number of voters.

	There are two versions of $\rp$ commonly discussed in the literature. The one we are concerned with is the \emph{neutral} one, \ie the one that does not differentiate among alternatives. Deciding whether a given alternative is a winner according to this version of $\rp$ is \NP-complete \citep{BrFi12a}.

	Usually, $\rp$ is regarded as a procedure. First, one defines priorities for all pairs of alternatives, and then ranks the alternatives iteratively in the order of their priority. The priority over pairs $(a,b)$ of alternatives is given by the number of voters who prefer $a$ to $b$. To avoid creating cycles, any pair 
	whose addition would yield a cycle is discarded in the procedure. The neutral version of $\rp$, which was defined by \citet{Tide87a} and considered by \citet{BrFi12a}, returns the set of all rankings the above procedure returns for \emph{some} tie breaking rule. From this point on, we refer to this variant by $\rp$.

	The \NP-hardness proof by \citet{BrFi12a} is by a reduction from \sat. For each Boolean formula~$\varphi$ in \cnf they constructed a weighted digraph~$G^{\rp}_\varphi$ such that a decision vertex~$d$ is selected by~$\rp$ from $G^{\rp}_\varphi$ if and only if~$\varphi$ is satisfiable. The construction, of course, works just as well for a reduction from \threesat. We may also assume that in every formula~$\varphi$ in \threecnf no variable occurs more than once in each clause. 

	Since the original construction by \citet{BrFi12a} does not yield a tournament, investigating it would give only results involving an even number of voters. However, a minor modification of the argument results in a tournament, which allows to consider the case of an odd number of voters. 
	We first define the class $\mathcal G^{\rp}$ in which the weighted digraphs $G^{\rp}_\varphi$ for formulas~$\varphi$ in \threecnf are contained. Then we prove that every digraph in this class is induced by an $8$-voter profile, showing that deciding whether a given alternative is a ranked pairs winner is already \NP-complete for eight voters. Later, we define the tournament class~$\mathcal T^{\rp}$ and show the same result for an odd number of voters. Finally we combine these two results into a corollary.

	A weighted digraph~$(V,E)$ (with weight function~$\weight$) belongs to $\mathcal G^{\rp}$ if and only if it satisfies the following conditions. There are some integers $m,l\ge 1$ such that
	\[
		V=D\cup U_1\cup\dots\cup U_m\cup X_1\cup \dots \cup X_l\text,
	\]
	where, for $1\le i\le m$ and $1\le j\le l$, 
	\begin{align*}
		D	&= \set d, \\
		U_i &= \set{u_i^1,u_i^2,u_i^3,u_i^4}, \text{ and }\\
		X_j	&= \set{x_j}\text.
	\end{align*}
	If $(V,E)$ is obtained as the digraph $G^{\rp}_\varphi$ for some~$\varphi$ in \threecnf,~$l$ is the number of clauses,~$m$ the number of variables occurring in~$\varphi$, the $U_i$s are the variable gadgets, the $X_j$s the clause gadgets, and, finally, $D$ the decision vertex.
	Let $U^j_i=\set{u^j_i}$, $U^j=\bigcup_{i=1}^m\set{u^j_i}$, $U=\bigcup_{i=1}^mU_i$ and $X=\bigcup_{i=1}^l X_i$.
	Moreover, $E=E^\sigma\cup E^\varphi$, where $E^\sigma$ (the \emph{skeleton}) and $E^\varphi$ (the \emph{formula dependent part}) are disjoint  such that
	\[
			E^\sigma = 
				(D\times(U^1\cup U^3)) \cup(X\times D) \cup
			 	\bigcup_{i=1}^m\big\{(u_i^1, u_i^2), (u_i^2, u_i^3), (u_i^3, u_i^4), (u_i^4, u_i^1)\big\}
	\]
	and $E^\varphi$ is such that for all $1\le i\le m$ and all $1\le j\le l$:
	\begin{align*}
		\mathwordbox{E^\varphi\subset (U^2\cup U^4)\times X\text,}{E^\varphi\cap(U^2\cup U^4)\times X_j)|}\\            
		|E^\varphi\cap(U^2\cup U^4)\times X_j)|& \le 3 \text{, and}           \\  
		|E^\varphi\cap(U^2_i\cup U^4_i)\times X_j)|& \le  1\text,
	\end{align*}
	\ie every vertex in~$X$ has at most three incoming arcs (intuitively corresponding to the literals~$x$ contains) and at most one from every $U_i$ (intuitively corresponding to the fact that no propositional variable occurs more than once in each clause).
	Finally, we check that the weight function~$\weight$ is defined such that all arcs in $E\cap((U^2\times U^3)\cup(U^4\times U^1))$ have weight~$4$ and all arcs in $E\setminus((U^2\times U^3)\cup(U^4\times U^1))$ have weight~$2$. 
	An example illustrating this definition of the class~$\mathcal G^{\rp}$ is shown in \figref{fig:rp}.

	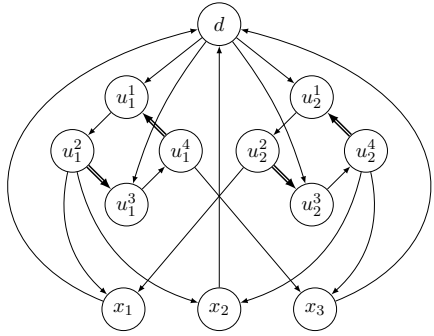
\begin{figure}[htb]
		\centering
		\scalebox{0.7}{
			\begin{tikzpicture}[scale=1.2,auto]
				\tikzstyle{every node}=[circle,draw,minimum size=2.1em,inner sep=0pt]

				\draw(.5,4) node(d){$d$};

				\draw (-1.8,2) 				  node(v1)  {$u_1^2$} 
								++(45:1.2cm)  node(v1s) {$u_1^1$} 
								++(315:1.2cm) node(v1b) {$u_1^4$} 
								++(225:1.2cm) node(v1bs){$u_1^3$};

				\draw (1.1,2) 				  node(v2b) {$u_2^2$}
				                ++(45:1.2cm)  node(v2s) {$u_2^1$}
				                ++(315:1.2cm) node(v2)  {$u_2^4$}
				                ++(225:1.2cm) node(v2bs){$u_2^3$};

				\draw[-latex,thick,double] (v1)   -- (v1bs);
				\draw[-latex]  			   (v1s)  -- (v1);
				\draw[-latex,thick,double] (v1b)  -- (v1s);
				\draw[-latex] 			   (v1bs) -- (v1b);

				\draw[latex-] 				(v2)   -- (v2bs);
				\draw[latex-,thick,double]  (v2s)  -- (v2);
				\draw[latex-] 				(v2b)  -- (v2s);
				\draw[latex-,thick,double] 	(v2bs) -- (v2b);
				
				\draw (-1,-0.5) node(y1){$x_1$} ++(1.5,0) node(y2){$x_2$} ++(1.5,0) node(y3){$x_3$};

				\draw[-latex] (v1) to [bend right] (y1);
				\draw[-latex] (v2b) -- (y1);

				\draw[-latex] (v1) to [bend right] (y2);
				\draw[-latex] (v2) to [bend left] (y2);

				\draw[-latex] (v1b) -- (y3);
				\draw[-latex] (v2) to [bend left] (y3);

				\foreach \v in {v1s,v2s}   
						{ \draw[-latex] (d) -- (\v); }

				\draw[-latex] (d) to [bend right=10] (v1bs);
				\draw[-latex] (d) to [bend left=10] (v2bs);

				\draw[-latex] (y2)  -- (d);
				\draw[-latex] (y1) .. controls (-3.5,0.5) and (-3.5,3) .. (d);
				\draw[-latex] (y3) .. controls (4.5,0.5) and (4.5,3) .. (d);	
			\end{tikzpicture}
			}
		\caption{A digraph~$(V,E)$ in the class $\mathcal G^{\rp}$. 
		Double arcs have weight~$4$ whereas normal arcs have weight~$2$.}
		\label{fig:rp}
	\end{figure}

	\begin{figure}[htb]
		\[
		\begin{array}{cccc}
		\scalebox{0.7}{
		\begin{tikzpicture}[scale=0.7,auto]
			\tikzstyle{every node}=[circle,draw,minimum size=0.7em,inner sep=0pt]

			\draw(.5,4) node(d){};

			\draw (-1.8,2) 				  node(v1)  {} 
							++(45:1.2cm)  node(v1s) {} 
							++(315:1.2cm) node(v1b) {} 
							++(225:1.2cm) node(v1bs){};
	                                                 
			\draw (1.1,2) 				  node(v2b) {}
			                ++(45:1.2cm)  node(v2s) {}
			                ++(315:1.2cm) node(v2)  {}
			                ++(225:1.2cm) node(v2bs){};

							\draw (-1,-0.5) node(y1){} ++(1.5,0) node(y2){} ++(1.5,0) node(y3){};

			\draw[-latex] (v1)   -- (v1bs);
			\draw[-latex,gray!40]  			   (v1s)  -- (v1);
			\draw[-latex] (v1b)  -- (v1s);
			\draw[-latex,gray!40] 			   (v1bs) -- (v1b);

			\draw[latex-,gray!40] 				(v2)   -- (v2bs);
			\draw[latex-]  (v2s)  -- (v2);
			\draw[latex-,gray!40] 				(v2b)  -- (v2s);
			\draw[latex-] 	(v2bs) -- (v2b);

			\draw[-latex] (v1) to [bend right] (y1);
			\draw[-latex,gray!40] (v2b) -- (y1);

			\draw[-latex] (v1) to [bend right] (y2);
			\draw[-latex,gray!40] (v2) to [bend left] (y2);

			\draw[-latex] (v1b) -- (y3);
			\draw[-latex,gray!40] (v2) to [bend left] (y3);

			\foreach \v in {v1s,v2s}   
					{ \draw[-latex,gray!40] (d) -- (\v); }

			\draw[-latex,gray!40] (d) to [bend right=10] (v1bs);
			\draw[-latex,gray!40] (d) to [bend left=10] (v2bs);

			\draw[-latex,gray!40] (y2)  -- (d);
			\draw[-latex,gray!40] (y1) .. controls (-3.5,0.5) and (-3.5,3) .. (d);
			\draw[-latex,gray!40] (y3) .. controls (4.5,0.5) and (4.5,3) .. (d);	
		\end{tikzpicture}
			}
	&
	\scalebox{0.7}{
	\begin{tikzpicture}[scale=0.7,auto]
		\tikzstyle{every node}=[circle,draw,minimum size=0.7em,inner sep=0pt]

		\draw(.5,4) node(d){};

		\draw (-1.8,2) 				  node(v1)  {} 
						++(45:1.2cm)  node(v1s) {} 
						++(315:1.2cm) node(v1b) {} 
						++(225:1.2cm) node(v1bs){};
	                                             
		\draw (1.1,2) 				  node(v2b) {}
		                ++(45:1.2cm)  node(v2s) {}
		                ++(315:1.2cm) node(v2)  {}
		                ++(225:1.2cm) node(v2bs){};

		\draw[-latex] (v1)   -- (v1bs);
		\draw[-latex,gray!40]  			   (v1s)  -- (v1);
		\draw[-latex] (v1b)  -- (v1s);
		\draw[-latex,gray!40] 			   (v1bs) -- (v1b);

		\draw[latex-,gray!40] 				(v2)   -- (v2bs);
		\draw[latex-]  (v2s)  -- (v2);
		\draw[latex-,gray!40] 				(v2b)  -- (v2s);
		\draw[latex-] 	(v2bs) -- (v2b);
		
		\draw (-1,-0.5) node(y1){} ++(1.5,0) node(y2){} ++(1.5,0) node(y3){};

		\draw[-latex,gray!40] (v1) to [bend right] (y1);
		\draw[-latex] (v2b) -- (y1);

		\draw[-latex,gray!40] (v1) to [bend right] (y2);
		\draw[-latex] (v2) to [bend left] (y2);

		\draw[-latex,gray!40] (v1b) -- (y3);
		\draw[-latex] (v2) to [bend left] (y3);

		\foreach \v in {v1s,v2s}   
				{ \draw[-latex,gray!40] (d) -- (\v); }

		\draw[-latex,gray!40] (d) to [bend right=10] (v1bs);
		\draw[-latex,gray!40] (d) to [bend left=10] (v2bs);

		\draw[-latex,gray!40] (y2)  -- (d);
		\draw[-latex,gray!40] (y1) .. controls (-3.5,0.5) and (-3.5,3) .. (d);
		\draw[-latex,gray!40] (y3) .. controls (4.5,0.5) and (4.5,3) .. (d);	
	\end{tikzpicture}
		}
	\\[2ex]
	E_1	& E_2 
	\\[3ex]
		\scalebox{0.7}{
		\begin{tikzpicture}[scale=0.7,auto]
			\tikzstyle{every node}=[circle,draw,minimum size=0.7em,inner sep=0pt]

			\draw(.5,4) node(d){};

			\draw (-1.8,2) 				  node(v1)  {} 
							++(45:1.2cm)  node(v1s) {} 
							++(315:1.2cm) node(v1b) {} 
							++(225:1.2cm) node(v1bs){};
	                                                 
			\draw (1.1,2) 				  node(v2b) {}
			                ++(45:1.2cm)  node(v2s) {}
			                ++(315:1.2cm) node(v2)  {}
			                ++(225:1.2cm) node(v2bs){};

			\draw (-1,-0.5) node(y1){} ++(1.5,0) node(y2){} ++(1.5,0) node(y3){};

				\draw[-latex,gray!40] (v1)   -- (v1bs);
				\draw[-latex,gray!40]  			   (v1s)  -- (v1);
				\draw[-latex,gray!40] (v1b)  -- (v1s);
				\draw[-latex,gray!40] 			   (v1bs) -- (v1b);

				\draw[latex-,gray!40] 				(v2)   -- (v2bs);
				\draw[latex-,gray!40]  (v2s)  -- (v2);
				\draw[latex-,gray!40] 				(v2b)  -- (v2s);
				\draw[latex-,gray!40] 	(v2bs) -- (v2b);

				\draw[-latex,gray!40] (v1) to [bend right] (y1);
				\draw[-latex,gray!40] (v2b) -- (y1);
				\draw[-latex,gray!40] (v1) to [bend right] (y2);
				\draw[-latex,gray!40] (v2) to [bend left] (y2);
				\draw[-latex,gray!40] (v1b) -- (y3);
				\draw[-latex,gray!40] (v2) to [bend left] (y3);

				\foreach \v in {v1s,v2s}   
						{ \draw[-latex] (d) -- (\v); }

				\draw[-latex] (d) to [bend right=10] (v1bs);
				\draw[-latex] (d) to [bend left=10] (v2bs);

				\draw[-latex,gray!40] (y2)  -- (d);
				\draw[-latex,gray!40] (y1) .. controls (-3.5,0.5) and (-3.5,3) .. (d);
				\draw[-latex,gray!40] (y3) .. controls (4.5,0.5) and (4.5,3) .. (d);
			\end{tikzpicture}
			}
	&
	\scalebox{0.7}{
	\begin{tikzpicture}[scale=0.7,auto]
		\tikzstyle{every node}=[circle,draw,minimum size=0.7em,inner sep=0pt]

		\draw(.5,4) node(d){};

		\draw (-1.8,2) 				  node(v1)  {} 
						++(45:1.2cm)  node(v1s) {} 
						++(315:1.2cm) node(v1b) {} 
						++(225:1.2cm) node(v1bs){};
	                                             
		\draw (1.1,2) 				  node(v2b) {}
		                ++(45:1.2cm)  node(v2s) {}
		                ++(315:1.2cm) node(v2)  {}
		                ++(225:1.2cm) node(v2bs){};

		\draw (-1,-0.5) node(y1){} ++(1.5,0) node(y2){} ++(1.5,0) node(y3){};

			\draw[-latex,gray!40] (v1)   -- (v1bs);
			\draw[-latex]  			   (v1s)  -- (v1);
			\draw[-latex,gray!40] (v1b)  -- (v1s);
			\draw[-latex] 			   (v1bs) -- (v1b);

			\draw[latex-] 				(v2)   -- (v2bs);
			\draw[latex-,gray!40]  (v2s)  -- (v2);
			\draw[latex-] 				(v2b)  -- (v2s);
			\draw[latex-,gray!40] 	(v2bs) -- (v2b);

			\draw[-latex,gray!40] (v1) to [bend right] (y1);
			\draw[-latex,gray!40] (v2b) -- (y1);
			\draw[-latex,gray!40] (v1) to [bend right] (y2);
			\draw[-latex,gray!40] (v2) to [bend left] (y2);
			\draw[-latex,gray!40] (v1b) -- (y3);
			\draw[-latex,gray!40] (v2) to [bend left] (y3);

			\foreach \v in {v1s,v2s}   
					{ \draw[-latex,gray!40] (d) -- (\v); }

			\draw[-latex,gray!40] (d) to [bend right=10] (v1bs);
			\draw[-latex,gray!40] (d) to [bend left=10] (v2bs);

			\draw[-latex] (y2)  -- (d);
			\draw[-latex] (y1) .. controls (-3.5,0.5) and (-3.5,3) .. (d);
			\draw[-latex] (y3) .. controls (4.5,0.5) and (4.5,3) .. (d);
		\end{tikzpicture}
		}		
		\\[2ex]
		E_3&E_4
			\end{array}
				\]
		\caption{The sets $E_1$, $E_2$, $E_3$, and $E_4$ for the digraph of \figref{fig:rp} as defined in the proof of \thmref{thm:RP-NP-even}.}
		\label{fig:rp_decomposed}
	\end{figure}
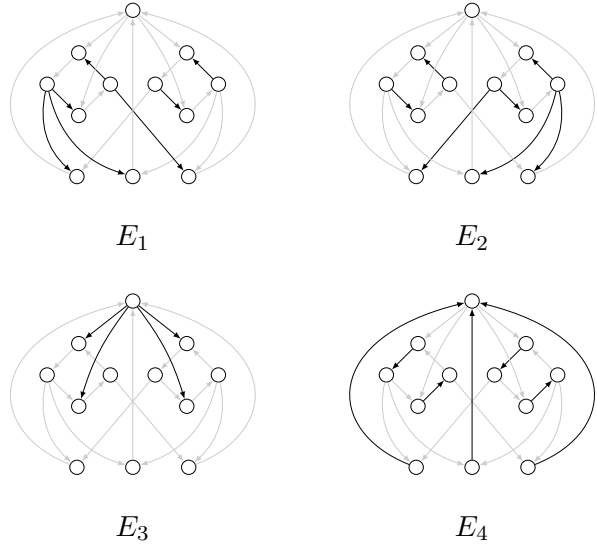

	Since $G^{RP}_\varphi$ is incomplete, it can only be induced by a profile involving an even number of voters. In fact, we will prove that only eight voters suffice to induce any digraph in~$\mathcal G^{\rp}$.

	\begin{theorem} \label{thm:RP-NP-even}
	Deciding whether a given alternative is a ranked pairs winner is \NP-complete if the number of voters is even and at least $8$.
	\end{theorem}

	\begin{proof}
	Membership in {\NP} follows from the fact that it is easy to verify whether a given ranking can be the outcome of the $\rp$ procedure, independently of the number of voters.

	For hardness, let $(V,E)$ be a digraph (with weight function~$\weight$) in~$\mathcal G^{\rp}$.
	Intuitively, $(V,E)=G^{\rp}_\varphi$ for some formula~$\varphi$ in \threecnf. It suffices to show that $(V,E)$ is induced by an $8$-voter profile. As an auxiliary notion, let for each $1\le j\le l$,
		\[
			E^\varphi\cap((U^2\cup U^4)\times X_j)=E^\varphi_{j,1}\cup E^\varphi_{j,2}\cup E^{\varphi}_{j,3}\text,
		\]
	where $|E^\varphi_{j,i}|\le 1$ for all $1\le i\le 3$. Intuitively, $E^\varphi_{j,1}$, $E^\varphi_{j,2}$, and $E^{\varphi}_{j,3}$ impose an ordering on the incoming arcs of vertex~$x_j$. Also set 
	\[
		E^\varphi_i=\bigcup_{j=1}^l E^\varphi_{j,i}
	\] 
	for each $1\le i\le 3$, \ie $E^\varphi_i$ collects the $i$-th incoming arcs of the vertices in~$X$.
	Now define the following arc sets.
	\begin{align*}
	   E_1 & =  E^\varphi_{1}\cup \bigcup_{i=1}^m\big((U_i^2\times U_i^3)\cup (U_i^4\times U_i^1)\big)\text,\\
	   E_2 & = E^\varphi_{2}\cup  \bigcup_{i=1}^m\big((U_i^2\times U_i^3)\cup (U_i^4\times U_i^1)\big)\text,\\
	   E_3 & = E^\varphi_{3}\cup (D\times(U^1\cup U^3))\text{, and}\\
	   E_4 & = (X\times D)\cup \bigcup_{i=1}^m\big((U_i^1\times U_i^2)\cup (U_i^3\times U_i^4)\big)\text.
	\end{align*}
	Observe that $E=E_1\cup E_2\cup E_3\cup E_4$ (see~\figref{fig:rp_decomposed}).
	Moreover, each of $(V,E_1)$, $(V,E_2)$, $(V,E_3)$, and $(V,E_4)$ is a vertex-disjoint union of unidirected stars. 
	Hence, by 
	\lemref{lem:trans-disjoint-rels} we may assume they are induced by the $2$-voter profiles $(R^1_1,R^1_2)$, $(R^2_1,R^2_2)$, $(R^3_1,R^3_2)$, and $(R^4_1,R^4_2)$, respectively. 
	Moreover,~$E_1$, $E_2$, $E_3$, and~$E_4$ all contained in~$E$ and therefore also pairwise orientation compatible. By \lemref{lem:2n-voter-suff} it thus follows that $(V,E)$ is induced by the $8$-voter profile 
	\[
	R=(R^1_1,R^1_2,R^2_1,R^2_2,R^3_1,R^3_2,R^4_1,R^4_2)\text.
	\]

	Moreover,
	$E_1$, $E_3$, and $E_4$ as well as~$E_2$, $E_3$, and $E_4$ are pairwise disjoint whereas $E_1\cap E_2=\bigcup_{i=1}^m\big((U_i^2\times U_i^3)\cup (U_i^4\times U_i^1)\big)$. Thus, all arcs in $E\setminus\bigcup_{i=1}^m\big((U_i^2\times U_i^3)\cup (U_i^4\times U_i^1)\big)$ have weight~$2$, whereas those in $\bigcup_{i=1}^m\big((U_i^2\times U_i^3)\cup (U_i^4\times U_i^1)\big)$ have weight~$4$. We may conclude that also the digraph $(V,E)$ with its weights is induced by the $8$-voter profile~$R$.
	\end{proof}

	The original hardness construction contained arcs with weights~$2$ or~$4$ and unspecified arcs, defining a priority over the arcs. It is easy to see that increasing all weights in such a digraph by~$1$ to~$3$ and~$5$ does not change this priority. Similarly, adding arcs with weight~$1$ is not harmful as the corresponding pairs are added to the bottom of the priority, making them irrelevent to determining whether $d$ is an $\rp$ winner or not.
	Therefore, by incorporating these observations into $G^{\rp}_\varphi$, for each Boolean formula~$\varphi$ in \threecnf, we can create a weighted tournament (call it $T^{\rp}_\varphi$) from which $d$ is selected by $\rp$ if and only if $\varphi$ is satisfiable. 
	We denote the class of weighted tournaments that consist of these $T^{\rp}_\varphi$ by $\mathcal{T}^{\rp}$.

	We adopt the same notation as for $\mathcal G^{\rp}$. A weighted tournament~$(V,E')$ (with weight function $\weight'$) belongs to $\mathcal T^{\rp}$ if and only if it satisfies the following conditions. The set of alternatives can be written as
	\[
		V=D\cup U_1\cup\dots\cup U_m\cup X_1\cup \dots \cup X_l
	\]
	whereas
	the arc set $E'$ is the union of two disjoint sets $E^{'\sigma}$ (the \emph{skeleton}) and $E^{'\varphi}$ (the \emph{formula dependent part}). Assuming that $E$ is the arc set of $\mathcal G^{\rp}_{\varphi}$, then $E^{'\varphi}=E^{\varphi}$  and $E^{'\sigma}=E^{\sigma}\cup E^{'\sigma}_c$ where

	\begin{align*}
		E^{'\sigma}_c = 
		&\big(((D\times U)\cup (U\times X))\setminus E\big)\cup
		\bigcup_{i=1}^{m}\big((U_i^1\times U_i^3)\cup (U_i^2\times U_i^4)\big)\cup\\
		&\bigcup_{i<j}(U_i\times U_j) \cup \bigcup_{i<j}(X_i\times X_j)\text{.}
	\end{align*}
	$E^{'\sigma}_c$ can be equivalently described as a reorientation of $\incompart E$. 
	Moreover, we check that $\weight'$ is defined such that all arcs in $E^{'\sigma}\cap ((U^2\times U^3)\cup(U^4\times U^1))$ have weight~$5$, all arcs in $(E^{'\varphi}\cup E^{'\sigma})\setminus((U^2\times U^3)\cup(U^4\times U^1))$ have weight~$3$, and all arcs in $E^{'\sigma}_c$ have weight $1$.

	Now we can give the second result of this section.

	\begin{theorem} \label{thm:RP-NP-odd}
	Deciding whether a given alternative is a ranked pairs winner is \NP-complete if the number of voters is odd and at least $11$.
	\end{theorem}

	\begin{proof}
		The proof here is similar to that of the previous theorem. Let $(V,E')$ be a tournament with weight function~$\weight'$ in~$\mathcal T^{\rp}$. Intuitively, $(V,E')=T^{\rp}_\varphi$ for some formula~$\varphi$ in \threecnf. It suffices to show that $(V,E')$ is induced by an $11$-voter profile. Using the notation provided in the proof of~\thmref{thm:RP-NP-even}, we define the following arc sets.
	\begin{align*}
	   E'_1 = & E^{'\varphi}_{1}\cup \bigcup_{i=1}^m\big((U_i^2\times U_i^3)\cup (U_i^4\times U_i^1)\big)\text,\\
	   E'_2 = & E^{'\varphi}_{2}\cup  \bigcup_{i=1}^m\big((U_i^2\times U_i^3)\cup (U_i^4\times U_i^1)\big)\text,\\
	   E'_3 = & E^{'\varphi}_{3}\cup (D\times(U^1\cup U^3))\text,\\
	   E'_4 = & (X\times D)\cup \bigcup_{i=1}^m\big((U_i^1\times U_i^2)\cup (U_i^3\times U_i^4)\big)\text,\\
	   E'_5 = & (X\times D)\cup \bigcup_{i=1}^m\big\{(u_i^4, u_i^1)\big\}\text{, and}\\
	   E'_6 = & (D\times U)\cup (D\times X)\cup (U\times X)\cup \bigcup_{\substack{1\le i\le m\\j<l}}(U_i^j\times U_i^l)\cup\\
	   & \bigcup_{i<j}(U_i\times U_j)\cup \bigcup_{i<j}(X_i\times X_j)\text.
	\end{align*}
	Observe that $E'_1$, $E'_2$, $E'_3$, $E'_4$, and $E'_5$ are contained in $E'$, making them pairwise orientation compatible, and that each of $(V,E'_1)$, $(V,E'_2)$, $(V,E'_3)$, $(V,E'_4)$, and $(V,E'_5)$ is a forest of stars. Therefore, in virtue of~\lemref{lem:trans-disjoint-rels} we may assume that they are induced by the $2$-voter profiles $(R^1_1,R^1_2)$, $(R^2_1,R^2_2)$, $(R^3_1,R^3_2)$, $(R^4_1,R^4_2)$,and $(R^5_1,R^5_2)$.  Moreover, it can readily be appreciated that $E'_6\supseteq E'\setminus(E'_1\cup\ldots\cup E'_5)$. As $E'_6$ defines a transitive closure for an order over all of the alternatives in $V$ (see \figref{fig:rp-order}), $(V,E'_6)$ is acyclic, and we may assume that it is induced by a voter with the preference relation $R^6=E'_6$. Thus by~\lemref{lem:2n+1-voter-suff}, $(V,E')$ is induced by the $11$-voter profile
	\[
	R=(R^1_1,R^1_2,R^2_1,R^2_2,R^3_1,R^3_2,R^4_1,R^4_2,R^5_1,R^5_2,R^6)\text.
	\]

	\begin{figure}[htb]
		\centering
		\scalebox{0.7}{
			\begin{tikzpicture}[scale=1.2,auto]
				\tikzstyle{every node}=[circle,draw,minimum size=2.1em,inner sep=0pt]

				\draw(.5,4) node(d){$d$};

				\draw (-1.8,2) 				  node(v1)  {$u_1^2$} 
								++(45:1.2cm)  node(v1s) {$u_1^1$} 
								++(315:1.2cm) node(v1b) {$u_1^4$} 
								++(225:1.2cm) node(v1bs){$u_1^3$};

				\draw (1.1,2) 				  node(v2b) {$u_2^2$}
				                ++(45:1.2cm)  node(v2s) {$u_2^1$}
				                ++(315:1.2cm) node(v2)  {$u_2^4$}
				                ++(225:1.2cm) node(v2bs){$u_2^3$};

				\draw[-latex] (v1)   -- (v1bs);
				\draw[-latex] (v1s)  -- (v1);
				\draw[-latex] (v1bs) -- (v1b);

				\draw[latex-] 				(v2)   -- (v2bs);
				\draw[latex-] 				(v2b)  -- (v2s);
				\draw[latex-] 	(v2bs) -- (v2b);
				
				\draw[latex-] (v2s) -- (v1b);
				
				\draw (-1,-0.5) node(y1){$x_1$} ++(1.5,0) node(y2){$x_2$} ++(1.5,0) node(y3){$x_3$};

				\draw[-latex] (v2) .. controls (3,0.3) and (0,0) .. (y1);

				\draw[-latex] (d) -- (v1s);
				
				\draw[-latex] (y1) -- (y2);
				\draw[-latex] (y2) -- (y3);
				
			\end{tikzpicture}
			}
		\caption{The order implied by the arc set $E'_6$ over the alternatives of a tournament $(V,E')$ in the class $\mathcal T^{\rp}$.}
		\label{fig:rp-order}
	\end{figure}
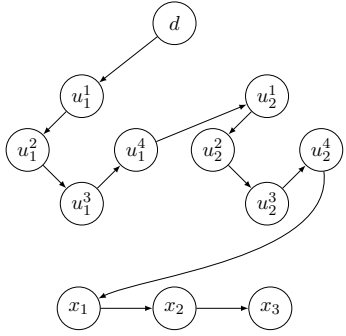

	Furthermore, note that there are some arcs in common among the arc sets and that $E'_6$ is not orientation compatible with $E'$. Arcs in $E^{'\sigma}\cap (U^2\times U^3)$ occur in $E'_1$, $E'_2$, and $E'_6$; arcs in $E^{'\sigma}\cap (U^4\times U^1)$ occur in $E'_1$, $E'_2$, and $E'_5$ while $E'_6$ includes arcs in the opposing direction or, equivalently, includes $\bigcup_{i=1}^{m}(U^1_i\times U^4_i)$; each arc in $(E^{'\varphi}\cup E^{'\sigma})\setminus((U^2\times U^3)\cup(U^4\times U^1))$ occurs in $E'_6$ and exactly one of the other arc sets; and, finally, arcs in $E^{'\sigma}_c$ occur only in $E'_6$. Thus, 
	arcs in $E^{'\sigma}\cap ((U^2\times U^3)\cup(U^4\times U^1))$ have weight~$5$, arcs in $(E^{'\varphi}\cup E^{'\sigma})\setminus((U^2\times U^3)\cup(U^4\times U^1))$ have weight~$3$, and arcs in $E^{'\sigma}_c$ have weight $1$. Therefore, we may conclude that $(V,E')$ together with its weights is induced by the $11$-voter profile $R$.
	\end{proof}

	\begin{corollary} \label{RP-NP}
		Deciding whether a given alternative is a ranked pairs winner is \NP-complete if the number of voters is either $8$ or at least $10$.
	\end{corollary}
	\begin{proof}
		This follows from Theorems~\ref{thm:RP-NP-even} and~\ref{thm:RP-NP-odd}.
	\end{proof}


\section{Conclusion and Future Work} 
\label{sec:conclusions_and_discussion}

\begin{table}[tb]
	\centering
	\begin{tabular}{lc}
		\toprule
		Voting rule & NP-hard for $n\ge$\\
		\midrule
		Banks set					& $5$ voters\\
		Tournament equilibrium set 	& $7$ voters\\
		Slater's rule 				& $7$ voters\\
		Kemeny's rule 				& $7$ voters\\
		Ranked pairs 				& $8$ voters ($n\ne 9$)\\
		\bottomrule
	\end{tabular}
	\caption{Numbers of voters for which winner determination is NP-hard. The Banks set and the tournament equilibrium set are defined for an odd number of voters only. 
	}
	\label{tbl:results}
\end{table}

	Many hardness results in computational social choice only hold if the number of voters is roughly of the same order as the number of alternatives. In some applications of voting, however, the number of voters can be much smaller than the number of alternatives and it is unclear whether hardness still holds.

	We gave complete characterizations of $2$-inducible and $3$-inducible majority digraphs, respectively, and provided sufficient conditions for $k$-inducible majority digraphs. Using an implementation based on SAT solving, we showed that majority digraphs of real-world and generated preference profiles are inducible by at most eight voters. We did not encounter a single tournament that is not $5$-inducible.\footnote{The smallest concrete tournament known not to be $5$-inducible consists of more than 600 million vertices while non-constructive arguments entail the existence of such a tournament with at most 41 vertices.}
	
	We then leveraged the sufficient conditions we obtained earlier to show that winner determination for the Banks set, the tournament equilibrium set, Slater's rule, Kemeny's rule, and ranked pairs remains hard even when there is only a small constant number of voters. This was achieved by analyzing existing hardness proofs and checking whether the class of majority digraphs used in these constructions can be induced by small constant numbers of voters. Our hardness results are summarized in \tabref{tbl:results}.

	We believe there is some very interesting potential for future work. First, it would be desirable to completely characterize the sets of digraphs inducible by four, five, or more voters. Second, the complexity of checking whether a given majority digraph is $k$-inducible is wide open for any fixed $k\ge3$. Finally, our techniques can be used to verify whether other existing hardness proofs in computational social choice remain intact for a bounded number of voters, most notably hardness shields against manipulation, bribery, and control.\footnote{Some advances in this direction have recently been made by \citet{CFNT15a}.}



\subsubsection*{Acknowledgements}
This material is based on work supported by the Deutsche Forschungsgemeinschaft under grants {BR~2312/7-2} and {BR~2312/9-1}. Additionally, the work by Paul Harrenstein was supported by the ERC under Advanced Grant 291528 (``RACE''). Dominik Peters has been supported by EPSRC and by COST Action IC1205 on Computational Social Choice.
Preliminary results of this paper were presented at the 12th International Conference on Autonomous Agents and Multi-Agent Systems (AAMAS-2013), the 1st Workshop on Exploring Beyond the Worst Case in Computational Social Choice (EXPLORE-2014), and the 6th International Workshop on Computational Social Choice (COMSOC-2016).
The authors thank Olivier Hudry and Rolf Niedermeier for pointing us to useful references.

\bibliographystyle{ACM-Reference-Format-Journals}

\begin{thebibliography}{57}
	\providecommand{\natexlab}[1]{#1}
	\providecommand{\url}[1]{\texttt{#1}}
	\expandafter\ifx\csname urlstyle\endcsname\relax
	\providecommand{\doi}[1]{doi: #1}\else
	\providecommand{\doi}{doi: \begingroup \urlstyle{rm}\Url}\fi
	
	\bibitem[Alon(2006)]{Alon06a}
	N.~Alon.
	\newblock Ranking tournaments.
	\newblock \emph{SIAM Journal on Discrete Mathematics}, 20\penalty0
	(1):\penalty0 137--142, 2006.
	
	\bibitem[Alon et~al.(2006)Alon, Brightwell, Kierstead, Kostochka, and
	Winkler]{ABKK+06a}
	N.~Alon, G.~Brightwell, H.~A. Kierstead, A.~V. Kostochka, and P.~Winkler.
	\newblock Dominating sets in $k$-majority tournaments.
	\newblock \emph{Journal of Combinatorial Theory Series B}, 96:\penalty0
	374--387, 2006.
	
	\bibitem[Austen-Smith and Banks(1999)]{AuBa00a}
	D.~Austen-Smith and J.~S. Banks.
	\newblock \emph{Positive Political Theory I: {Collective Preference}}.
	\newblock University of Michigan Press, 1999.
	
	\bibitem[Biedl et~al.(2009)Biedl, Brandenburg, and Deng]{BBD09a}
	T.~Biedl, F.~J. Brandenburg, and X.~Deng.
	\newblock On the complexity of crossings in permutations.
	\newblock \emph{Discrete Mathematics}, 309\penalty0 (7):\penalty0 1813--1823,
	2009.
	
	\bibitem[Biere(2013)]{Bier13a}
	A.~Biere.
	\newblock {L}ingeling, {P}lingeling and {T}reengeling entering the {SAT}
	competition 2013.
	\newblock In \emph{Proceedings of the SAT Competition 2013}, pages 51--52,
	2013.
	
	\bibitem[Brandt and Seedig(2013)]{BrSe13a}
	F.~Brandt and H.~G. Seedig.
	\newblock A tournament of order 24 with two disjoint {TEQ}-retentive sets.
	\newblock Technical report, http://arxiv.org/abs/1302.5592, 2013.
	
	\bibitem[Brandt and Seedig(2014)]{BrSe14a}
	F.~Brandt and H.~G. Seedig.
	\newblock On the discriminative power of tournament solutions.
	\newblock In \emph{Proceedings of the 1st AAMAS Workshop on Exploring Beyond
		the Worst Case in Computational Social Choice (EXPLORE)}, 2014.
	
	\bibitem[Brandt et~al.(2010)Brandt, Fischer, Harrenstein, and Mair]{BFHM09a}
	F.~Brandt, F.~Fischer, P.~Harrenstein, and M.~Mair.
	\newblock A computational analysis of the tournament equilibrium set.
	\newblock \emph{Social Choice and Welfare}, 34\penalty0 (4):\penalty0 597--609,
	2010.
	
	\bibitem[Brandt et~al.(2011)Brandt, Brill, and Seedig]{BBS11a}
	F.~Brandt, M.~Brill, and H.~G. Seedig.
	\newblock On the fixed-parameter tractability of composition-consistent
	tournament solutions.
	\newblock In \emph{Proceedings of the 22nd International Joint Conference on
		Artificial Intelligence (IJCAI)}, pages 85--90. AAAI Press, 2011.
	
	\bibitem[Brandt et~al.(2013)Brandt, Chudnovsky, Kim, Liu, Norin, Scott,
	Seymour, and Thomass\'{e}]{BCK+11a}
	F.~Brandt, M.~Chudnovsky, I.~Kim, G.~Liu, S.~Norin, A.~Scott, P.~Seymour, and
	S.~Thomass\'{e}.
	\newblock A counterexample to a conjecture of {S}chwartz.
	\newblock \emph{Social Choice and Welfare}, 40\penalty0 (3):\penalty0 739--743,
	2013.
	
	\bibitem[Brandt et~al.(2016{\natexlab{a}})Brandt, Brill, and
	Harrenstein]{BBH15a}
	F.~Brandt, M.~Brill, and P.~Harrenstein.
	\newblock Tournament solutions.
	\newblock In F.~Brandt, V.~Conitzer, U.~Endriss, J.~Lang, and A.~D. Procaccia,
	editors, \emph{Handbook of Computational Social Choice}, chapter~3. Cambridge
	University Press, 2016{\natexlab{a}}.
	
	\bibitem[Brandt et~al.(2016{\natexlab{b}})Brandt, Conitzer, Endriss, Lang, and
	Procaccia]{BCE+14a}
	F.~Brandt, V.~Conitzer, U.~Endriss, J.~Lang, and A.~Procaccia, editors.
	\newblock \emph{Handbook of Computational Social Choice}.
	\newblock Cambridge University Press, 2016{\natexlab{b}}.
	
	\bibitem[Brill and Fischer(2012)]{BrFi12a}
	M.~Brill and F.~Fischer.
	\newblock The price of neutrality for the ranked pairs method.
	\newblock In \emph{Proceedings of the 26th AAAI Conference on Artificial
		Intelligence (AAAI)}, pages 1299--1305. AAAI Press, 2012.
	
	\bibitem[Charbit et~al.(2007)Charbit, Thomass{\'e}, and Yeo]{CTY07a}
	P.~Charbit, S.~Thomass{\'e}, and A.~Yeo.
	\newblock The minimum feedback arc set problem is {NP}-hard for tournaments.
	\newblock \emph{Combinatorics, Probability and Computing}, 16\penalty0
	(1):\penalty0 1--4, 2007.
	
	\bibitem[Chen et~al.(2015)Chen, Faliszewski, Niedermeier, and Talmon]{CFNT15a}
	J.~Chen, P.~Faliszewski, R.~Niedermeier, and N.~Talmon.
	\newblock Elections with few voters: {Candidate} control can be easy.
	\newblock In \emph{Proceedings of the 29th AAAI Conference on Artificial
		Intelligence (AAAI)}, pages 2045--2051. AAAI Press, 2015.
	
	\bibitem[Conitzer(2006)]{Coni06a}
	V.~Conitzer.
	\newblock Computing {S}later rankings using similarities among candidates.
	\newblock In \emph{Proceedings of the 21st National Conference on Artificial
		Intelligence (AAAI)}, pages 613--619. AAAI Press, 2006.
	
	\bibitem[Conitzer et~al.(2007)Conitzer, Sandholm, and Lang]{CSL07a}
	V.~Conitzer, T.~Sandholm, and J.~Lang.
	\newblock When are elections with few candidates hard to manipulate?
	\newblock \emph{Journal of the ACM}, 54\penalty0 (3), 2007.
	
	\bibitem[Critchlow et~al.(1991)Critchlow, Fligner, and Verducci]{CFV91a}
	D.~E. Critchlow, M.~A. Fligner, and J.~S. Verducci.
	\newblock Probability models on rankings.
	\newblock \emph{Journal of Mathematical Psychology}, 35:\penalty0 294--318,
	1991.
	
	\bibitem[Dushnik and Miller(1941)]{DuMi41a}
	B.~Dushnik and E.~W. Miller.
	\newblock Partially ordered sets.
	\newblock \emph{American Journal of Mathematics}, 63\penalty0 (3):\penalty0
	600--610, 1941.
	
	\bibitem[Dwork et~al.(2001)Dwork, Kumar, Naor, and Sivakumar]{DKNS01a}
	C.~Dwork, R.~Kumar, M.~Naor, and D.~Sivakumar.
	\newblock Rank aggregation methods for the web.
	\newblock In \emph{Proceedings of the 10th International Conference on the
		World Wide Web (WWW)}, pages 613--622. ACM Press, 2001.
	
	\bibitem[Eggermont et~al.(2013)Eggermont, Hurkens, and Woeginger]{EHW13a}
	C.~Eggermont, C.~Hurkens, and G.~J. Woeginger.
	\newblock Realizing small tournaments through few permutations.
	\newblock \emph{Acta Cybernetica}, 21\penalty0 (2):\penalty0 267--271, 2013.
	
	\bibitem[Erd\H{o}s and Moser(1964)]{ErMo64a}
	P.~Erd\H{o}s and L.~Moser.
	\newblock On the representation of directed graphs as unions of orderings.
	\newblock \emph{Publications of the Mathematical Institute of the Hungarian
		Academy of Science}, 9:\penalty0 125--132, 1964.
	
	\bibitem[Faliszewski et~al.(2009)Faliszewski, Hemaspaandra, Hemaspaandra, and
	Rothe]{FHHR09a}
	P.~Faliszewski, E.~Hemaspaandra, L.~Hemaspaandra, and J.~Rothe.
	\newblock A richer understanding of the complexity of election systems.
	\newblock In S.~Ravi and S.~Shukla, editors, \emph{Fundamental Problems in
		Computing: Essays in Honor of Professor Daniel J.~Rosenkrantz}.
	Springer-Verlag, 2009.
	
	\bibitem[Fidler(2011)]{Fidl11a}
	D.~Fidler.
	\newblock A recurrence for bounds on dominating sets in $k$-majority
	tournaments.
	\newblock \emph{The Electronic Journal of Combinatorics}, 18\penalty0 (1),
	2011.
	
	\bibitem[Fiol(1992)]{Fiol92a}
	M.~A. Fiol.
	\newblock A note on the voting problem.
	\newblock \emph{Stochastica}, XIII-1:\penalty0 155--158, 1992.
	
	\bibitem[Fischer et~al.(2016)Fischer, Hudry, and Niedermeier]{FHN15a}
	F.~Fischer, O.~Hudry, and R.~Niedermeier.
	\newblock Weighted tournament solutions.
	\newblock In F.~Brandt, V.~Conitzer, U.~Endriss, J.~Lang, and A.~D. Procaccia,
	editors, \emph{Handbook of Computational Social Choice}, chapter~4. Cambridge
	University Press, 2016.
	
	\bibitem[Garey and Johnson(1979)]{GaJo79a}
	M.~R. Garey and D.~S. Johnson.
	\newblock \emph{Computers and Intractability: A Guide to the Theory of
		NP-Completeness}.
	\newblock W. H. Freeman, 1979.
	
	\bibitem[Geist and Endriss(2011)]{GeEn11a}
	C.~Geist and U.~Endriss.
	\newblock Automated search for impossibility theorems in social choice theory:
	Ranking sets of objects.
	\newblock \emph{Journal of Artificial Intelligence Research}, 40:\penalty0
	143--174, 2011.
	
	\bibitem[Graham and Spencer(1971)]{GrSp71a}
	R.~L. Graham and J.~H. Spencer.
	\newblock A constructive solution to a tournament problem.
	\newblock \emph{Canadian Mathematical Bulletin}, 14\penalty0 (1):\penalty0
	45--48, 1971.
	
	\bibitem[Hudry(2004)]{Hudr04a}
	O.~Hudry.
	\newblock A note on ``{B}anks winners in tournaments are difficult to
	recognize'' by {G}.~{J}.~{W}oeginger.
	\newblock \emph{Social Choice and Welfare}, 23:\penalty0 113--114, 2004.
	
	\bibitem[Hudry(2008)]{Hudr08a}
	O.~Hudry.
	\newblock {NP}-hardness results for the aggregation of linear orders into
	median orders.
	\newblock \emph{Annals of Operations Research}, 163\penalty0 (1):\penalty0
	63--88, 2008.
	
	\bibitem[Hudry(2010)]{Hudr09b}
	O.~Hudry.
	\newblock On the complexity of {S}later's problems.
	\newblock \emph{European Journal of Operational Research}, 203\penalty0
	(1):\penalty0 216--221, 2010.
	
	\bibitem[Karp(1972)]{Karp72a}
	R.~M. Karp.
	\newblock Reducibility among combinatorial problems.
	\newblock In R.~E. Miller and J.~W. Thatcher, editors, \emph{Complexity of
		Computer Computations}, pages 85--103. Plenem Press, 1972.
	
	\bibitem[Kemeny(1959)]{Keme59a}
	J.~G. Kemeny.
	\newblock Mathematics without numbers.
	\newblock \emph{Daedalus}, 88:\penalty0 577--591, 1959.
	
	\bibitem[Kierstead et~al.(2009)Kierstead, {Trotter Jr.}, Charbit, Milans, and
	Wenger]{KTC+09a}
	H.~A. Kierstead, W.~T. {Trotter Jr.}, P.~Charbit, K.~Milans, and P.~Wenger.
	\newblock k-majority digraphs, 2009.
	\newblock URL \url{http://www.math.uiuc.edu/{\~{}}west/regs/k-majority.html}.
	
	\bibitem[Laslier(1997)]{Lasl97a}
	J.-F. Laslier.
	\newblock \emph{Tournament Solutions and Majority Voting}.
	\newblock Springer-Verlag, 1997.
	
	\bibitem[Laslier(2010)]{Lasl10a}
	J.-F. Laslier.
	\newblock In silico voting experiments.
	\newblock In J.-F. Laslier and M.~R. Sanver, editors, \emph{Handbook on
		Approval Voting}, chapter~13, pages 311--335. Springer-Verlag, 2010.
	
	\bibitem[Mallows(1957)]{Mall57a}
	C.~L. Mallows.
	\newblock Non-null ranking models.
	\newblock \emph{Biometrika}, 44\penalty0 (1/2):\penalty0 114--130, 1957.
	
	\bibitem[Marden(1995)]{Mard95a}
	J.~I. Marden.
	\newblock \emph{Analyzing and Modeling Rank Data}.
	\newblock Number~64 in Monographs on Statistics and Applied Probability.
	Chapman \& Hall, 1995.
	
	\bibitem[Mattei and Walsh(2013)]{MaWa13a}
	N.~Mattei and T.~Walsh.
	\newblock Pref{L}ib: A library for preference data.
	\newblock In \emph{Proceedings of the 3rd International Conference on
		Algorithmic Decision Theory (ADT)}, volume 8176 of \emph{Lecture Notes in
		Computer Science (LNCS)}, pages 259--270. Springer-Verlag, 2013.
	
	\bibitem[Mattei et~al.(2012)Mattei, Forshee, and Goldsmith]{MFG12a}
	N.~Mattei, J.~Forshee, and J.~Goldsmith.
	\newblock An empirical study of voting rules and manipulation with large
	datasets.
	\newblock In \emph{Proceedings of the 4th International Workshop on
		Computational Social Choice (COMSOC)}, 2012.
	
	\bibitem[{McCabe-Dansted} and Slinko(2006)]{McSl06a}
	J.~C. {McCabe-Dansted} and A.~Slinko.
	\newblock Exploratory analysis of similarities between social choice rules.
	\newblock \emph{Group Decision and Negotiation}, 15\penalty0 (1):\penalty0
	77--107, 2006.
	
	\bibitem[{McConnell} and {de Montgolfier}(2005)]{McMo05a}
	R.~M. {McConnell} and F.~{de Montgolfier}.
	\newblock Linear-time modular decomposition of directed graphs.
	\newblock \emph{Discrete Applied Mathematics}, 145\penalty0 (2):\penalty0
	198--209, 2005.
	
	\bibitem[McGarvey(1953)]{McGa53a}
	D.~C. McGarvey.
	\newblock A theorem on the construction of voting paradoxes.
	\newblock \emph{Econometrica}, 21\penalty0 (4):\penalty0 608--610, 1953.
	
	\bibitem[{McKay} and Piperno(2013)]{McPi13a}
	B.~D. {McKay} and A.~Piperno.
	\newblock Practical graph isomorphism, {II}.
	\newblock \emph{Journal of Symbolic Computation}, 2013.
	
	\bibitem[Moon(1968)]{Moon68a}
	J.~W. Moon.
	\newblock \emph{Topics on Tournaments}.
	\newblock Holt, Reinhard and Winston, 1968.
	
	\bibitem[Ordeshook(1993)]{Orde93a}
	P.~C. Ordeshook.
	\newblock The spatial analysis of elections and committees: four decades of
	research.
	\newblock Technical report, California Institute of Technology, 1993.
	
	\bibitem[Pnueli et~al.(1971)Pnueli, Lempel, and Even]{PLE71a}
	A.~Pnueli, A.~Lempel, and S.~Even.
	\newblock Transitive orientation of graphs and identification of permutation
	graphs.
	\newblock \emph{Canadian Journal of Mathematics}, 23:\penalty0 160--175, 1971.
	
	\bibitem[Rothe(2015)]{Roth15a}
	J.~Rothe, editor.
	\newblock \emph{Economics and Computation: An Introduction to Algorithmic Game
		Theory, Computational Social Choice, and Fair Division}.
	\newblock Springer, 2015.
	
	\bibitem[Shepardson and Tovey(2009)]{ShTo09a}
	D.~Shepardson and C.~A. Tovey.
	\newblock Smallest tournament not realizable by $\frac{2}{3}$-majority voting.
	\newblock \emph{Social Choice and Welfare}, 33\penalty0 (3):\penalty0 495--503,
	2009.
	
	\bibitem[Stearns(1959)]{Stea59a}
	R.~Stearns.
	\newblock The voting problem.
	\newblock \emph{American Mathematical Monthly}, 66\penalty0 (9):\penalty0
	761--763, 1959.
	
	\bibitem[Tideman(1987)]{Tide87a}
	T.~N. Tideman.
	\newblock Independence of clones as a criterion for voting rules.
	\newblock \emph{Social Choice and Welfare}, 4\penalty0 (3):\penalty0 185--206,
	1987.
	
	\bibitem[Tovey(1984)]{Tove84a}
	C.~A. Tovey.
	\newblock A simplified {NP}-complete satisfiability problem.
	\newblock \emph{Discrete Applied Mathematics}, 8\penalty0 (1):\penalty0 85--89,
	1984.
	
	\bibitem[Tseitin(1983)]{Tsei83a}
	G.~S. Tseitin.
	\newblock On the complexity of derivation in propositional calculus.
	\newblock In \emph{Automation of Reasoning}, pages 466--483. Springer, 1983.
	
	\bibitem[Woeginger(2003)]{Woeg03a}
	G.~J. Woeginger.
	\newblock Banks winners in tournaments are difficult to recognize.
	\newblock \emph{Social Choice and Welfare}, 20\penalty0 (3):\penalty0 523--528,
	2003.
	
	\bibitem[Yannakakis(1982)]{Yann82a}
	M.~Yannakakis.
	\newblock The complexity of the partial order dimension problem.
	\newblock \emph{SIAM Journal on Algebraic and Discrete Methods}, 3\penalty0
	(3):\penalty0 351--358, 1982.
	
	\bibitem[Young and Levenglick(1978)]{YoLe78a}
	H.~P. Young and A.~Levenglick.
	\newblock A consistent extension of {C}ondorcet's election principle.
	\newblock \emph{SIAM Journal on Applied Mathematics}, 35\penalty0 (2):\penalty0
	285--300, 1978.
	
\end{thebibliography}



\end{document}